\documentclass{elsarticle}
\journal{Signal Processing}

\usepackage{amsmath,amssymb,graphicx,epsfig,algorithm,algorithmic,epstopdf, url}
\usepackage[utf8]{inputenc}
\usepackage[mathscr]{euscript}
\usepackage{hyperref}
\usepackage{color} 
\usepackage{amsthm}
\usepackage{flushend}
\usepackage{bbm}
\usepackage{booktabs}
\usepackage{xcolor}
\usepackage{algorithm}
\usepackage{algorithmic}
\usepackage{bm}

\newcommand{\bM}{\mathbf{M}}

\def\minwrt[#1]{\underset{#1}{\text{minimize }}}
\def\argminwrt[#1]{\underset{#1}{\text{arg min }}}
\def\maxemphwrt[#1]{\underset{#1}{\text{\emph{maximize} }}}
\newcommand{\ett}{{\bf 1}}

\DeclareMathOperator*{\maxwrt}{maximize}

\newcommand{\mR}{{\mathbb R}}

\newcommand{\diag}{{\rm diag}}

\newtheorem{theorem}{Theorem}
\newtheorem{remark}{Remark}
\newtheorem{proposition}{Proposition}
\newtheorem{lemma}{Lemma}
\newtheorem{example}{Example}

\DeclareMathOperator{\expm}{expm}

\newcommand{\norm}[1]{\left\lVert#1\right\rVert}
\newcommand{\abs}[1]{\left|#1\right|}
\newcommand{\trace}[1]{\text{tr}\left(#1\right)}

\def\bx{{\bf x}}

\def\bC{{\bf C}}

\def\bK{{\bf K}}

\def\bU{{\bf U}}

\def\ccA{{\mathcal{A}}}

\def\ccM{{\mathcal{M}}}

\def\ccX{{\mathcal{X}}}
\def\ccY{{\mathcal{Y}}}

\def\imagunit{{\mathfrak{i}}}

\def\discstate{{\mathrm{x}}}

\def\mongemap{{g}}

\newcommand{\cD}{{\mathcal D}}

\newcommand{\cT}{{\mathcal T}}
\newcommand{\bcX}{{\bm{\mathcal X}}}
\newcommand{\cX}{{\mathcal{X}}}

\def\RC{{\mathbb{C}}}
\def\RE{{\mathbb{E}}}

\def\RM{{\mathbb{M}}}
\def\RN{{\mathbb{N}}}
\def\RR{{\mathbb{R}}}
\def\RT{{\mathbb{T}}}
\def\RZ{{\mathbb{Z}}}

\newcommand{\realpart}{\mathfrak{Re}}
\newcommand{\imagpart}{\mathfrak{Im}}

\newcommand{\Omtspec}{S}

\newcommand{\costfunc}{c}
\newcommand{\costfuncmulti}{\bm{\mathcal{C}}}

\newcommand{\disccovop}{G}

\newcommand{\wavelength}{\xi}
\newcommand{\projop}[2]{\mathcal{P}_{#1}^{#2}}
\newcommand{\projopdisc}[2]{P_{#1}^{#2}} 

        \makeatletter
        \def\fps@eqnfloat{!t}
        \def\ftype@eqnfloat{4}
        
        \newenvironment{eqnfloat*}
               {\@dblfloat{eqnfloat}}
               {\end@dblfloat}
        \makeatother
\renewenvironment{thebibliography}[1]{%
 \begin{oldthebibliography}{#1}%
 \setlength{\parskip}{0ex}%
 \setlength{\itemsep}{0ex}%
}%
{%
\end{oldthebibliography}%
}%

\begin{document}
\begin{frontmatter}


\title{Multi-Marginal Optimal Mass Transport with Partial Information\tnoteref{t1}}
\tnotetext[t1]{This work was supported in part by eSSENCE grant 2017-4:2, Carl Trygger's foundation grant 16:208, and the Swedish science foundation grants 2015-04148 and 2014-5870.}
\author[LU]{Filip Elvander\corref{cor1}} 
\ead{filip.elvander@matstat.lu.se}

\author[KTH]{Isabel Haasler}
\ead{haasler@kth.se}

\author[LU]{Andreas Jakobsson}
\ead{aj@maths.lth.se}

\author[KTH]{Johan Karlsson}
\ead{johan.karlsson@math.kth.se}

\address[LU]{Department of Mathematical Statistics, Lund University, P.O. Box 118, SE-221 00 Lund, Sweden}
\address[KTH]{Department of Mathematics, KTH Royal Institute of Technology}
\cortext[cor1]{Corresponding author. Phone: +46722004363.}

\begin{abstract}
During recent decades, there has been a substantial development in optimal mass transport theory and methods.
In this work, we consider multi-marginal problems wherein only partial information of each marginal is available, which is a setup common in many inverse problems in, e.g., imaging and spectral estimation.
By considering an entropy regularized approximation of the original transport problem, we propose an algorithm corresponding to a block-coordinate ascent of the dual problem, where Newton's algorithm is used to solve the sub-problems. 
In order to make this computationally tractable for large-scale settings, we utilize the tensor structure that arises in practical problems,
allowing for computing projections of the multi-marginal transport plan using only matrix-vector operations of relatively small matrices. As illustrating examples, we apply the resulting method to tracking and barycenter problems in spatial spectral estimation.
In particular, we show that the optimal mass transport framework allows for fusing information from different time steps, as well as from different sensor arrays, also when the sensor arrays are not jointly calibrated. Furthermore, we show that by incorporating knowledge of underlying dynamics in tracking scenarios, one may arrive at accurate spectral estimates, as well as faithful reconstructions of spectra corresponding to unobserved time points. 
\end{abstract}
\end{frontmatter}

\section{Introduction}

The use of mass, or power, distributions for modeling and describing properties of signals and stochastic processes is ubiquitous in the field of signal processing, being fundamental to areas such as spectral estimation, classification, and localization \cite{StoicaM05}. Spectral representations, detailing the distribution of power over frequency, serve as compact signal descriptions and constitute the foundation for many applications, such as speech enhancement \cite{HuL03_11}, non-invasive estimation \cite{AminiEG05_52}, spectroscopy \cite{MierisovaA01}, as well as radar and sonar \cite{trees1992detection}. In such applications, one requires a means of comparing and quantifying distances between spectra. Popular choices include the Kullback-Leibler \cite{KullbackL51_22} and Itakura-Saito divergences \cite{GrayBGM80_28}, the R\'enyi entropy \cite{Renyi61}, and $L_p$-norms, where the latter is often implicitly utilized through the use of matrix norms applied to the signal covariance matrix \cite{StoicaBL11_59b, SwardAJ18_143}.
As an alternative to these measures, it has also been proposed to define distances between spectra based on the concept of optimal mass transport \cite{GeorgiouKT09_57}. Optimal mass transport (OMT) is concerned with the task of moving one distribution of mass into another, with the description of the movement of mass being referred to as a transport plan.
The minimal total cost of rearrangement, i.e., the one associated with the optimal plan, may then be used as a measure of distance between the two distributions \cite{Villani08}.
Historically, OMT has been widely used in economics, e.g., for planning and logistics (see \cite{Villani08} for an introduction and an overview of the topic).
Recently, there has been a rapid development in the theory and methods for optimal mass transport, and the ideas have attracted considerable attention in various economic and engineering fields, being used, for example, for option pricing \cite{BeiglbockHLP13_17}, color and texture transfer in image processing \cite{DominitzT10_16}, as well as machine learning \cite{SchmitzHBNCCPS18_11,ArjovskyCB17_arXiv} and other signal processing tasks \cite{KolouriPTSR17_34}.
These developments have lead to a mature framework for OMT, with computationally efficient algorithms \cite{Cuturi13, peyre2019computational}, and constitute a flexible modeling tool that may be used to address many problems in the areas of signal processing and systems theory.

One of the appeals of using OMT for defining distances between distributions lies in its geometric properties. Specifically, as it models transport between two mass distributions, taking place on their domain of definition, the distance is not only related to the point-wise differences between the distributions, but also depends on where the mass is located. As a consequence of this, OMT distances directly reflect proximity of mass on the underlying space \cite{GeorgiouKT09_57}. This property becomes very attractive when used in array processing and spatial spectral estimation problems, in which the support of the spectrum directly details the location of signal emitting sources \cite{StoicaM05}. For example, OMT has been shown to produce robust and predictable results when applied to direction of arrival (DoA) and localization problems \cite{ElvanderHJK18_eusipco,ElvanderHJK19_icassp}. 
Also, when computing the distance between two mass distributions, the OMT problem naturally induces a way of interpolating distributions by considering the flow of mass on the underlying space \cite{mccann1997convexity, AngenentHT03_35}. This procedure has been shown to yield physically meaningful results, e.g., when used for modeling heat diffusion on graphs \cite{SimouF19_icassp}, as well as for interpolating and tracking speech/sound signals \cite{JiangLG12_60, ning2015matrix} and (structured) covariance matrices \cite{NingJG13_20, ElvanderJK18_66, yamamoto2018regularization}.

An important generalization, not least for applications in signals and systems, is the introduction of dynamics to the OMT formulation, thereby allowing for modeling, e.g., inertia. Such a model was proposed and developed in \cite{ChenGP17_62}, wherein the mass distributions were modeled as ensembles of mass particles endowed with linear time-varying dynamics specified by a state space model.
Building on the fluid dynamics formulation in \cite{benamou2000computational}, this may also be viewed as an optimal control problem of an ensemble of particles that is required to reach a given final distribution \cite{chen2016relation}.
The cost of transport then corresponds to the control energy required for deviating from the expected trajectory as to reach the specified state. 
This theory has also been used for tracking and estimation of ensembles where only access to certain output states are available \cite{ChenK18_2}.
However, the increased flexibility of the dynamical setup comes at the price of higher computational complexity of the resulting OMT problem, as the domain is multi-dimensional.
Furthermore, the OMT formulation has recently been extended as to consider transport between several distributions, referred to as the multi-marginal OMT problem \cite{pass2015multi}. Several problems such as tracking, barycenter,  and flow problems can be seen as special cases of multi-marginal OMT by selecting the cost function in a suitable manner \cite{peyre2019computational}. However, also here the complexity increases with the number of marginals.

The optimal mass transport problem has many useful properties but it also has drawbacks. 
The original formulation of Gaspard Monge is a non-convex optimization problem, and even though the formulation by Leonid Kantorovich is convex, it results in very large optimization problems that are intractable with standard methods, even for medium-sized transportation problems. Recently, this computational issues have been addressed by regularizing the problem using an entropic barrier term. The resulting optimization problem may be solved using the so called Sinkhorn iterations \cite{Cuturi13, Sinkhorn67_74}, allowing for finding approximate solutions to large transportation problems in a computationally efficient manner, which
has proved useful for many cases where no computationally feasible method previously existed.
However, both multi-dimensional and multi-marginal problems suffer from the curse of dimensionality, and the computational complexity remains a challenge even for the Sinkhorn iteration scheme, although it has been shown that for the problem of computing the Euler flow, computationally efficient solutions may be found by utilizing the structure of the cost function \cite{benamou2015iterative}. 
As we will see in this paper, several other cases of multi-marginal OMT problems can also be addressed in a computationally efficient manner by exploiting cost function structures.

In this work, we consider multi-marginal OMT problems where the distributions themselves are not directly observed, but manifested through linear measurement equations. We refer to this problem as multi-marginal OMT with partial information. Several problems of interest may be expressed within such a framework, e.g., CT imaging (cf. \cite{benamou2015iterative,KarlssonR17_10}), ensemble estimation \cite{ChenK18_2}, image deblurring \cite{lellmann2014imaging}, spectral analysis \cite{JiangLG12_60}, and radar imaging problems \cite{ElvanderHJK19_icassp}.  
As motivating examples, we consider problems from radar imaging, being spatial spectral estimation problem, wherein the mass distributions, i.e., power spectra, are observable from covariance matrices corresponding to a set of sensor arrays. For such problems, typically only the co-located sensors, i.e., sensors in an array, can be processed coherently due to calibration errors between the sensor arrays. Herein, we propose an OMT framework for fusing the information provided by the array covariances in a non-coherent manner. For such scenarios, we demonstrate that multi-marginal OMT constitutes a flexible and robust tool for performing information fusion, illustrating the inherent robustness to spectral perturbations, as caused by, e.g., calibration errors, provided by the use of transport models.
In addition to this, we show that the proposed OMT framework may be used to fuse information corresponding to different time points, allowing for modeling the time evolution of spatial spectra. In this case, we demonstrate how the framework allows for incorporating prior knowledge of underlying dynamics of the signal generating mechanism, providing a means of forming accurate reconstructions of target trajectories in spatial estimation problems.

In addition to the modeling framework, one of the main contributions of this work is deriving a computationally efficient method for solving multi-marginal OMT problems with partial information. As the original OMT formulation is posed in an infinite-dimensional function space, we formulate discrete approximations of this problem, allowing for practical implementations. Specifically, one of the main contributions of this work is proposing a generalization of the Sinkhorn iterations, which may be applied to OMT problems with partial information.
In particular, we consider the corresponding dual problem and propose to solve this with a block-coordinate ascent method, wherein each block update is computed using Newton's method. Furthermore, for certain multi-marginal OMT problems, such as the tracking and barycenter problems, we show that the structure of the cost function may be exploited, allowing for solving high-dimensional problems with a large number of marginals.
It should be stressed that even though spatial spectral estimation is here used as a motivating and illustrating example, the proposed framework applies to a large class of multi-marginal OMT problem where only partial information of the marginal distributions is available.

This work is based on several of our previous work on tracking and information fusion. In particular, the original tracking formulation is from \cite{ChenK18_2}. This was further developed for radar problems \cite{ElvanderHJK18_eusipco} as well as computationally efficient computation using pairwise entropy regularization \cite{ElvanderHJK19_icassp}. In this paper, by extending this work to the  multi-marginal setting, we obtain a more flexible framework, as well as more numerically stable methods. 

This paper is structured as follows. In Section~\ref{sec:motivating_examples}, we present the motivating example of spatial spectral estimation, describing the relation between spectra and covariance matrices. Section~\ref{sec:background} then provides a background to the problem of optimal mass transport. In Section~\ref{sec:omt_multi_margin}, we present the multi-marginal OMT problem with partial information, and describe the use of this formulation in spectral estimation. Section~\ref{sec:entropy_reg_comp} presents computational tools necessary for formulating an efficient solution algorithm for the multi-dimensional OMT problem, which is then extended to problems with partial information in Section~\ref{sec:discrete_omt_PI}. In Section~\ref{sec:numerical_results}, we provide numerical examples, illustrating the properties of the proposed framework, whereas Section~\ref{sec:conclusions} concludes upon the work.

%
%
\section{Notation}\label{sec:notation}
Let $\RM^p$ denote the space of Hermitian matrices of size $p\times p$, and let $\bar{(\cdot)}$, $(\cdot)^T$ and $(\cdot)^H$ denote complex conjugation, transpose, and conjugate transpose, respectively. For vectors $x \in \RC^n$, let $\norm{x}_p$ denote the standard $\ell_p$-norm, i.e.,
%
	$\norm{x}_p^p = \sum_{k=1}^n \abs{x_k}^p $
%
for $p\geq 1$. Furthermore, for a space $\ccX\subset \RR^d$, let $\ccM(\ccX)$ be the set of generalized integrable functions on $\ccX$ (i.e., signed measures on $\ccX$), with $\ccM_+(\ccX)$ denoting the subset of these that are non-negative (i.e., non-negative measures on $\ccX$).  

Unless stated otherwise, vectors, matrices, and uni- and bivariate functions are denoted by italic letters, whereas boldface letters are reserved for multivariate functions and tensors. The notation ./, $\odot$, $\exp$, and $\log$ denote element-wise division, multiplication, exponential function, and logarithm, respectively. Finally, for a square matrix $A$, the matrix exponential is denoted $\expm(A)$.
%
%
%
%
%
\section{Motivating examples}\label{sec:motivating_examples}

The modeling tools proposed in this work may be applied to several different classes of reconstruction and estimation problems, such as, e.g., inverse problems encountered in imaging, radar, tomography, and signal processing. However, in order to give a concise description, we will herein focus on a few motivating applications in spectral estimation and radar imaging. 
In a typical reconstruction problem, one seeks to recover a spectrum based on a number of measurements, e.g., estimates of the covariance function of a time series, or estimates of the covariance matrices corresponding to a set of sensor arrays.
However, in practice, the time series may not be perfectly stationary, and the set of sensors may not be perfectly calibrated. Such modeling errors may, when combining all available measurements in forming an estimate of the spectrum, result in artifacts and/or lack of feasible solutions.
In addition, spectral reconstruction problems are typically ill-posed, meaning that solutions, if they exist, are not unique, or that small measurement or calibration errors may result in large errors in the reconstruction of the spectrum. A main goal of this work is to develop a robust modeling framework for addressing such reconstruction problems, allowing for fusing the available measurements in an informed an dependable manner. A key tool for this will be the use of OMT.

In this section, we describe the basic mathematical model relating the covariance estimates to the power spectra and the problem of localization. In the subsequent sections, we will develop theory and formulate OMT problems which may be used to fuse covariance measurements obtained from different sensors or from different time points. 
%
\subsection{Spectral estimation} \label{sub_sec:spec}
Generally, the power spectrum of a time series represents the distribution of power over frequency.  
More precisely, consider a complex-valued discrete-time stochastic process, $y(t)$, for $t \in \RZ$, which we will assume to be zero mean and wide sense stationary (WSS), i.e., $\RE (y(t))=0$ for all $t\in \RZ$, and with covariance
\begin{equation} \label{eq:r_k}
r_k \triangleq \RE (y(t)\overline{y(t-k)})
\end{equation}
being independent of $t$. Here, $\RE(\cdot)$ denotes the expectation operator. The frequency content of the process $y(t)$ is described by the power spectrum, $\Phi\in \ccM_+(\RT)$, which is the non-negative measure on $\RT=(-\pi,\pi]$ whose Fourier coefficients coincide with the covariances (see, e.g.,  \cite[Chapter 2]{JohnsonD93})
\begin{equation}\label{eq:moments}
r_k=\frac{1}{2\pi}\int_{-\pi}^\pi \Phi(\theta)e^{-\imagunit k\theta} d\theta
\end{equation}
for $k\in \RZ$, where $\imagunit = \sqrt{-1}$ is the imaginary unit. Commonly, one considers the inverse problem of reconstructing the power spectrum $\Phi$ from a set of given covariances, $r_k$, for $k \in \RZ,$ with $\abs{k}\leq p-1$ for some $p \in \RN$, i.e., one seeks a function $\Phi$ that is consistent with the covariance sequence $\left\{ r_k \right\}_{\abs{k}\leq p-1}$, such that \eqref{eq:moments} holds for $\abs{k}\leq p-1$. The corresponding $p\times p$ covariance matrix is then given by
\begin{align*}
R=\begin{pmatrix}
r_0&r_{-1}& r_{-2}&\cdots &r_{-p+1}\\
r_1&r_0& r_{-1}&\cdots &r_{-p+2}\\
r_2&r_1& r_{0}&\cdots &r_{-p+3}\\
\vdots &\vdots& \vdots&\ddots&\vdots\\
r_{p-1}& r_{p-2}&r_{p-3}&\cdots& r_{0}
\end{pmatrix}
\end{align*}
which is a Hermitian Toeplitz matrix, since $y(t)$ is WSS. Thus, expressed using matrices, a spectrum is consistent with a partial covariance sequence, or, equivalently, a finite covariance matrix, if $\Gamma(\Phi) = R$,
where \mbox{$\Gamma:\ccM(\RT) \to \RM^{p}$} is the linear operator 
\begin{align}\label{eq:gamma1}
	\Gamma(\Phi)\triangleq\frac{1}{2\pi}\int_{\RT} a(\theta)\Phi(\theta) a(\theta)^Hd\theta
\end{align}
where
\begin{align*}
	a(\theta) \triangleq \left[\begin{array}{cccc} 1 & e^{\imagunit\theta} & \cdots &  e^{\imagunit(p-1)\theta} \end{array} \right]^T
\end{align*}
denotes the Fourier vector. Note that since $\Phi$ is non-negative, it follows from \eqref{eq:gamma1} that  $R=\Gamma(\Phi)$ is positive semi-definite. Furthermore, if the Toeplitz matrix $R$ is positive semi-definite, then there exists at least one non-negative power spectrum $\Phi$ such that $R=\Gamma(\Phi)$. If $R$ is positive definite, there exists an infinite family of solutions $\Phi$ \cite{GrenanderS58}. For parameterizations of specific families of such spectra, e.g., rational and exponential spectra, see \cite{ByrnesGL00_48, Georgiou05_50, FerrantePR08_53}.

In most practical cases, the stochastic process cannot be assumed to be stationary over a longer time interval. Thus, the covariances generally have to be estimated based on shorter intervals, resulting in possibly inaccurate and low resolution spectral estimates.
In this work, we will see how OMT may be utilized for modeling slowly time-varying spectra, thereby allowing for improved spectral reconstructions from noisy covariance measurements.
%
\subsection{Direction of arrival estimation and localization} \label{sub_sec:doa_loc}

The inverse problem of estimating a power spectrum from a finite set of temporal covariances, as described above, may be extended to DoA and localization problems by noting that the spatial covariance matrix corresponding to a sensor array may be parametrized analogously to equation \eqref{eq:moments}. Specifically, consider a scenario in which point-like sources located in a space $\ccX \subset \RR^d$ emit waves impinging on a set of $p \in \RN$ sensors located at $x_k \in \RR^d$, with $d \leq 3$. Letting the corresponding sensor signal vector be denoted $y \in \RC^p$, one may consider the sensor covariance matrix, often called the spatial covariance matrix,
\begin{align*}
	R = \RE\left( y y^H \right).
\end{align*}
Then, just as in \eqref{eq:gamma1}, the covariance $R$ may be related to a positive measure on $\ccX$, denoted $\Phi \in \ccM_+(\ccX)$, i.e., $R = \Gamma(\Phi)$, where $\Gamma: \ccM(\ccX) \to \RM^{p}$ is given by
\begin{align} \label{eq:gamma_operator}
	\Gamma(\Phi) = \int_\ccX a(x)\Phi(x)a(x)^H dx,
\end{align}
where $dx$ denotes the Lebesgue measure on $\ccX$ \cite{Georgiou05_50}. 
Here, $a: \ccX \to \RC^p$ is, generally, no longer a Fourier vector, as in the temporal case, but instead constitutes an array steering vector, i.e., the manifold vector. The specific functional form of $a$ depends on the sensor locations, the wavelength, as well as on the propagation properties of this space. For example, the array manifold vector corresponding to spherical wave fronts in $\ccX \subset \RR^d$ may be defined as \cite{JohnsonD93}
\begin{equation}
a(x)=\begin{pmatrix}\frac{1}{\norm{x_k-x}_2^{(d-1)/2} }e^{-{2\pi
\imagunit}\!\frac{\norm{x_k-x}_2}{\wavelength}} \end{pmatrix}_{k=1}^p,
\label{eq:array_mainifold_vector}
\end{equation}
where $\norm{\cdot}_2$ denotes the Euclidean norm, and $\wavelength$ is the source signal wavelength. Here, $\Phi$ is referred to as a spatial spectrum, as it details the distribution of mass, or power, on $\ccX$, with locations of the signal sources on $\ccX$ corresponding to the support of $\Phi$. In practice, signal sources are often identified with peaks of power, i.e., the modes of $\Phi$, which is, for instance, utilized in detection methods such as MUSIC \cite{Schmidt79} and ESPRIT \cite{RoyPK86_34}. It may be noted that the identifiability of the spatial spectral estimation problem depends on the functional form of the array manifold vector $a$. For example, the problem of localization requires non-zero curvature of the impinging wave fronts, as expressed in \eqref{eq:array_mainifold_vector}, modeling spherical waves.
It may also be noted that the array geometry, i.e., the relative location of the sensors, determine the identifiability of source localization and DoA problems. For example, when using a uniform linear array (ULA), the geometry of the array gives rise to spatial aliasing such that source locations are only identifiable when confined to one of the half-spaces defined by the array. It should be stressed that even for array geometries that do not give rise to spatial aliasing,%
\footnote{ That is, the steering vector is such that $\tilde{x} \neq x$ implies $a(\tilde{x}) \neq a(x)$ for all $x,\tilde{x} \in \ccX$.} the operator $\Gamma$ in \eqref{eq:gamma_operator} is not injective and hence there are several distinct $\Phi \in \ccM_+(\ccX)$ satisfying $\Gamma(\Phi) = R$ for any given $R$ in the range of $\Gamma$.

In this work, we will use the problem of localization, seen as a spectral estimation problem, as an illustrative example of how to utilize the proposed multi-marginal OMT framework for fusing information. Specifically, we will consider the problems of sensor fusion and spectral tracking. Sensor fusion corresponds to combining measurements obtained from separate sensor arrays in forming a single estimate of the spectral spectrum, whereas tracking refers to combining measurements collected at consecutive time points in order to reconstruct the time evolution of a spatial spectrum.
%

\section{Background: Optimal mass transport}\label{sec:background}
In this section, we provide a brief background on the problem of OMT, together with extensions to settings with underlying dynamics and incomplete information.
\subsection{Monge's problem and the Kantorovich relaxation} \label{subsec:monge_kantorovich}
The problem of OMT, as formulated by Monge, is concerned with finding a mapping between two distributions of mass such that a measure of cost, interpreted as the cost of transportation, is minimized. That is, given two mass distributions $\Phi_0\in\ccM_+(\ccX_0)$ and  $\Phi_1 \in \ccM_+(\ccX_1)$, defined on the spaces $\ccX_0$ and $\ccX_1$, respectively, one seeks a function $\mongemap: \ccX_0 \to \ccX_1$ minimizing the functional \cite{Villani08}
\begin{align} \label{eq:monge_omt}
	\int_{\ccX_0} \costfunc\left(x_0, \mongemap(x_0)\right) \Phi_0(x_0)dx_0,
\end{align}
subject to the constraint
\begin{align} \label{eq:push_forward_for_mapping0}
	\int_{x_1 \in \mathcal{U}} \Phi_1(x_1)dx_1 = \int_{\mongemap(x_0) \in \mathcal{U}} \Phi(x_0) dx_0, \mbox{ for all } \mathcal{U} \subset \ccX_1.
\end{align}
Here, $\costfunc: \ccX_0 \times \ccX_1 \to \RR_+$ denotes the cost of transporting one unit of mass from points in $\ccX_0$ to points in $\ccX_1$, yielding that \eqref{eq:monge_omt} equals the total cost of transporting the distribution $\Phi_0$ to $\Phi_1$ for a given mapping $\mongemap$. 
The constraint in \eqref{eq:push_forward_for_mapping0} enforces the condition that $\mongemap$ transports $\Phi_0$ to $\Phi_1$, i.e., the mapping $\mongemap$ moves the mass in $\Phi_0$ onto $\Phi_1$ without any loss or addition of mass in the transition. This may be expressed using the shorthand notation $\mongemap_\# \Phi_0 = \Phi_1$, where $\mongemap_\# \Phi_0$ denotes the so-called push-forward measure, i.e.,
\begin{align} \label{eq:push_forward_for_mapping}
	\int_{x_1 \in \mathcal{U}} \left(\mongemap_\#\Phi_0\right)(x_1)dx_1 = \int_{\mongemap(x_0) \in \mathcal{U}} \Phi(x_0) dx_0,
\end{align}
for all $\mathcal{U} \subset \ccX_1$. The notation thus implies that for any subset $\mathcal{U}$ of $\ccX_1$, the mass of $\mongemap_\#\Phi_0$ on $\mathcal{U}$ is equal to the mass of $\Phi_0$ on the pre-image of $\mathcal{U}$. It directly follows that the total mass of $\mongemap_\#\Phi_0$ equals that of $\Phi_0$.
Finding such a $\mongemap$ that minimizes \eqref{eq:monge_omt}, in general, constitutes a difficult problem. Also, the use of a mapping for modeling transport imposes some additional requirements on $\Phi_0$ and $\Phi_1$; for example, the problem in \eqref{eq:monge_omt} does not allow for splitting up Dirac delta functions in $\Phi_0$ \cite{Villani08}, and the problem of minimizing \eqref{eq:monge_omt} subject to $\Phi_1=\mongemap_\#\Phi_0$ may lack an optimal solution. In contrast, the Kantorovich relaxation of \eqref{eq:monge_omt} circumvents these restrictions by instead of a mapping $\mongemap$ seeking an optimal coupling between the distributions $\Phi_0$ and $\Phi_1$.
The set of valid couplings is the set of measures on the product space $ \ccX_0\times \ccX_1$ with marginals coinciding with $\Phi_0$ and $\Phi_1$, i.e.,
\begin{align} \label{eq:marginal_constraint}
	\Omega(\Phi_0,\Phi_1) = \left\{ M \in \ccM_+(\ccX_0\times \ccX_1) \mid \int_{\ccX_1} M(\cdot,x_1) dx_1 = \Phi_0, \int_{\ccX_0} M(x_0,\cdot) dx_0 = \Phi_1 \right\}.
\end{align}
Here, if $M \in \Omega(\Phi_0,\Phi_1)$, then $M$ is referred to as a transport plan, as the function value $M(x_0,x_1)$ may be interpreted as the amount of mass that is transported between $x_0$ and $x_1$. Thus, the marginal constraint  \eqref{eq:marginal_constraint} replaces the condition $\mongemap_\#\Phi_0 = \Phi_1$ used in the Monge formulation. It may be noted that $\Omega(\Phi_0,\Phi_1)$ is always non-empty; the normalized outer product of $\Phi_0(x_0)$ and $\Phi_1(x_1)$ is an element of $\Omega(\Phi_0,\Phi_1)$. It should also be noted that it is assumed that the mass of $\Phi_0$ and $\Phi_1$ are the same, i.e., $\int_{\ccX_0} \Phi_0(x_0) dx_0 = \int_{\ccX_1} \Phi_1(x_1) dx_1$. The Monge-Kantorovich problem of optimal mass transport may then be stated as
\begin{equation}\label{eq:omt_standard}
\begin{aligned}
	 \minwrt[{M \in \Omega(\Phi_0,\Phi_1) }]&\quad \int_{\cX_0\times \cX_1} M(x_0,x_1)\costfunc(x_0,x_1) dx_0 dx_1,
\end{aligned}
\end{equation}
i.e., it seeks the transport plan $M$ minimizing the cost of transportation between $\Phi_0$ and $\Phi_1$, where the point-wise cost of moving a unit mass from $x_0\in\ccX_0$ to $x_1\in\ccX_1$ is given by $\costfunc(x_0,x_1)$.
In contrast with \eqref{eq:monge_omt}, the objective function in  \eqref{eq:omt_standard} is linear in $M$ and the optimization problem always has an optimal solution. Further, for absolutely continuous $\Phi_0$ and $\Phi_1$, the problems \eqref{eq:monge_omt} and \eqref{eq:omt_standard} are equivalent \cite{Villani08}.

The resulting minimal objective value of \eqref{eq:omt_standard} may be used in order to define a notion of distance $S: \ccM_+(\cX_0)\times \ccM_+(\cX_1) \to \RR$ according to
\begin{equation} \label{eq:omt_stationary_spectra}
\begin{aligned}
	\Omtspec(\Phi_0,\Phi_1) = \min_{M \in \Omega(\Phi_0,\Phi_1) }&\quad \int_{\cX_0\times \cX_1} M(x_0,x_1)\costfunc(x_0,x_1) dx_0 dx_1.
\end{aligned}
\end{equation}
Recently, the idea of using an OMT-induced concept of proximity has attracted considerable interest in a plethora of modeling applications, such as serving as the learning criterion in machine learning scenarios \cite{AdlerROK17_arxiv,ArjovskyCB17_arXiv}, for performing texture and color transfer \cite{DominitzT10_16}, dictionary learning \cite{SchmitzHBNCCPS18_11}, as well as being used as an evaluation metric \cite{AguilaJ18_66}. Also, it has been used for inducing metric structure on the space of power spectra \cite{GeorgiouKT09_57}, where it was shown that choosing $\costfunc$ as a metric on $\ccX=\cX_0=\cX_1$ in general results in $\Omtspec$ being metric on $\ccM_+(\ccX)$.
Similar efforts have been directed towards defining distances for Toeplitz covariance matrices \cite{ElvanderJK18_66}, allowing for formulating a spectral estimation framework with inherent robustness \cite{ElvanderHJK18_eusipco,ElvanderHJK19_icassp}.
One of the main motivations of using OMT-based distances such as \eqref{eq:omt_stationary_spectra} is its ability to reflect the geometric properties of the underlying space, $\ccX$. As a consequence of this, using OMT distances in order to perform interpolation-like tasks, such as, e.g., morphing of images \cite{KolouriPTSR17_34}, or, recently, for describing the time evolution of a signal on a graph \cite{SimouF19_icassp}, typically yields meaningful and interpretable signal reconstructions. This should be contrasted with the standard Euclidean metric that, when used in the same applications, has been shown to give rise to fade-in fade-out effects, thereby failing to respect underlying assumptions for the signal generating mechanisms, such as, e.g., finite movement speed \cite{SimouF19_icassp, ElvanderJK18_66,GeorgiouKT09_57}.
%
\subsection{Optimal mass transport with dynamics} \label{sub_sec:omt_dyn}
It may be noted that the specific choice of cost function $\costfunc$ used in the OMT objective in \eqref{eq:omt_standard} may greatly influence the structure of the solution, i.e., the transport plan $M$. Preferably, $\costfunc$ should be chosen as to reflect properties or limitations of the considered application, e.g., in order to ensure expected smooth trajectories of transported mass that may be dictated by physical considerations. In many applications, the $\ell_p$ norms, i.e., $c(x_0,x_1) = \|x_0-x_1\|_p^p$ constitute a popular choice, giving rise to the Wasserstein$-p$ distances, with the special case of $p = 1$ often being referred to as the Earth Mover's distance \cite{Villani08}. Among the appeals of this choice are the metric properties of the Wasserstein$-p$ distances \cite{Villani08}, the availability of closed form solutions in one-dimensional problems via the cumulative distribution function (see, e.g., \cite{HakerZTA04_60}), and, more recently, computationally efficient approximations allowing for their use in high-dimensional machine learning applications \cite{BonneelRPP15_51}. However, inherent in this class of cost functions is the tacit assumption that the distribution of mass should remain stationary, i.e., it is expected not to change. Although reasonable, and even desirable, in, e.g., classification applications, we will herein consider a generalization of these distances in order to allow for incorporating knowledge of the dynamics of the process generating the signal observations \cite{ChenGP17_62}.

To this end, let the underlying space be $\ccX \subset \RR^d$, and consider a mass particle with time-varying location, or state, $x(t) \in \ccX$. Furthermore, assume that the dynamics of the particle may be well modeled by the (continuous time) linear differential equation
\begin{equation} \label{eq:dyn_eq}
\begin{aligned}
	\dot{x}(t) = A(t)x(t) + B(t)u(t),
\end{aligned}
\end{equation}
where $x(t) \in \ccX$, $A(t) \in \RR^{d\times d}$, $B(t) \in \RR^{d\times f }$, and where $u(t) \in \RR^f$ denotes an input signal, steering the evolution of the state $x(t)$. 
As in \cite{ChenGP17_62}, we will herein consider this type of model in order to define a cost function $\costfunc: \ccX\times\ccX \to \RR_+$, describing the cost of transport for a unit mass between two points of $\ccX$, corresponding to the time instances $t = t_0$ and $t = t_1$. For notational simplicity, and without loss of generality\footnote{It may be noted that the here presented formulation may be modified in a straightforward manner to other time values $t_0$ and $t_1$.}, we here use $t_0 = 0$ and $t_1 = 1$. Specifically, we define the cost for transporting mass between $x_0 = x(0)$ to $x_1 = x(1)$ via a quadratic optimal control problem, according to
\begin{equation} \label{eq:quad_costfunc}
\begin{aligned}
	c(x_0,x_1) = &\quad\min_{u}\qquad \int_0^1 \norm{u(t)}_2^2 dt \\
	&\text{subject to } \quad\dot{x}(t) = A(t)x(t) + B(t)u(t) \\
	&\qquad\qquad\qquad x(0) = x_0 \;,\; x(1) = x_1,
\end{aligned}
\end{equation}
i.e., the cost is equal to the minimal input signal energy required to steer the state from $x_0$ at $t = 0$ to $x_1$ at $t = 1$. For this case, the cost may be expressed in closed form as
\begin{align} \label{eq:dynamic_cost}
	c(x_0,x_1) = (x_1 - \ccA(1,0) x_0)^TW(1,0)^{-1}(x_1-\ccA(1,0) x_0),
\end{align}
where $\ccA(t,\tau)$ is the state transition matrix and
\begin{align*}
	W(t,s) =\int_s^t \ccA(t,\tau)B(\tau)B(\tau)^T\ccA(t,\tau)^T d\tau
\end{align*}
is the controllability Gramian. It may be noted that for the case of a constant $A(t) \equiv A$, the state transition matrix is given by the matrix exponential $\ccA(t,\tau) = \expm\left(A(t-\tau)\right)$.
It is further worth observing that the transportation cost is not necessarily zero when the initial and end states are the same, but instead when the end state coincides with the expected evolution of the initial state, in accordance with the state equation in \eqref{eq:dyn_eq} with $u\equiv 0$. Also, the cost function is not, in general, symmetric in its arguments. Thus, as the cost of transport is dictated by the system dynamics, the optimal transport plan is the coupling $M \in \ccM_+(\ccX\times \ccX)$ that provides the mass association most in accordance with the trajectories implied by \eqref{eq:dyn_eq}.
It may be noted that the standard stationary case may be retrieved using $A \equiv 0$ and $B \equiv I$, where $I$ is the identity matrix, in which case the cost is equal to the squared Euclidean distance, i.e., $\costfunc(x_0,x_1) = \norm{x_0 - x_1}_2^2$.
Furthermore, having access to a transport plan, $M$, that minimizes the OMT criterion with the dynamic cost in \eqref{eq:quad_costfunc} directly allows for computing intermediate mass distributions interpolating the given marginals using the system dynamics in \eqref{eq:dyn_eq}; this is termed a displacement interpolation \cite{mccann1997convexity}. Specifically, for a particle $x(t)$ such that $x(0) = x_0$ and $x(1) = x_1$, the optimal input signal is given by \cite[equation (24)]{ChenGP17_62}
\begin{align*}
	\hat{u}(t;x_0,x_1) = B(t)^T \ccA(1,t)^T W(1,0)^{-1} \left( x_1 - \ccA(1,0)x_0  \right),
\end{align*}
implying that the corresponding trajectory is
\begin{align*}
	\hat{x}(t;x_0,x_1) = \ccA(t,1) W(1,t)W(1,0)^{-1} \ccA(1,0)x_0 + W(t,0) \ccA(1,t)^TW(1,0)^{-1}x_1,
\end{align*}
obtained by solving \eqref{eq:dyn_eq} using the input signal $\hat{u}(t;x_0,x_1)$. Then, one may construct an interpolating mass distribution, or spectrum, $\Phi_\tau^{M}$, corresponding to time $\tau \in [0,1]$ and parametrized by the transport plan $M$, according to
\begin{align*}
	\Phi_\tau^{M}(x) = \int_{\ccX \times \ccX} M(x_0,x_1) \delta\left(x-\hat{x}(\tau;x_0,x_1)\right)dx_0 dx_1,
\end{align*}
where $\delta$ denotes the Dirac delta function. That is, $\Phi_\tau^{M}(x)$ details the mass located at $x$ at time $\tau$ by considering the mass of all pairs $(x_0,x_1)$ such that their trajectories contains the point $(x,\tau)$. Also, this framework directly allows for defining extrapolating spectra, i.e., distributions of mass for $\tau > 1$, e.g., by letting the input signal be identically zero for $\tau >1$, implying that the unforced system dynamics are used to describe the time evolution of the state, and thus the spectrum. It may be noted that the interpolation for Toeplitz covariance matrices presented in \cite{ElvanderJK18_66} may be obtained as a special case of this formulation, using $A(t) \equiv 0$ and $B(t) \equiv I$.
%

\subsection{Partial information}
As may be noted in the OMT formulations in Section~\ref{subsec:monge_kantorovich}, the input data, i.e., the measurements, used for computing distances and obtaining transport plans, are the mass distributions themselves, i.e., $\Phi \in \ccM_+(\ccX)$. However, in many problem formulations these distributions are not only unknown, but are in fact the quantities of interest to be estimated. In such applications, the measurements instead correspond to functions of the spectrum $\Phi \in \ccM_+(\ccX)$, making the task of inferring it an inverse problem. As noted earlier, a common example of this is the availability of covariance matrices, or estimates thereof, that are linear mappings of power spectra, as may be described using the operator $\Gamma$ in \eqref{eq:gamma_operator}. Thus, the measurements convey only partial information of the underlying distributions, or spectra. Recently, this connection was used in order to define a notion of distance on the space of Toeplitz covariance matrices based on OMT problems in the spectral domain \cite{ElvanderJK18_66}. Specifically, the distance between two covariance matrices $R_0$ and $R_1$ was defined as the minimal objective value of
\begin{equation} \label{eq:omt_cov_mat}
\begin{aligned}
\minwrt[M \in \ccM_+(\ccX\times \ccX)]&\qquad \int_{\ccX\times \ccX} M(x_0,x_1)c(x_0,x_1) dx_0 dx_1 \\
	\text{subject to }&\quad\Gamma\left( \int_\ccX M(\cdot,x_1) dx_1\right) = R_0 \\
	&\quad \Gamma\left( \int_\ccX M(x_0,\cdot) dx_0\right)  = R_1.
\end{aligned}
\end{equation}
Note that this is equal to the smallest possible OMT distance between any two spectra $(\Phi_0,\Phi_1)$ consistent with the observed covariance matrices, i.e., that satisfies $\Gamma(\Phi_t) = R_t$, for $t = 0,1$. As was shown in \cite{ElvanderJK18_66}, the distance notion in \eqref{eq:omt_cov_mat} is in general not a metric, but is rather a semi-metric on the space of covariance matrices for any semi-metric cost function $\costfunc$.

In this work, we will limit our attention to the typical case when one does not have direct access to the mass distributions, or spectra, but instead to (estimates of) covariance matrices. Also, in the context of using dynamical models such as \eqref{eq:dyn_eq}, the partial information contained in the measurements may only be related to a state vector in a space of lower dimension than $\ccX$. Specifically, the system description in \eqref{eq:dyn_eq} may be extended with an observation equation according to
\begin{equation} \label{eq:system_eq}
\begin{aligned}
	\dot{x}(t) &= A(t)x(t) + B(t)u(t) \\
	y(t) &= F(t) x(t),
\end{aligned}
\end{equation}
implying that covariances are related to spectra detailing mass distributions on a space $\ccY$, potentially different from $\ccX$, where $y(t) \in \ccY$ is the observed state. Using the system in \eqref{eq:system_eq}, one may thus incorporate prior knowledge of underlying dynamics, expressed through the state equation, even though one only has access to partial measurements of an observed state. For the corresponding spectra, it holds that if $\Phi \in \ccM_+(\ccX)$, then $F(t)_\# \Phi \in \ccM_+(\ccY)$, where $F(t)_\#$ denotes the push-forward operation, as defined in \eqref{eq:push_forward_for_mapping}, i.e.,
\begin{align*}
	\int_{x \in \mathcal{U}} \left(F(t)_\#\Phi\right)(x)dx = \int_{F(t)x \in \mathcal{U}} \Phi(x) dx,
\end{align*}
for any $\mathcal{U} \subset \ccY$. 
Note that, using the system in \eqref{eq:system_eq}, covariance measurements corresponding to observable quantities are related to spectra on the space $\ccY$, i.e., to elements of $\ccM_+(\ccY)$, and only have an indirect link to spectra on the space $\ccX$. However, the push-forward operation $F(t)_\#$ allows for expressing the connection between covariance matrices and elements of $\ccM_+(\ccX)$ by modifying the operators $\Gamma$ to incorporate $F(t)_\#$ according to
\begin{align}\label{eq:gamma_push_forward}
	\Gamma_{t}\left(\Phi \right) = \int_\ccY a(y) \left( F(t)_\# \Phi \right)(y) a(y)^H dy.
\end{align}
It is worth noting that the domain of definition of the linear operator $\Gamma_t$ is still $\ccM(\ccX)$, whereas the domain of definition for the array steering vectors is $\ccY$, reflecting that the covariance measurements are related to elements of $\ccM_+(\ccY)$. Correspondingly, one may formulate an OMT problem with partial information as
\begin{equation} \label{eq:omt_with_dyn_observation}
\begin{aligned}
\minwrt[M \in \ccM_+(\ccX\times \ccX)]&\qquad \int_{\ccX\times \ccX} M(x_0,x_1)c(x_0,x_1) dx_0 dx_1 \\
	\text{subject to }&\quad \Gamma_{0}\left(\int_\ccX M(\cdot,x_1) dx_1 \right) = R_0 \\
	&\quad \Gamma_{1}\left(\int_\ccX M(x_0,\cdot) dx_0 \right) = R_1.
\end{aligned}
\end{equation}
It may here be noted that the transport plan $M$ is still an element of $\ccM_+(\ccX\times \ccX)$, allowing for modeling dynamical structures on the space generating data, i.e., $\ccX$, that are not present in the measurement space $\ccY$. 
Thus, in this context, the notion of partial information is twofold. Firstly, mass distributions are not observed directly, but only through linear measurements on a considerably smaller dimension, i.e., covariances in a finite-dimensional space as opposed to the infinite-dimensional space of measures for spectra. Secondly, these spectra are mappings of spectra that are elements of an even bigger space. We will in the next section see how these ideas, together with the concept of multi-marginal transport plans, can be used in order to arrive at compact formulations for spectral estimation.

\section{Multi-marginal optimal mass transport}\label{sec:omt_multi_margin}
It is worth noting that the formulation in \eqref{eq:omt_standard} seeks a transport plan connecting two mass distributions, or marginals, $\Phi_0$ and $\Phi_1$, which may be interpreted as a description of how to morph $\Phi_0$ into $\Phi_1$. A natural extension of this is to consider transport between a larger set of marginals, i.e., a set $\left\{ \Phi_\cT  \right\}_{t=0}^{\cT}$ for $\cT \geq 1$ \cite{pass2015multi}.
To this end, let $\bM \in \ccM_+(\bcX)$, where $\bcX=\cX_0\times \cX_1 \times \cdots \times \cX_\cT$, and define the set of projection operators $\projop{t}{}: \ccM(\bcX) \to \ccM(\ccX_t)$, for $t = 0,\ldots,\cT$, as
\begin{align} \label{eq:proj_operator}
	\projop{t}{}(\bM) = \int_{\ccX_0\times \cdots \times \ccX_{t-1}\times \ccX_{t+1}\times \cdots \times \ccX_\cT}\bM(x_0,\ldots,x_t,\ldots,x_\cT)dx_0\ldots dx_{t-1}dx_{t+1}\ldots dx_{\cT}.
\end{align}
It may be noted that one may express marginalization also for the standard two-marginal case using the operator $\projop{t}{}$. Then, $\bM \in \ccM_+(\bcX)$ is referred to as a multi-marginal transport plan for the set $\left\{ \Phi_t  \right\}_{t=0}^{\cT}$ if $\projop{t}{}(\bM) = \Phi_t$, for $t = 0,\ldots,\cT$. Note here that $\bM$ provides a complete description of the association of mass for the set $\left\{ \Phi_t  \right\}_{t=0}^{\cT}$, i.e., $\bM(x_0,x_1,\ldots,x_\cT)$ denotes the amount of mass at $x_0$, corresponding to the marginal $\Phi_0$, that is transported to location $x_t$, corresponding to marginal $\Phi_t$, for $t = 1,2,\ldots,\cT$. Correspondingly, one may define a cost function $\costfuncmulti: \bcX \to \RR_+$, yielding a generalization of the OMT problem as (see also \cite{pass2015multi})
\begin{equation} \label{eq:multi_margin_omt}
\begin{aligned}
	\minwrt[\bM \in \ccM_+(\bcX)]&\quad \int_{\bcX} \bM(\bx) \costfuncmulti(\bx) d\bx\\
	\text{subject to }&\quad  \projop{t}{} (\bM)  = \Phi_t, \; \quad \mbox{ for } t = 0,\ldots,\cT,
\end{aligned}
\end{equation}
where $\bx=(x_0,\ldots,x_\cT)$ and $d\bx=dx_0 dx_1\ldots dx_{\cT}$.
Thus, the minimizer of \eqref{eq:multi_margin_omt} is the transport plan associated with the minimal cost of transportation for the set $\left\{ \Phi_t  \right\}_{t=0}^{\cT}$. This type of formulation has previously found application in, e.g., tomography \cite{AbrahamABC17_75} and fluid dynamics \cite{Brenier08_237}, with actual numerical approximations and implementations being considered in \cite{benamou2015iterative}.
In this work, we consider extending this problem to scenarios in which the marginals $\left\{ \Phi_t  \right\}_{t=0}^{\cT}$ are not directly observable. Instead, we assume that one has access to partial information in the form of linear mappings of the marginals. In the context of spectral estimation, this corresponds to the temporal or spatial covariance matrix, or estimates thereof, i.e., a sequence $\left\{ R_t  \right\}_{t=0}^{\cT}$, where $R_t \in \RM^{p_t}$ is related to a marginal $\Phi_t \in \ccM_+(\ccX_t)$ through the measurement equation $R_t = \Gamma_t(\Phi_t)$. It should be noted that the operators $\Gamma_t: \ccM(\ccX_t) \to \RM^{p_t}$ may, in general, be different for different indices $t$. With this, the problem considered herein may be stated as
\begin{equation} \label{eq:multi_margin_omt_PI}
\begin{aligned}
	\minwrt[\substack{\bM \in \ccM_+(\bcX)\\ \Delta_t \in \RC^{p_t\times p_t}}]&\quad \int_{\bcX} \bM(\bx) \costfuncmulti(\bx) d\bx+\sum_{t=0}^\cT \gamma_t\|\Delta_t\|_{\rm F}^2\\
	\text{subject to }&\quad  \Gamma_t(\projop{t}{} (\bM))  = R_t+\Delta_t \; \quad \mbox{ for } t = 0,\ldots,\cT,
\end{aligned}
\end{equation}
where $\|\cdot \|_{\rm F}$ denotes the Frobenius norm and where $\gamma_t>0$, for $t = 0,1,\ldots,\cT$, are user-specified weights. 
Note that the measurement equations have been augmented by error terms $\Delta_t \in \RC^{p_t\times p_t}$ in order to allow for noisy covariance matrix estimates $R_t$. The parameters $\gamma_t$ then allow for 
determining a trade-off between how much consideration should be taken to the transport cost and the measurement error, respectively. 
For simplicity, we here penalize $\Delta_t$ by the squared Frobenius norm. However, if prior knowledge of the structure of the measurement noise is available, this may be replaced by other norms or penalty functions.
Also, in order to compute couplings between marginals corresponding to index pairs $(t,s)$, one requires the corresponding bi-marginal transport plan. Due to the structure of the multi-marginal transport plan $\bM$, this is directly computable by means of the bi-marginal projection operator $\projop{t,s}{}: \ccM_+(\bcX) \to \ccM_+(\ccX_t \times \ccX_s)$ defined by

\begin{align} \label{eq:proj_operator_bi_marginal}
	\projop{t,s}{}(\bM) = \int_{\ccX_0\times \cdots \times \ccX_{s-1}\times \ccX_{s+1} \times \cdots \times \ccX_{t-1}\times \ccX_{t+1}\times \cdots \times \ccX_\cT}\bM(x_0,\ldots,x_\cT)dx_0\ldots dx_{t-1}dx_{t+1}\ldots dx_{s-1}dx_{s+1}\ldots dx_{\cT}.
\end{align}
This may then be used, e.g., for interpolation between adjacent marginals for which the cost describes a dynamic model as in Section~\ref{sub_sec:omt_dyn}.
In Section~\ref{subsec:decoupling_costfuncs}, we present computationally efficient methods for computing discrete approximations of the projections in \eqref{eq:proj_operator} and \eqref{eq:proj_operator_bi_marginal}. In fact, computing the marginal projections in \eqref{eq:proj_operator} will constitute an integral component in the algorithm for solving \eqref{eq:multi_margin_omt_PI}, which will be presented in Section~\ref{sec:discrete_omt_PI}.

Next, we will see that a number of spectral estimation problems may be cast in the form \eqref{eq:multi_margin_omt_PI}. The first, referred to as sensor fusion, is concerned with combining measurements, in the form of covariance matrices, from several separate sensor arrays in order to, e.g., resolve ambiguities resulting from the geometry of individual arrays. It is here assumed that only one covariance matrix is available for each sensor array, and that there is no meaningful difference between the times at which these are collected. In contrast to this, in the second scenario, referred to as tracking, a sequence of covariance matrices is observed, and the problem of interest is to reconstruct a corresponding sequence of spatial spectra. In this setting, the dynamical OMT formulation described in Section~\ref{sub_sec:omt_dyn} serves as the building block for defining cost functions that represent a certain dynamics model. The last problem, referred to as barycenter tracking, aims to combine the former two approaches in order to form a robust reconstruction of a spectral trajectory. 
As will be clear in the following, the only thing distinguishing these spectral estimation problem is the choice of cost function $\costfuncmulti$. That is, depending on the structure of $\costfuncmulti$, the general problem in \eqref{eq:multi_margin_omt_PI} may be interpreted as a sensor fusion, tracking, or barycenter tracking problem.
All estimation problems in this section are formulated as convex programming problems. However, as they are defined on the space of measures, their dimensions are infinite. For ease of notation, we will in this description often assume that the domains of the marginals are the same, however, this generalizes straightforwardly to general domains.
In Sections~\ref{sec:entropy_reg_comp} and \ref{sec:discrete_omt_PI}, we address the problem of formulating practically implementable problems using discretized approximations, and present computationally efficient solvers.
%

\subsection{Sensor fusion} \label{sub_sec:sensor_fusion}
Consider a scenario in which waves impinge on a set of $J \in \RN$ sensor arrays, with corresponding array manifold vectors $a_j: \ccX\to \RC^{p_j}$ and operators $\Gamma_j: \ccM(\ccX) \to \RM^{p_j}$, for $j = 1,\ldots,J$, where it may be noted that the number of sensors, $p_j$, may differ among the different arrays. Then, a spatial spectrum $\Phi \in \ccM_+(\ccX)$ gives rise to a set of covariance matrices $R_j$, for $j = 1,\ldots,J$, where $R_j  = \Gamma_j(\Phi)$. Thus, as all sensor array measurements have been generated by the same signal, the problem of estimating the spatial spectrum $\Phi$, and thereby the location of the signal sources, corresponds to the inverse problem of finding $\Phi$ such that $R_j = \Gamma_j(\Phi)$, for $j = 1,\ldots,J$. However, in practice, calibration errors in the sensor arrays, i.e., errors in the operators $\Gamma_j$, may result in there existing no single $\Phi \in \ccM_+(\ccX)$ such that it is consistent with all array measurements. For cases in which the joint covariance matrix for the full set of sensors, or the signal itself, is available, this may be addressed using calibration methods such as \cite{PesaventoGW02_50,SeeG04_52}. However, in some scenarios, this information may not be available, due to, e.g., the covariance matrix for each separate array being estimated locally and then transmitted to a fusion center \cite{Kaplan06_42,LiWHS02_19}, with bandwidth limitations and/or synchronization errors between the separate arrays preventing the estimation of the joint covariance matrix, as well as direct processing of the signal itself \cite{Kaplan06_42}. Having only access to the separate covariance matrices, we herein propose to address this issue using a spectral barycenter formulation. Specifically, we propose to solve
\begin{equation} \label{eq:pair_wise_barycenter}
\begin{aligned}
	\minwrt[\substack{\Phi_j \in \ccM_+(\ccX)\\ \Delta_j \in \RC^{p_j\times p_j}}]&\quad \sum_{j=1}^J \Omtspec(\Phi_0, \Phi_j) + \gamma_j \norm{\Delta_j}_F^2\\
	\text{subject to }&\quad \Gamma_j(\Phi_j)  = R_j + \Delta_j, \; j = 1,\ldots,J, \\
\end{aligned}
\end{equation}
and estimate the spatial spectrum as the minimizer $\Phi_0$. Here, $\Omtspec$ is the OMT distance defined in \eqref{eq:omt_stationary_spectra}.
Note here that $\Phi_j$, for $j = 1,\ldots,J$, are spectra consistent with the set of observations $R_j$, whereas $\Phi_0$ is the spectrum closest, in the OMT sense, to this set of spectra, i.e., $\Phi_0$ is the spectral barycenter. The motivation for using the OMT distance $\Omtspec$, as opposed to, e.g., the standard $L_p$-norms, in forming the barycenter is that $\Omtspec$ explicitly takes the geometry of the underlying space, i.e., $\ccX$, into account, meaning that small values of $\Omtspec$ correspond to spectra that are close, in the sense of the peaks being close\footnote{In the localization example, this thus implies that the sources are physically close.}, as quantified by the cost function $\costfunc$. Choosing $\costfunc$ to be an intrinsic distance measure on $\ccX$, e.g., $c(x_0,x_1) = \norm{x_0-x_1}_2^2$, thus results in low values of $\Omtspec$ implying actual proximity, and vice versa. This also induces robustness to the estimate $\Phi$; small spatial perturbations imply small perturbations of the distance measure. In contrast, using $L_p$-norms as divergences on the space of spectra are inherently non-robust to such perturbations, as spatial perturbations are equivalent to changing the support of the corresponding spatial spectrum.

One may formulate the problem in \eqref{eq:pair_wise_barycenter} as a multi-marginal problem on the form \eqref{eq:multi_margin_omt_PI} by specifying the multi-marginal cost function $\costfuncmulti$ as
\begin{align}\label{eq:costfunc_central}
	\costfuncmulti(\bx) = \sum_{j=1}^J \costfunc(x_0,x_j),
\end{align}
where $\costfunc: \ccX \times \ccX \to \RR_+$ is the pair-wise cost function in the definition of $\Omtspec$. 
With this, the barycenter problem may be stated as
\begin{equation} \label{eq:multi_margin_omt_barycenter}
\begin{aligned}
	\minwrt[\substack{\bM \in \ccM_+(\bcX)\\ \Delta_j \in \RC^{p_j\times p_j}}]&\quad \int_{\bcX} \bM(\bx) \left( \sum_{j=1}^J \costfunc(x_0,x_j) \right) d\bx+\sum_{j=1}^J \gamma_j\|\Delta_j\|_{\rm F}^2\\
	\text{subject to }&\quad  \Gamma_j(\projop{j}{} (\bM))  = R_j+\Delta_j \; \quad \mbox{ for } j = 1,\ldots,J.
\end{aligned}
\end{equation}
To see that \eqref{eq:costfunc_central} indeed induces a barycenter solution when used in \eqref{eq:multi_margin_omt_barycenter}, note that the cost function $\costfuncmulti$ is structured as to penalize transport between a central marginal, corresponding to the index $j = 0$, and all other marginals, corresponding to indices $j = 1,\ldots,J$. As only marginals with indices $j = 1,\ldots,J$ correspond to actual measurements, this implies that the marginal $t = 0$, which is not constrained to be related to any measurements, should be close, in the OMT sense, to the entire ensemble of marginals.
Note here that the barycenter does not appear as an explicit variable, but is instead given by the projection of $\bM$ on the zeroth marginal, i.e., $\Phi = \projop{0}{} (\bM)$. 
It may be noted that, in general, one may utilize different OMT distances $S_j$ for different arrays, e.g., one may replace the sums in \eqref{eq:pair_wise_barycenter} and \eqref{eq:costfunc_central} with a weighted sum to get a weighted barycenter.
%
%
%
\subsection{Tracking}\label{sub_sec:tracking} 
The previous section introduced a method for fusing information from different sensor arrays as to produce an estimate of a single spatial spectrum. Formulated in this way, the time of the measurements does not play any role, i.e., the spectrum may either be considered to be a stationary estimate, or to correspond to a snapshot in time. In contrast to this, one may also consider a scenario in which a single sensor array\footnote{We present a formulation for handling multi-array cases in the next section.} is used to observe the evolution of a signal over time. 
Specifically, assume that estimates of the array covariance matrix are obtained at $\cT+1$ time instances corresponding to the nominal times $t = 0,1,\ldots,\cT$, giving rise to the sequence $\left\{ R_t \right\}_{t=0}^{\cT}$. 
Having access to an estimate of the sequence $\left\{ R_t \right\}_{t=0}^{\cT}$, one may naturally proceed to estimate a corresponding sequence of spectra, $\left\{ \Phi_t \right\}_{t=0}^{\cT}$, by considering each covariance matrix $R_t$ in isolation, thereby arriving at $\cT+1$ separate estimation problems. However, this approach ignores the quite reasonable assumption that the signal sources move continuously and hence, spatial spectra spaced closely in time should be similar. Arguably, using OMT in order to quantify similarity is a reasonable choice, as this divergence is directly related to proximity on the underlying space, distinguishing it from $L_p$-norms that lack this property. 
It should also be stressed that the OMT formulation does not require any \textit{a priori} knowledge of the number of signal sources, or the transmitted signal waveforms.
Also, as was shown in \cite{ElvanderJK18_66}, OMT formulations provide a powerful tool for finding smooth spectral trajectories interpolating noisy measurements. This tracking problem may be formulated as
\begin{equation}\label{eq:multi_fuse_spec}
\begin{aligned}
	\minwrt[\substack{\Phi_t \in \ccM_+(\ccX)\\ \Delta_t \in \RC^{p\times p}}]&\quad \sum_{t = 1}^TS_t(\Phi_{t-1},\Phi_t) + \gamma_t \norm{\Delta_t}_F^2\\
	\text{subject to }&\quad \Gamma_t\left(\Phi_t \right) = R_t + \Delta_t, \; t = 0,\ldots,\cT,
\end{aligned}
\end{equation}
i.e., the sequence $\left\{ \Phi_t \right\}_{t=0}^\cT$ minimizes a sequential OMT distance, while interpolating the (de-noised) measurement sequence $\left\{ R_t \right\}_{t=0}^\cT$. 
It may here be noted that the OMT distance, $\Omtspec_t$, is indexed by $t$, as to express that the underlying cost of transport may be time-varying.
Correspondingly, a multi-marginal formulation may be obtained by considering a cost function $\costfuncmulti$ such that
\begin{align} \label{eq:costfunc_sequential}
	\costfuncmulti(\bx) = \sum_{t=1}^\cT \costfunc_t(x_{t-1},x_t),
\end{align}
with the index $t$ expressing time-dependence, yielding
\begin{equation} \label{eq:multi_margin_omt_tracking}
\begin{aligned}
	\minwrt[\substack{\bM \in \ccM_+(\bcX)\\ \Delta_t \in \RC^{p\times p}}]&\quad \int_{\bcX} \bM(\bx) \left( \sum_{t=1}^\cT \costfunc_t(x_{t-1},x_t) \right) d\bx+\sum_{t=0}^\cT \gamma_t\|\Delta_t\|_{\rm F}^2\\
	\text{subject to }&\quad  \Gamma_t(\projop{t}{} (\bM))  = R_t+\Delta_t \; \quad \mbox{ for } t = 0,\ldots,\cT.
\end{aligned}
\end{equation}
Here, the estimated spectral sequence is given by the marginals of $\bM$, i.e., $\Phi_t = \projop{t}{}(\bM)$, for $t = 0,1,\ldots,\cT$. 
It is worth noting that the cost function $\costfuncmulti$ penalizes transport between consecutive marginals, i.e., between pairs corresponding to indices $(t-1,t)$, for $t = 1,\ldots,\cT$. When used in \eqref{eq:multi_margin_omt_tracking}, this implies that the underlying spatial spectrum is not expected to change rapidly over time.
Further, this formulation allows for exploiting assumptions of an underlying dynamic model, expressed in $\costfunc_t$ in \eqref{eq:costfunc_sequential}, as detailed in Section~\ref{sub_sec:omt_dyn}. Specifically, let the space on which the spatial spectra are defined be denoted $\ccY$, and let $\ccX$ correspond to a larger space, allowing for expressing the system's dynamics.
For example, in physically realistic applications, it is reasonable to assume that not only the movement speed of signal sources is finite, but also the acceleration. In this setting, one may, e.g., let the space $\ccX$ detail all pairs of location and velocity, whereas the measurement space $\ccY$ only details feasible locations (see also Example~\ref{example:angle_and_speed} below). In general, if the system description in \eqref{eq:system_eq} is used, with the observed state corresponding to the spatial spectrum, one may extend the operators $\Gamma_t$ to also include the push-forward operation, as detailed in \eqref{eq:gamma_push_forward}.
Here, the system dynamics are encoded into the tracking problem \eqref{eq:multi_margin_omt_tracking} by choosing $\costfunc_t$ as the quadratic cost in \eqref{eq:dynamic_cost}, as defined by the state space equation in \eqref{eq:dyn_eq}. The sequence of spatial spectra may then be obtained as $\Phi_t = F(t)_\# \projop{t}{}(\bM)$. We will in Section~\ref{sec:numerical_results} examine the implications of using a dynamical model as opposed to a static one when performing OMT-based spectral tracking.
\begin{example} \label{example:angle_and_speed}
Consider a DoA estimation scenario in $\RR^2$, in which a sequence of array covariance matrices is available, where one in addition to the direction of arrival, $\theta \in [-\pi,\pi)$ aims to model also the velocity of the targets, in order to enforce smooth trajectories. Letting $v(t) \in \RR$ denote the angular speed, i.e., $\dot{\theta}(t) = v(t)$, one may thus select $\ccY = [-\pi,\pi)$ and $\ccX = [-\pi,\pi) \times \RR$. In order to impose the assumption of smooth trajectories, one may then enforce finite acceleration by, in the state space representation in \eqref{eq:system_eq}, using the matrices
\begin{align*}
	A = \begin{bmatrix}
		0 & 1 \\ 0 & 0
	\end{bmatrix} \;,\; B = \left[\begin{array}{cc} 0 & 1  \end{array}\right]^T \;,\; F = \left[\begin{array}{cc} 1 & 0  \end{array}\right].
\end{align*}
It may be noted that the pair $(A,B)$ enforces that the angular position can only be influenced through its angular velocity, i.e., through acceleration, whereas the matrix $F$ expresses the fact that only the angular position is manifested in the array covariance matrix. The corresponding relation between spectra on $\ccX$ and $\ccY$, denoted $\Phi_t^\ccX$ and $\Phi_t^\ccY$, respectively, may then be expressed as
\begin{align*}
	\Phi^\ccY_t(\theta) = \left(F_\# \Phi^\ccX_t \right)(\theta) = \int_\RR \Phi_t^\ccX(\theta,v)dv,
\end{align*}
for $\theta \in [0,2\pi)$. The cost function $\costfunc$ implied by the LQ control signal $u$ is then given by (see, e.g., \cite{ChenK18_2})
\begin{align*}
	\costfunc(x_0,x_1) = \left(x_1 - \expm(A) x_0 \right)^TW(1,0)^{-1} \left(x_1 - \expm(A)x_0 \right)
\end{align*}
with
\begin{align*}
	\expm(A(t-\tau)) = \begin{bmatrix}
		1 & t-\tau \\ 0 & 1
	\end{bmatrix} \;,\;
	W(1,0) = \begin{bmatrix}
		1/3 & 1/2 \\ 1/2 & 1
	\end{bmatrix}.
\end{align*}
\end{example}\hfill$\qed$
%
%
%
%
%
%
%
%
%
%
\subsection{Combining tracking and barycenters} \label{subsec:barycenter_tracking}
In addition to the problems described in Sections~\ref{sub_sec:sensor_fusion} and \ref{sub_sec:tracking}, i.e., that of sensor fusion and spectral tracking, respectively, one may consider a combination of the two, i.e., the tracking of a spectral barycenter over time. This then allows for addressing issues of resolving spatial ambiguities, as well as promoting smooth spectral trajectories over time as to, e.g., induce robustness to noisy measurements. Therefore, consider $J$ sensor arrays observing a signal during $\cT+1$ time instances, giving rise to a set of covariance matrices $R_t^{(j)}$, for $j = 1,\ldots,J$ and $t = 0,\ldots,\cT$, where the subscript indicates time, and the superscript corresponds to the array index. Thus, we seek to estimate a sequence of spectral barycenters, $\left\{ \Phi_t \right\}_{t=0}^\cT$, given covariance matrix estimates $R_t^{(j)}$, for $t = 0,1,\ldots,\cT$ and $j = 1,2,\ldots,J$. In order to formulate the barycenter tracking problem, one may consider two separate types of transport costs; $\Omtspec_t$ detailing the OMT distance between consecutive barycenters, and $\tilde{\Omtspec}$ describing the distances between barycenters and spectra consistent with the array measurements. For example, the distance $\Omtspec_t$ may reflect underlying dynamics, e.g., expected movement, thereby allowing for incorporating assumptions of the time evolution of the spectrum. In contrast, $\tilde{\Omtspec}$ may correspond to a static OMT problem, corresponding to the assumption that the array covariance matrices are generated by the same spectrum. Practically, this may be implemented by using different costs of transport, $\costfunc$, for $\Omtspec_t$ and $\tilde{\Omtspec}$. 
Correspondingly, one may let the barycenter spectra be elements of $\ccM_+(\ccX)$, where the space $\ccX$ may detail the system dynamics, whereas the consistent spectra are elements of $\ccM_+(\ccY)$, where $\ccY$ is the measurement space.
With this, the barycenter tracking problem may be formulated as
\begin{equation} \label{eq:barycenter_tracking}
\begin{aligned}
	\minwrt[\substack{\Phi_t \in \ccM_+(\ccX)\\ \Phi_t^{(j)} \in \ccM_+(\ccY)\\ \Delta_t^{(j)} \in \RC^{p_j \times p_j}}]&\quad \sum_{t = 1}^{\cT}\Omtspec_t\left(\Phi_{t-1},\Phi_{t}\right) + \alpha \sum_{t=0}^\cT\sum_{j=1}^J \left(\tilde{\Omtspec}\left(\Phi_t^{(j)},F(t)_\#\Phi_t\right) +\gamma^{(j)} \norm{\Delta_t^{(j)}}_F^2\right) \\
	\text{subject to }&\quad \Gamma^{(j)}\left(\Phi_t^{(j)} \right) = R_t^{(j)} + \Delta_t^{(j)}, \;  \mbox{ for } t = 0,\ldots,\cT \text{ and } j = 1,\ldots,J,
\end{aligned}
\end{equation}
where $\alpha>0$ is a user-specified constant allowing for deciding a trade-off between the smoothness of the tracking and the distance between a barycenter $\Phi_t$ and its corresponding consistent spectra $\Phi_t^{(j)}$, for $t = 0,1,\ldots,\cT$ and $j = 1,2,\ldots,J$.
It may here be noted the latter distance is stated in terms of the push-forward of the barycenter spectrum, thereby allowing for some states to be measured (e.g., position) whereas others are not (e.g., velocity). Also, the operators $\Gamma^{(j)}$ are indexed by $j$, reflecting differences in geometry and location between the different arrays.
In order to arrive at a multi-marginal formulation on the form \eqref{eq:multi_margin_omt_PI}, one may define the cost function $\costfuncmulti$ as
\begin{align} \label{eq:cost_tensor_barycenter_tracking}
	\costfuncmulti(\bx) = \sum_{t=1}^{\cT} \costfunc_t(x_{t-1,0},x_{t,0}) + \alpha\sum_{t=0}^\cT \sum_{j=1}^J \tilde{\costfunc}(x_{t0},x_{tj})
\end{align}
where $\bx \in \bcX^J = \bcX \times \bcX \times \ldots \times \bcX$, which is the direct product of $\cT+1$ instances of $\bcX$, each corresponding to a barycenter problem at one time point. That is, this space is a product of the form 
$\bcX= \cX \times \ccY \times \dots \times \ccY$, with one instance of $\cX$ and $J$ instances of $ \ccY$.
Here, $\bx = (\bx_{0},\bx_{1},\ldots,\bx_\cT)$, where
\begin{align*}
	\bx_t = (x_{t0},x_{t1},\ldots,x_{tJ}),
\end{align*}
where the first index indicate the time $t$, for $t = 0,1,\ldots,\cT$, and the second index indicates the array $j$, for $j = 0,1,\ldots,J$, with 0 corresponding to the barycenter. Here, $\costfunc_t$ denotes the cost of transport for the tracking, whereas $\tilde{\costfunc}$ defines the cost of transportation between the barycenters and the marginal spectra. With this setup, one may formulate the multi-marginal problem as
\begin{equation} \label{eq:multi_margin_omt_barycenter_tracking}
\begin{aligned}
	\minwrt[\substack{\bM \in \ccM_+(\bcX^J)\\ \Delta_t^{(j)} \in \RC^{p_j\times p_j}}]&\quad \int_{\bcX^J} \bM(\bx) \costfuncmulti(\bx)d\bx+\alpha \sum_{t=0}^\cT \sum_{j=1}^J\gamma_t^{(j)}\|\Delta_t^{(j)}\|_{\rm F}^2\\
	\text{subject to }&\quad  \Gamma^{(j)}(\projop{tj}{} (\bM))  = R_t^{(j)}+\Delta_t^{(j)} \; \quad \mbox{ for } t = 0,\ldots,\cT \text{ and } j = 1,\ldots,J.
\end{aligned}
\end{equation}
Here, the operator $\projop{tj}{}$ projects onto the marginal $(t,j)$ of $\bM$ in direct analogy with \eqref{eq:proj_operator}, which for each time $t$ correspond to a spectrum consistent with the measurements of array $j$, for $j = 1,2,\ldots,J$. The corresponding sequence of barycenters may then be constructed as $\Phi_t = \projop{t0}{} (\bM)$, for $t = 0,1,\ldots,\cT$. It may be noted that the problem in \eqref{eq:multi_margin_omt_barycenter_tracking}, through the use of the cost function in \eqref{eq:cost_tensor_barycenter_tracking}, is a direct hybrid of the sensor fusion and the tracking problems from Sections~\ref{sub_sec:sensor_fusion} and \ref{sub_sec:tracking}.

Having shown that the multi-marginal OMT problem with partial information in \eqref{eq:multi_margin_omt_PI} may be used to formulate these spectral estimators, we proceed to present the main results of this work, i.e., how to construct discrete approximations of \eqref{eq:multi_margin_omt_PI} that allow for computationally efficient solvers. Illustrative numerical examples of the behavior of \eqref{eq:multi_margin_omt_PI} in spatial spectral estimation scenarios are then provided in Section~\ref{sec:numerical_results}.
%
\section{Entropy regularization for structured multimarginal OMT problems}\label{sec:entropy_reg_comp}

The multi-marginal OMT problem with partial information in \eqref{eq:multi_margin_omt_PI} is an infinite-dimensional problem, as the multi-marginal transport plan $\bM$ is an element of the function space $\ccM_+(\bcX)$. Fortunately, one may in practical implementations approximate \eqref{eq:multi_margin_omt_PI} by a discrete counterpart. In this section, we present the necessary building blocks for deriving a computationally efficient solution algorithm for such a problem. Starting by describing discretization of OMT problems with full information, i.e., problems of the form \eqref{eq:omt_standard}, followed by multi-marginal generalizations, we show that the critical component of the resulting practical algorithm is the possibility to efficiently compute the projections in \eqref{eq:proj_operator}. Section~\ref{sec:discrete_omt_PI} then presents an algorithm for solving the discrete counterpart of \eqref{eq:multi_margin_omt_PI}.
\subsection{Sinkhorn iterations for the bi-marginal optimal mass transport problem} 
In order to discretize the OMT problem in \eqref{eq:omt_standard}, one may construct a grid consisting of $n$ grid points $X = \left\{ \discstate_1, \discstate_2, \ldots, \discstate_n \right\}$ on the support of the marginals. This allows for representing the cost function $\costfunc$ on $X$ by a matrix $C\in\RR^{n\times n}$, where the element $c_{ij} = \costfunc(\discstate_i,\discstate_j)$ denotes the cost of transporting a unit mass from point $\discstate_i$ to $\discstate_j$. Correspondingly, the transport plan may be represented by a matrix $M\in\RR_+^{n\times n}$, whose elements $m_{ij}$ denote the amount of mass that is transported from $\discstate_i$ to $\discstate_j$. The Kantorovich optimal mass transport problem in discretized form then reads
\begin{equation}
\begin{aligned}
\minwrt[M \in \RR_+^{n\times n}]  \quad & \trace{C^T M} \label{eq:omt_discrete}\\
  \mbox{ subject to } \quad &\Phi_0=M \ett \\
  &\Phi_1=M^T \ett,
\end{aligned}
\end{equation}
where $\trace{\cdot}$ denotes the trace operation, and where $\ett \in \RR^n$ is a vector of all ones. It may be noted that the two marginal distributions are represented by $\Phi_0, \Phi_1 \in \RR_+^n$. Here, in direct correspondence with \eqref{eq:omt_standard}, the objective function $\trace{C^TM} = \sum_{i,j} c_{ij} m_{ij}$ gives the total cost of transport as implied by $M$, whereas the marginal constraints translate to the requirement that the row and column sums of $M$ equal $\Phi_0$ and $\Phi_1$, respectively.
Although \eqref{eq:omt_discrete} is a linear program, finding an optimal transport plan can be computationally cumbersome due to the large number of variables. In particular, a grid of size $n$ yields an OMT problem with $n^2$ variables. It may be noted that standard linear programming methods require $\mathcal{O}(n^3\log(n))$ operations to find the optimum for such a problem \cite{PeleW09}, proving infeasible even for moderate grid sizes.

One way to bypass the expensive computations inherent in finding the exact transport plan is to regularize the OMT problem, yielding an approximation of \eqref{eq:omt_discrete}. To this end, it was first suggested in \cite{Cuturi13} to introduce an entropy term to the objective function. In this setting, the problem in \eqref{eq:omt_discrete} is modified to
\begin{equation}
\begin{aligned} 
\minwrt[M \in \RR^n]  \quad & \trace{C^T M} + \epsilon \cD(M) \label{eq:omt_regularized} \\
  \mbox{ subject to } \quad &\Phi_0=M \ett \\
  &\Phi_1=M^T \ett,
\end{aligned}
\end{equation}
where
\begin{align*}
	\cD(M) \triangleq \sum_{i,j=1}^n \left( m_{ij} \log(m_{ij}) - m_{ij} + 1 \right)
\end{align*}
is an entropy term and $\epsilon>0$ a small constant. The entropy regularized OMT problem is related to the Schrödinger bridge problem, which studies the most likely evolution of a particle cloud observed at two different time instances \cite{ leonard2013schrodinger, chen2016relation}. Note that due to the entropy term, the optimal transport plan $M$ to problem \eqref{eq:omt_regularized} always has full support, whereas the solution to \eqref{eq:omt_discrete} may be sparse. 
However, in contrast to \eqref{eq:omt_discrete}, the regularized problem \eqref{eq:omt_regularized} is strictly convex and thus always guarantees a unique solution. Further, it can be shown that, as $\epsilon \to 0$, the solution of \eqref{eq:omt_regularized} converges to the solution of \eqref{eq:omt_discrete} with maximal entropy \cite{Cuturi13}. More importantly, though, the formulation in \eqref{eq:omt_regularized} allows for deriving an efficient method to find the (approximate) optimal transport plan $M$, as it implies a low-dimensional representation requiring only $2n$ variables instead of the original $n^2$ variables. Specifically, the regularized problem allows for expressing the solution $M$ through diagonal scaling of an $n\times n$ constant matrix. To see this, consider the Lagrange function of \eqref{eq:omt_regularized} with dual variables $\lambda_0, \lambda_1 \in \RR^n$, i.e., 
\begin{equation*}
L(M,\lambda_0,\lambda_1) = \trace{C^T M} + \epsilon \cD(M) + \lambda_0^T (\Phi_0- M\ett) + \lambda_1^T( \Phi_1 - M^T \ett).
\end{equation*}
For fixed dual variables, the minimum of $L$ with respect to $M$ is attained when the gradient with respect to the matrix entries $m_{ij}$ vanishes, i.e.,
\begin{align*}
& 0 = \frac{\partial L}{\partial m_{ij}} = c_{ij} + \epsilon \log(m_{ij}) - (\lambda_0)_i - (\lambda_1)_j,
\end{align*}
implying that
\begin{align*}
	m_{ij} = e^{ (\lambda_0)_i / \epsilon} e^{-c_{ij}/\epsilon} e^{ (\lambda_1)_j / \epsilon}.
\end{align*}
Hence, the solution is of the form
\begin{equation} \label{eq:transfer_matrix}
M = \diag(u_0) K \diag(u_1),
\end{equation}
where $K = \exp(-C/\epsilon)$, $u_0=\exp(\lambda_0/\epsilon), u_1=\exp(\lambda_1/\epsilon)$, and where $\exp(\cdot)$ denotes element-wise application of the exponential function. Then, as the matrix $K$ has strictly positive elements, it follows from Sinkhorn's theorem \cite{Sinkhorn67_74} that there is a unique matrix $M$ on the form \eqref{eq:transfer_matrix} with prescribed strictly positive row and column sums $\Phi_0$ and $\Phi_1$.
Moreover, the two positive vectors $u_0$ and $u_1$ are unique up to multiplication with a scalar and may be found via Sinkhorn iterations, i.e., by iteratively updating $u_0$ and $u_1$ as to satisfy the marginal constraints. Specifically, the Sinkhorn iterations (Algorithm~\ref{alg:sinkhorn}) are given by \cite{Cuturi13}
\begin{equation} \label{eq:sinkhorn}
\begin{aligned}
u_0 & = \Phi_0./ (Ku_1)  \\
u_1 & = \Phi_1./ (K^T u_0), 
\end{aligned}
\end{equation}
where $./$ denotes elementwise division. These iterations converge linearly \cite{LorenzF89_114_115}, with the computational bottleneck of the scheme being the two matrix-vector multiplications. The Sinkhorn iterations thus provide an efficient technique for finding approximate solutions to OMT problems.
Parallels between the Sinkhorn iteration scheme and established algorithms for convex optimization problems have been explored in several works. For instance, it has been shown that Sinkhorn iterations correspond to iterative Bregman projections \cite{benamou2015iterative}, and they have been derived as a block-coordinate ascent in the dual formulation of \eqref{eq:omt_regularized} \cite{KarlssonR17_10}. In this work, we utilize the latter result in Section~\ref{sec:discrete_omt_PI} in order to develop new Sinkhorn-type schemes for OMT problems with partial information of the marginals.

\begin{algorithm}
\begin{algorithmic}
\STATE Given: Initial guess $u_t$, for $t=0,1$
\WHILE{Sinkhorn not converged}
	\STATE $u_0  \leftarrow \Phi_0./ (Ku_1)$ 
	\STATE $u_1  \leftarrow \Phi_1./ (K^T u_0)$ 
\ENDWHILE
\RETURN $M \leftarrow \diag(u_0) K \diag(u_1)$
\end{algorithmic}
\caption{Sinkhorn method for the optimal mass transport problem \cite{Cuturi13}.
}\label{alg:sinkhorn}
\end{algorithm}

\subsection{Sinkhorn iterations for the multimarginal optimal mass transport problem} \label{sec:multimarginal_sinkhorn}
Consider a discretized and entropy regularized formulation of the multi-marginal OMT problem in \eqref{eq:multi_margin_omt}, i.e., a direct multi-marginal generalization of \eqref{eq:omt_discrete}.
The discretized marginals are then represented by vectors $\Phi_t\in\RR^n$ for $t=0,\dots,\cT$, with $\cT \geq 1$. In the multi-marginal setting, the mass transport plan and cost function are described by $(\cT+1)$-dimensional tensors $\bM, \bC \in \RR^{n^{\cT+1}}$. 
For ease of notation, we will herein assume that all marginals are discretized using the same number of points. However, this, as well as all presented results, generalizes straightforwardly to the case when the number of discretization points differ among the marginals.
As a discrete analog to the projection in \eqref{eq:proj_operator}, let $\projopdisc{t}{}: \RR^{n^{\cT+1}} \to \RR^n$ be the projection on the $t$-th marginal of a tensor $\bM$, which is computed by summing over all modes of $\bM$ except $t$, i.e., the $j$-th element of the vector $\projopdisc{t}{}( \bM ) \in \RR^n$ is given by
\begin{equation*}
\projopdisc{t}{}( \bM )_j = \sum_{ \substack{i_0,\dots, i_{t-1},\\ i_{t+1},\dots,i_\cT}} \bM_{i_0,\dots,i_{t-1},j,i_{t+1},\dots,i_\cT}.
\end{equation*}
Similarly, in analog to the bi-marginal projection in \eqref{eq:proj_operator_bi_marginal}, let $\projopdisc{t_1,t_2}{}( \bM )\in \RR^{n\times n}$ denote the projection of $\bM$ on the two joint marginals $t_1$ and $t_2$, i.e.,
\begin{equation*}
\projopdisc{t_1,t_2}{}( \bM )_{j,\ell} = \sum_{ \substack{i_t: \\t\in \{0,\ldots,\cT\}\setminus \{t_1,t_2\}}} \bM_{i_0,\dots,i_{t_1-1},j,i_{t_1+1},\dots,i_{t_2-1},\ell,i_{t_2+1},\dots,i_\cT}.
\end{equation*}
With this, the discretized and entropy regularized multi-marginal OMT problem is given by
\begin{equation} \label{eq:multimarginal_entropyreg_omt}
\begin{aligned}
\minwrt[\bM \in \RR^{n^{\cT+1}} ] \ &  \langle \bC,\bM \rangle  + \epsilon \cD(\bM) \\
\mbox{subject to } \ & \projopdisc{t}{}(\bM) = \Phi_t, \quad t=0,\dots,T,
\end{aligned}
\end{equation}
where $\langle \bC,\bM \rangle \triangleq\sum_{i_0,\dots,i_T} \bM_{i_0,\dots,i_T} \bC_{i_0,\dots,i_T}$, and where $\epsilon>0$ is a regularization parameter, with the definition of the entropy term being generalized to tensors as
\begin{equation*}
\cD(\bM) \triangleq \sum_{i_0,\dots,i_\cT} \left(\bM_{i_0,\dots,i_\cT} \log\left( \bM_{i_0,\dots,i_\cT} \right) - \bM_{i_0,\dots,i_\cT} + 1\right).
\end{equation*}

In direct analog to the bi-marginal case, by considering the Lagrangian relaxation \eqref{eq:multimarginal_entropyreg_omt}, one finds that the transport tensor may be represented by "diagonal" scaling of a constant tensor. Specifically, one may write $\bM = \bK \odot \bU$ for the two tensors $\bK, \bU \in\RR^{n^{\cT+1}}$, given by $\bK=\exp(-\bC/\epsilon)$ and 
\begin{equation*}
\bU = (u_0\otimes u_1\otimes \cdots \otimes u_\cT), \quad \iff \bU_{i_0,\dots,i_\cT} = \prod_{t=0}^\cT (u_t)_{i_t},
\end{equation*}
for a set of vectors $u_0,\dots,u_\cT\in \RR^n$.
It may be noted that this representation for the mass transfer tensor $\bM$ is a direct generalization of the representation in the standard OMT theory. That is, in the bi-marginal case we have that $\bM$, $\bK$, and $\bU = u_0 \cdot u_1^T$ are matrices, and
\begin{equation*}
\bM = \bK \odot \bU = \bK \odot (u_0 \cdot u_1^T) = \diag(u_0) \bK \diag(u_1),
\end{equation*}
as in \eqref{eq:transfer_matrix}. Thus, the problem in \eqref{eq:multimarginal_entropyreg_omt} may be reduced to finding $\cT+1$ scaling vectors in $\RR^n$ instead of directly optimizing over the complete mass tensor $\bM \in \RR^{n^{\cT+1}}$. Interestingly, the scaling vectors $u_0,\dots,u_\cT$ are unique up to a scalar for any $\cT \geq 1$ \cite{FranklinPA94}. Furthermore, they may be determined via a Sinkhorn iteration scheme similar to \eqref{eq:sinkhorn}, as presented in \cite{benamou2015iterative} for multi-marginal OMT problems. For completeness, we state the result here: given an initial set of positive vectors $u_t\in\RR^n$, for $t = 0,1,\ldots, \cT+1$, the Sinkhorn method (Algorithm~\ref{alg:multimarginal_sinkhorn}) is to iteratively update according to
\begin{equation} \label{eq:multimarginal_sinkhorn}
u_t \leftarrow u_t\odot \Phi_t./\projopdisc{t}{}(\bK \odot \bU)  , \quad t=0,1,\dots, \cT.
\end{equation}
It may be noted that also this scheme reduces to the standard Sinkhorn iterations \eqref{eq:sinkhorn} for the case $\cT=1$, i.e., a bi-marginal problem.
In \cite{benamou2015iterative}, these iterations were derived based on projections using the Kullback-Leibler divergence as a distance measure. In contrast, we here derive them as a block-coordinate ascent in a Lagrange dual of \eqref{eq:multimarginal_entropyreg_omt}, extending the results for the bi-marginal problem presented in \cite{KarlssonR17_10}. These results are presented in Section~\ref{sec:discrete_omt_PI} for the more general setting 
where only partial information of the marginals $\Phi_t$ is available for $t = 0,1,\ldots,\cT$, which the problem with full information in \eqref{eq:multimarginal_entropyreg_omt} is a special case of. 

\begin{algorithm}
\begin{algorithmic}
\STATE Given: Initial guess $u_t>0,$ for $t=0,\ldots,\cT$; starting point $t$
\WHILE{Sinkhorn not converged}
	\STATE  
		$u_t \leftarrow  u_t\odot \Phi_t./\projopdisc{t}{}(\bK \odot \bU)$  ,
	\STATE $t \leftarrow t + 1 \;({\rm mod}\; \cT+1)$
\ENDWHILE
\RETURN $\bM=\bK \odot \bU$
\end{algorithmic}
\caption{Sinkhorn method for the multimarginal optimal mass transport problem \cite{benamou2015iterative}.}\label{alg:multimarginal_sinkhorn}
\end{algorithm}

\subsection{Efficient computation for decoupling cost functions} \label{subsec:decoupling_costfuncs}
In many applications, the size of the mass transport tensor $\bM$, and thus of the auxiliary tensors $\bK$ and $\bU$, may be too large to manipulate directly. It is thus crucial to utilize additional structures in the problem whenever this is possible. In this section, some examples for structures in the cost tensor $\bC$ are described, allowing for efficient computation of subproblems in the multi-marginal Sinkhorn algorithm, detailed in Algorithm~\ref{alg:multimarginal_sinkhorn}.
The computational bottleneck of the Sinkhorn scheme in \eqref{eq:multimarginal_sinkhorn} is the computation of $\projopdisc{t}{}(\bK \odot \bU)$, where both the elementwise multiplication $\bK \odot \bU$ and the application of the projection operator $\projopdisc{t}{}$ may be expensive. However, in some cases, the choice of cost function, $\costfunc$, induces structure in the tensor $\bK$ that may be exploited in order to dramatically increase the efficiency of computing $\projopdisc{t}{}(\bK \odot \bU)$. Here, we will consider cases relevant for the applications presented in Section~\ref{sec:omt_multi_margin}.
In fact, for these cases, the computation of $\projopdisc{t}{}(\bK \odot \bU)$ may be performed by sequences of matrix-vector multiplications of the form $K u$, where $K$ is a matrix and $u$ is a vector, as detailed in the propositions in the following sections.
The proofs of the propositions may be found in Appendix~\ref{appendix:proofs}.
\subsubsection{Sequentially decoupling cost function}
Consider problems on the form \eqref{eq:multi_margin_omt} where the cost function $\costfuncmulti: \bcX \times \bcX \to \RR_+$ decouples sequentially according to \eqref{eq:costfunc_sequential}, which as detailed in Section~\ref{sub_sec:tracking} may be utilized in order to model sequential tracking over time. The following proposition shows how this special case allows for efficient computations of the projections $\projopdisc{t}{}(\bK \odot \bU)$ and $\projopdisc{t_1,t_2}{}(\bK \odot \bU)$ in the discrete OMT problem.
%
%
\begin{proposition} \label{prp:sequential_cost}
Let the elements of the cost tensor $\bC$ in \eqref{eq:multimarginal_entropyreg_omt} be of the form
\begin{equation*}
\bC_{i_0,\dots,i_\cT} = \sum_{t=1}^\cT C_{i_{t-1}, i_t},
\end{equation*}
for a cost matrix $C\in\RR^{n\times n}$, and let $\bK = \exp(-\bC/\epsilon)$, $K=\exp(-C/\epsilon)$, and $\bU=(u_0\otimes u_1 \otimes\cdots\otimes u_T)$. Then, for $0 \le t\le \cT$,
\begin{align*}
\projopdisc{t}{}(\bK \odot \bU)=&\left( u_0^T K \diag(u_1) K \dots K\diag(u_{t-1}) K \right)^T  \odot u_t \odot \left( K\diag(u_{t+1})K \ldots K \diag(u_{\cT-1}) K u_\cT \right)\\
\end{align*}
and, for $0 \le t_1< t_2\le \cT$,
\begin{equation} \label{eq:coupling_tracking}
\begin{aligned}
\projopdisc{t_1,t_2}{}(\bK \odot \bU)=&\diag\left( u_0^T K \diag(u_1) K \cdots K\diag(u_{t_1-1}) K \right)  \\
&\odot \diag(u_{t_1}) K\diag(u_{t_1+1}) K\ldots K \diag(u_{t_2-1})K \diag(u_{t_2})  \\
&\odot \diag\left( K\diag(u_{{t_2}+1})K \ldots K \diag(u_{\cT-1}) K u_\cT \right).
\end{aligned}
\end{equation}
\end{proposition}

\begin{proof}
	See Appendix~\ref{appendix:proofs}.
\end{proof}
%
It is worth noting that the computation of $\projopdisc{t}{}(\bK \odot \bU)$ requires only $\cT$ matrix-vector multiplications of the form $K u_\tau$, i.e., $\mathcal{O}(\cT n^2)$ operations, whereas the complexity for a brute-force approach is $\mathcal{O}(n^{\cT+1})$ operations. The same holds for the multiplication of $\projopdisc{t_1,t_2}{}(\bK \odot \bU)$ with a vector. Furthermore, if the domains of the marginals have cost functions that decouple then the efficiency of the matrix vector computations can be improved even further, as described in the following remark. 

\begin{remark} \label{remark:tensor_mode_product}
In many multidimensional problems, the domain $X$ can be represented as a direct product $X=X_1\times X_2\times\cdots \times X_N$, with $|X_i|=n_i$. If, in addition, the  cost decouples in the same way, i.e.,  for two points $\discstate_0=\left(\discstate_0^{(1)}, \discstate_0^{(2)},\ldots, \discstate_0^{(N)}\right) \in X $ and $\discstate_1=\left(\discstate_1^{(1)}, \discstate_1^{(2)},\ldots, \discstate_1^{(N)}\right) \in X$,
		\begin{equation*} 
		c(\discstate_0,\discstate_1) = \sum_{i=1}^N c_i\left(\discstate_0^{(i)},\discstate_1^{(i)}\right)
		\end{equation*}
		where $c_i:X_i\times X_i\to\RR$ is a cost function in the $i$-th dimension, for $i=1,\ldots,N$, then the matrix $K=\exp(-C/\epsilon) \in \RR^{n\times n}$, where $n = \prod_{i=1}^N n_i$, can be decoupled as a tensor product $K=K_1\otimes K_2\otimes\cdots \otimes K_N$. 
		Multiplications with $K$ can then be done one dimension at a time. Practically, this can be efficiently implemented using specialized tensor toolboxes, e.g., with the command \texttt{tmprod} of the Matlab toolbox Tensorlab \cite{VervlietDSBL16_tensorlab}.
		Thus, the multiplication $K u$ requires only $\mathcal{O}(n \sum_i n_i)$ operations instead of $\mathcal{O}(n^2)$ which is required when using standard multiplication. From a memory perspective, this is also important since $K$ does not need to be computed and stored. In particular, this holds for a regular rectangular grid in $\mR^N$ with cost function $c(x,y)=\|x-y\|_p^p$, for some $p\geq 1$, which is the typical setup for the Wasserstein metric.
\end{remark}

\begin{remark}
Note that with a sequential cost as in Proposition \ref{prp:sequential_cost}, the multi-marginal OMT problem \eqref{eq:multimarginal_sinkhorn} may be separated into optimization problems containing only two marginals each:
\begin{equation*}
T(\Phi_{t-1}, \Phi_t),  \quad t=1,\dots,\cT.
\end{equation*}
\end{remark}
%

\subsubsection{Centrally decoupling cost function}
In order to model the barycenter formulation of OMT, Section~\ref{sub_sec:sensor_fusion} presented the use of centrally decoupling cost functions $\costfuncmulti$ according to \eqref{eq:costfunc_central}, as to model transport between a common barycenter and several marginal distributions. The following proposition describes the projection computations for this case.
%
%
\begin{proposition} \label{prp:barycenter_cost}
Let the elements of the cost tensor $\bC$ in \eqref{eq:multimarginal_entropyreg_omt} be of the form
\begin{equation*}
\bC_{i_0,\dots,i_J} = \sum_{j=1}^J C_{i_0, i_j},
\end{equation*}
for a cost matrix $C\in\RR^{n\times n}$, and let $\bK = \exp(-\bC/\epsilon)$, $K=\exp(-C/\epsilon)$, and $\bU=(u_0\otimes u_1 \otimes\cdots\otimes u_T)$. Then,
\begin{align}
\projopdisc{0}{}(\bK \odot \bU)=& u_0 \odot \bigodot_{\ell=1}^J ( K u_\ell), \label{eq:proj_barycenter_0}\\
\projopdisc{j}{}(\bK \odot \bU)=& u_j \odot K^T\left( u_0 \odot \bigodot_{\ell=1,\ell\neq j}^J( K u_\ell)\right)\quad \mbox{ for } j=1,\ldots, J. \label{eq:proj_barycenter_j}
\end{align}
For $j,j_1,j_2 =1,\dots,J $, and $j_1 \neq j_2$, the pair-wise projections are given by
\begin{align}
\projopdisc{0,j}{} (\bK \odot \bU) = & \diag\left( u_0 \odot \bigodot_{\substack{\ell=1 \\ \ell \neq j}}^J (K u_\ell) \right) K \diag(u_j),\label{eq:proj_barycenter_0j}\\
\projopdisc{j_1,j_2}{} (\bK \odot \bU) = & \diag(u_{j_1}) K^T \diag \left( u_0 \odot \bigodot_{\substack{\ell=1 \\\ell \neq j_1,j_2}}^J (K u_\ell) \right)  K \diag(u_{j_2}).\label{eq:proj_barycenter_jj}
\end{align}

\end{proposition}

\begin{proof}
	See Appendix~\ref{appendix:proofs}.
\end{proof}

\subsubsection{Sequential and central decoupling}
Consider a cost function $\costfuncmulti$ being a combination of sequentially and centrally decoupling costs, according to \eqref{eq:cost_tensor_barycenter_tracking}, which in Section~\ref{subsec:barycenter_tracking} was applied to modeling tracking of barycenters over time. For this case, the following proposition holds.
%

\begin{proposition} \label{prp:combined_cost}
Let the elements of the cost tensor $\bC$ in \eqref{eq:multimarginal_entropyreg_omt} be of the form
\begin{equation*}
\bC_{\left(i_{(t,j)}|(t,j)\in \Lambda\right)} = \sum_{t=1}^\cT C_{i_{({t-1},0)}, i_{(t,0)}} +\sum_{t=1}^\cT \sum_{j=1}^J \tilde{C}_{i_{(t,0)}, i_{(t,j)}},
\end{equation*}
for cost matrices $C\in\RR^{n\times n}$ and $\tilde{C} \in \RR^{n\times \tilde n}$, and define $\Lambda=\{(t,j)\,|\, t\in\{0,1,\ldots, \cT\}, j\in\{0,1\,\ldots,J\}\}$. Furthermore, let $\bK = \exp(-\bC/\epsilon)$, $K=\exp(-C/\epsilon)$, $\tilde{K}=\exp(-\tilde{C}/\epsilon)$, and let $\bU=\bigotimes_{(t,j)\in \Lambda} u_{(t,j)}$. Then, for the central marginals, corresponding to $(t,j)$ such that $j = 0$,
\begin{equation} \label{eq:projection_t0}
\projopdisc{(t,0)}{}(\bK \odot \bU)=\left( p_0^T K \diag(p_1) K \dots \diag(p_{t-1}) K \right)^T  \odot p_t \odot \left( K\diag(p_{t+1}) \ldots K \diag(p_{\cT-1}) K p_\cT \right),
\end{equation}
where $p_t=u_{(t,0)}\odot \bigodot_{j=1}^J \tilde{K}u_{(t,j)}$. Furthermore, for the non-central marginals, i.e., $(t,j)$ with $j = 1,2,\ldots,J$, it holds that
\begin{equation} \label{eq:projection_tj}
\begin{aligned}
	\projopdisc{(t,j)}{}(\bK \odot \bU)&=u_{(t,j)}\odot \tilde{K}^T\bigg(( p_0^T K \diag(p_1) K \dots K\diag(p_{t-1}) K )^T \\
	&\qquad\odot (p_t./(Ku_{(t,j)}) ) \odot ( K\diag(p_{t+1}) K \ldots K \diag(p_{\cT-1}) K p_\cT )\bigg).
\end{aligned}
\end{equation}
\end{proposition}
\begin{proof}
	See Appendix~\ref{appendix:proofs}.
\end{proof}
%
%
\begin{remark} \label{remark:bi_projection_barycenter_tracking}
The pairwise projections for index pairs $(t_1,0)$ and $(t_2,0)$, i.e., indices corresponding to the central marginals, may be computed analogously to Proposition~\ref{prp:sequential_cost}. That is,
\begin{align*}
\projopdisc{(t_1,0),(t_2,0)}{} (\bK \odot \bU) =&\diag\left( p_0^T K \diag(p_1) K \dots K\diag(p_{t_1-1}) K \right)  \\
&\odot \diag(p_{t_1}) K\diag(p_{t_1+1}) K\cdots K \diag(p_{t_2-1})K \diag(p_{t_2})  \\
&\odot \diag\left( K\diag(p_{{t_2}+1})K \ldots K \diag(p_{\cT-1}) K p_\cT \right).
\end{align*}
In case the index pairs are $(t,j_1)$ and $(t,j_2)$, i.e., when the indices correspond to the same central marginal, the pairwise projections can be computed similarly to Proposition~\ref{prp:barycenter_cost}. Specifically, the bi-marginal projection for a central marginal and a non-central marginal corresponding to the same time $t$ is given by
\begin{align*}
	\projopdisc{(t,0),(t,j)}{} =&  \diag(u_{(t,0)}) \diag  \Bigg( (\tilde K)^T\Big(( p_0^T K \diag(p_1) K \dots K\diag(p_{t-1}) K )^T  \odot (p_t./(u_{(t,0)} \odot \tilde K u_{(t,j)}) ) \\
	& \hspace{80pt} \odot ( K \diag(p_{t+1})K \ldots K \diag(p_{T-1}) K p_T )\Big) \Bigg) \tilde K \diag(u_{(t,j)}),
\end{align*}
whereas the bi-marginal projection for two non-central marginals corresponding to the same time $t$ may be expressed as
\begin{align*}
\projopdisc{(t,j_1),(t,j_2)}{} = & \diag(u_{(t,j_1)}) \tilde K^T \diag \Bigg( ( p_0^T K \diag(p_1) K \dots \diag(p_{t-1}) K )^T  \odot (p_t./(\tilde K u_{(t,j_1)} \odot \tilde K u_{(t,j_2)}) ) \\
	& \hspace{80pt} \odot ( K \diag(p_{t+1}) \ldots K \diag(p_{T-1}) K p_T ) \Bigg)  \tilde K \diag(u_{(t,j_2)}).
\end{align*}
For arbitrary index pairs $(t,j)$, expressions for the pairwise projections may be derived in a similar way.
\end{remark}

\section{Solving OMT problems with partial information of marginals} \label{sec:discrete_omt_PI}
In this section, we derive an efficient algorithm for solving discrete OMT problems where only partial information of the marginals are available, i.e., formulating discrete approximations of the problem in \eqref{eq:multi_margin_omt_PI}. In this setting, one does not have access to the marginal vectors $\Phi_t \in \RR^n$ directly, but instead have measurements of the form $r = \disccovop \Phi$, where $r \in \RR^m$, with $m < n$, is a vector of available information, and where $\disccovop \in \RR^{m \times n}$ represents a mapping from the full state information to the partial information. For example, $\disccovop$ may represent a discretization of the push-forward operator $F_\#$ in the case of an OMT problem with dynamics detailed by \eqref{eq:system_eq}, or of the operator $\Gamma$, mapping spectra to covariance matrices. Correspondingly, $r$ may be a discretization of the push-forward measure, or the (vectorized) covariance matrix, respectively. Note that all quantities in this section are assumed to be real-valued. As the considered covariance matrices $R$ are, in general, complex-valued, we arrive at the equivalent real-valued problems by constructing the corresponding information vectors $r$ as
\begin{align} \label{eq:cov_vector}
	r = \left[ \begin{array}{cc} \realpart\left(\text{vec}(R)\right)^T  & \imagpart\left(\text{vec}(R)\right)^T \end{array} \right]^T,
\end{align}
where $\realpart(\cdot)$ and $\imagpart(\cdot)$ denote the real and imaginary parts, respectively. The discretizations of the operators $\Gamma$, i.e., $\disccovop$, are then structured as to be consistent with this construction. It should be noted that the solution algorithm presented in this section coincides with the multi-marginal Sinkhorn iterations in \eqref{eq:multimarginal_sinkhorn} for the special case of full  information (corresponding to letting $\disccovop$ be the identity operator)
 and noiseless observations of the distribution marginals ($\gamma_t=\infty$).

\subsection{Entropy regularized multi-marginal OMT with partial information of marginals} \label{}
Consider a multi-marginal OMT problem, in which only partial, as well as noisy, information of the marginals is available, as detailed in \eqref{eq:multi_margin_omt_PI}. A direct discretization of this problem, with added entropy regularization, may be expressed as
\begin{equation} \label{eq:multimargin_generalcost}
\begin{aligned}
\minwrt[\bM, \Delta_t] \quad & \; \langle \bC, \bM \rangle 
+\epsilon \cD(\bM)+\sum_t \gamma_t \|\Delta_t\|_2^2 \\
\mbox{subject to} &\; \disccovop_t \projopdisc{t}{}(\bM)=r_t+\Delta_t \quad \mbox{ for } t=0,\ldots, \cT,
\end{aligned}
\end{equation}
where $\disccovop_t \in \RR^{m_t \times n}$ is a mapping from the full state information to the partial information $r_t \in \RR^{m_t}$, as described above. Also, to allow for noisy measurements $r_t$, the problem is augmented by perturbation vectors $\Delta_t \in \RR^{m_t}$ by direct analog to the matrix perturbations of \eqref{eq:multi_margin_omt_PI}, which are penalized by the squared $\ell_2$-norm scaled by the penalty parameters $\gamma_t > 0$. It may be noted that due to the presence of the entropy term, $\cD$, the problem in \eqref{eq:multimargin_generalcost} is strictly convex. Also, for finite $\gamma_t$, a solution always exists due to the introduction of the perturbation vectors $\Delta_t$. Altogether, this implies that for any $\epsilon >0$ and $\gamma_t < \infty$, the optimization problem \eqref{eq:multimargin_generalcost} has a unique solution. The following proposition characterizes the solution $\bM$ in terms of the variables of the Lagrange dual problem, yielding the dual formulation. 
%

\begin{proposition} \label{prp:multimarginal_dual}
The optimal solution to \eqref{eq:multimargin_generalcost} may be expressed as $\bM = \bK \odot \bU$, for two tensors $\bK, \bU \in\RR^{ n^{\cT+1}}$, given by $\bK=\exp(-\bC/\epsilon)$, and 
\begin{equation}\label{eq:U}
\bU = u_0\otimes u_1\otimes \cdots \otimes u_\cT, 
\quad \iff \ \bU_{i_0,\dots,i_\cT} = \prod_{t=0}^\cT (u_t)_{i_t},
\end{equation}
where the vectors $u_t$ are given by $u_t = \exp(\disccovop_t^T \lambda_t / \epsilon)$ for $t=0,\dots,\cT$, with $\lambda_t \in \RR^{m_t}$ denoting the Lagrange dual variable corresponding to the constraint on the projection $\projopdisc{t}{}(\cdot)$.
The optimal perturbation vectors may be expressed as $\Delta_t=-\frac{1}{2\gamma_t} \lambda_t$, for $t = 0,1,\ldots, \cT$.
%
Furthermore, a Lagrange dual of the multi-marginal OMT problem \eqref{eq:multimargin_generalcost} is given by
\begin{equation} \label{eq:dual_multimarginal_omt}
\maxwrt_{\lambda_0,\ldots, \lambda_\cT} \ \ -\epsilon \ \langle \bK, \bU\rangle -\frac{1}{4\gamma}\sum_{t=0}^\cT \|\lambda_t\|_2^2+ \sum_{t=0}^\cT\lambda_t^Tr_t.
\end{equation}
where $\bU$ is given by \eqref{eq:U} with $u_t = \exp(\disccovop_t^T \lambda_t / \epsilon)$ for $t=0,\dots,\cT$.
\end{proposition}

\begin{proof}
	See Appendix~\ref{appendix:proofs}.
\end{proof}

%
With the result from Proposition~\ref{prp:multimarginal_dual}, we are now ready to state the method for solving \eqref{eq:multimargin_generalcost}.
%
%
%
\begin{theorem} \label{thm:multimarginal_scheme}
Given an initial set of vectors $\lambda_0,\dots,\lambda_\cT$, a block-coordinate ascent method for the dual of the multi-marginal OMT problem \eqref{eq:multimargin_generalcost} is to iteratively, until convergence, iterate the following steps:
\begin{itemize}
\item Let \begin{equation} \label{eq:multimarginal_omt_partial_info_vt}
v_t = \projopdisc{t}{}(\bK \odot \bU)./u_t
\end{equation}
where 
$\bU = u_0\otimes u_1\otimes \cdots \otimes u_\cT$
for the vectors $u_t = \exp(\disccovop_t^T \lambda_t / \epsilon)$, $t=0,\dots,\cT$.

\item Update the vector $\lambda_t$ as the solution to
\begin{equation}
\disccovop_t \left( v_t \odot \exp\left(\frac{\disccovop_t^T\lambda_t}{\epsilon}\right) \right) + \frac{\lambda_t}{2\gamma_t} - r_t = 0. \label{eq:tracking_sinkhorn_maximization0}
\end{equation}
\end{itemize}
\end{theorem}

\begin{proof}
	See Appendix~\ref{appendix:proofs}.
\end{proof}


%
\begin{remark}
It is worth noting that, given the dual optimal variables $\lambda_t$, for $t = 0,1,\ldots,\cT$, obtained as the limit point of the iterations in Theorem~\ref{thm:multimarginal_scheme}, the primal optimal variables may be recovered using Proposition~\ref{prp:multimarginal_dual}.
\end{remark}
\begin{remark}
The number of constraints in \eqref{eq:multimargin_generalcost} may be smaller than the number of modes in the mass transfer tensor $\bM$.
In cases where there are no constraints on the $t$-th projection of $\bM$, the corresponding dual variable $\lambda_t$ is set to zero. The method described in Theorem \ref{thm:multimarginal_scheme} may then be modified by setting $u_t=\ett$ and only iterating over the smaller set of remaining vectors.
\end{remark}
Recall that in the standard Sinkhorn iterations \eqref{eq:sinkhorn} and \eqref{eq:multimarginal_sinkhorn}, the scaling factors $u_t$ are iteratively updated as to satisfy the marginal constraints. The same property holds for the iterations in Theorem~\ref{thm:multimarginal_scheme}. Specifically, if $\lambda_t$ is a solution to \eqref{eq:tracking_sinkhorn_maximization0}, then the implied transport plan in that iteration, i.e., $\bM = \bK \odot \bU$, where $\bU = u_0\otimes u_1\otimes \cdots \otimes u_\cT$, with $u_t = \exp(\disccovop_t^T \lambda_t / \epsilon)$, satisfies
\begin{equation*}
	\disccovop_t P_t(\bM)=r_t+\Delta_t = r_t -\frac{1}{2\gamma_t} \lambda_t.
\end{equation*}
In the case of full information and exact matching in a marginal (i.e., $\disccovop_t=I$ and $\gamma_t=\infty$), the update \eqref{eq:tracking_sinkhorn_maximization0} in Theorem~\ref{thm:multimarginal_scheme} reduces to the Sinkhorn iterations (Algorithm~\ref{alg:multimarginal_sinkhorn}) introduced in Section \ref{sec:multimarginal_sinkhorn}, i.e.,
\[
u_t \leftarrow r_t./v_t=u_t\odot r_t./\projopdisc{t}{}(\bK \odot \bU),
\]
whereas in the case with only full information (i.e., $\disccovop_t=I$ and $0<\gamma_t<\infty$), the update reduces to solving
\begin{equation}
 v_t \odot \exp(\lambda_t/\epsilon)  + \lambda_t/(2\gamma_t) - r_t = 0. 
\end{equation}
This equation may be solved element-wise using the Wright omega function (see appendix B in \cite{KarlssonR17_10}). This allows for computing the proximal operator of the regularized optimal mass transport \cite{KarlssonR17_10}, which is often used in first-order methods for non-smooth optimization. 
Similar expressions may also be obtained for other penalization terms than the squared $\ell_2$-error, which relates to the entropic proximal operator \cite{chizat2018scaling}.
Herein, we propose to solve \eqref{eq:tracking_sinkhorn_maximization0} using Newton's method. Therefore, note that the corresponding Jacobian is given by
\begin{equation*}
\frac{1}{\epsilon} \disccovop_t \diag( v_t \odot u_t ) \disccovop_t^T + \frac{1}{2\gamma_t}I.
\end{equation*}
The full method for solving \eqref{eq:multimargin_generalcost} is summarized in Algorithm~\ref{alg:multitracking_sinkhornnewton}. It may be noted that the multi-marginal transport plan $\bM$ may for some problems be too large to be stored in the memory. However, when used in modeling applications, one is primarily interested in being able to study projections of $\bM$, which, due to the structure of the transport plan, do not require $\bM$ to be constructed explicitly, as described in Section~\ref{subsec:decoupling_costfuncs}. Also, it may be noted that the marginals of $\bM$ are given directly at convergence of Algorithm~\ref{alg:multitracking_sinkhornnewton}, as $\Phi_t = u_t \odot v_t$, for $t=0,1,\ldots,\cT$.
%
\begin{algorithm}
\begin{algorithmic}
\STATE Given: Initial guess $\lambda_t, $ for $t=0,\ldots,\cT$; starting point $t$\\
$u_t \leftarrow \exp(\disccovop_t\lambda_t/\epsilon)$ for $t=0,\ldots,\cT$
\WHILE{Sinkhorn not converged}
	\STATE Construct $v_t$ according to \eqref{eq:multimarginal_omt_partial_info_vt}, i.e.,  \\
	$v_t \leftarrow \projopdisc{t}{}(\bK \odot \bU)./u_t$
	\WHILE{Newton not converged}
		\STATE $f \leftarrow \disccovop_t \diag(v_t)u_t - r_t + 1/(2\gamma_t) \lambda_t$
		\STATE $df \leftarrow (1/\epsilon) \disccovop_t \diag( v_t \odot u_t ) \disccovop_t^T + 1/(2\gamma_t)I $
		\STATE $\Delta\lambda \leftarrow - df\backslash f$
		\STATE $\lambda_t \leftarrow \lambda_t+\eta\Delta\lambda$, with $\eta$ determined by a linesearch
		\STATE $u_t \leftarrow \exp(\disccovop_t\lambda_t/\epsilon)$
	\ENDWHILE
	\STATE  $t \leftarrow t + 1 \;({\rm mod}\; \cT+1)$ 
\ENDWHILE
\RETURN $\bM = \bK \odot \bU$
\end{algorithmic}
\caption{Sinkhorn-Newton method for the multimarginal optimal mass transport with partial information of the marginals.}\label{alg:multitracking_sinkhornnewton}
\end{algorithm}

\begin{remark} \label{remark:Newton}
Sufficiently close to the optimal solution, the quadratic approximation of the dual objective underlying the Newton method becomes increasingly accurate. The inner Newton method for solving \eqref{eq:multimarginal_omt_partial_info_vt} then converges in the first iteration. In the authors' experience, this is typically achieved within the first few outer Sinkhorn iterations.	
\end{remark}

The computational bottleneck of Algorithm~\ref{alg:multitracking_sinkhornnewton} is the construction of the vectors $v_t$ in \eqref{eq:multimarginal_omt_partial_info_vt}, requiring the computation of $\projopdisc{t}{}(\bK \odot \bU)$. However, the methods presented in Section~\ref{subsec:decoupling_costfuncs} for exploiting structure in the cost tensor $\bC$, and thereby $\bK$, are directly applicable to solving the multi-marginal OMT problem with partial information. In particular, Theorem~\ref{thm:multimarginal_scheme} provides an efficient scheme for implementing the spatial spectral estimators in Sections~\ref{sub_sec:sensor_fusion} - \ref{subsec:barycenter_tracking}, with Propositions~\ref{prp:sequential_cost}, \ref{prp:barycenter_cost}, and \ref{prp:combined_cost} detailing the computation of $\projopdisc{t}{}(\bK \odot \bU)$ for the tracking, sensor fusion, and barycenter tracking problems, respectively.
%
%

\section{Numerical results}\label{sec:numerical_results}\vspace{-1mm}

In this section, we demonstrate the behavior of the proposed multi-marginal OMT problem in \eqref{eq:multi_margin_omt_PI} when applied to the discussed spatial spectral estimation problems. Throughout, the problem in \eqref{eq:multi_margin_omt_PI} is approximated by the discrete counterpart in \eqref{eq:multimargin_generalcost}, and solved using Algorithm~\ref{alg:multitracking_sinkhornnewton}. For the implementation of the projection operators, $\projopdisc{t}{}$, we use the results from the propositions in Section~\ref{subsec:decoupling_costfuncs}.
For the example in Section~\ref{subsec:2D_barycenter_tracking}, we provide a dimensionality analysis of the corresponding OMT problem, illustrating the benefit of the presented computational tools.

\subsection{Tracking with static and dynamical models} \label{subsec:1D_tracking}
In order to illustrate the different properties of the static and dynamic OMT formulations, we consider a DoA tracking example with two moving signal sources. Specifically, we consider a ULA consisting of 5 sensors with half-wavelength spacing, measuring the superposition of the source signals, modeled as independent Gaussian processes, together with a spatially white complex Gaussian sensor noise. The signal to noise ratio (SNR), defined as 
\begin{align*}
	\text{SNR} = \log_{10}(\sigma^2_{signal}/\sigma^2_{noise}),
\end{align*}
where $\sigma^2_{signal}$ and $\sigma^2_{noise}$ are the signal and noise powers, respectively, is 10 dB. The trajectories of the two sources are displayed in the top panel of Figure~\ref{fig:tracking_velocity_dynamic_gt}, with the bottom panel showing the velocities. At six different time instances, evenly spaced in time throughout the trajectories, we collect 25 array signal snapshots, from which the array covariance is estimated using the sample covariance matrix. Using the resulting sequence of covariance matrices, we then attempt to reconstruct the target trajectories using the static and dynamic OMT problems. For the static OMT problem, we use the cost function $\costfunc(\theta,\varphi) = \abs{\theta-\varphi}^2$, for $\theta,\varphi \in (-\pi,\pi]$. For the dynamical model, we introduce a latent velocity state, $v$, such that $\dot{\theta}(t) = v(t)$, and use the state space representation in \eqref{eq:system_eq}, where
\begin{align*}
	A = \begin{bmatrix}
		0 & 1 \\ 0 & 0
	\end{bmatrix} \;,\; B = \left[\begin{array}{cc} 0 & 1  \end{array}\right]^T \;,\; F = \left[\begin{array}{cc} 1 & 0  \end{array}\right],
\end{align*}
with the state vector being formed as $x(t) = \left[\begin{array}{cc} \theta(t) & v(t)\end{array}\right]^T$, i.e., the same model as in Example~\ref{example:angle_and_speed}. Note here that the $F$ matrix reflects the fact that only the angle, $\theta$, is directly observable in the array covariance matrix. The choice of $B$ implies that the angle may only be influenced via the velocity, i.e., through acceleration. Thus, the resulting dynamic cost function, as defined in \eqref{eq:quad_costfunc}, penalizes transport requiring acceleration, whereas the cost function for the static problem in contrast penalizes transport requiring velocity.
In the discrete implementations of the methods, we use 100 grid points to represent the angle, $\theta$, and the dynamical model uses 30 grid points for the velocity, $v$. Also, we use the regularization parameters $\epsilon = 10^{-1}$, and common parameters $\gamma = 30$ and $\gamma = 5$ for all marginals for the static and dynamic OMT models, respectively.
%
%
\begin{figure}[t]
        \centering
            \includegraphics[width=.5\textwidth]{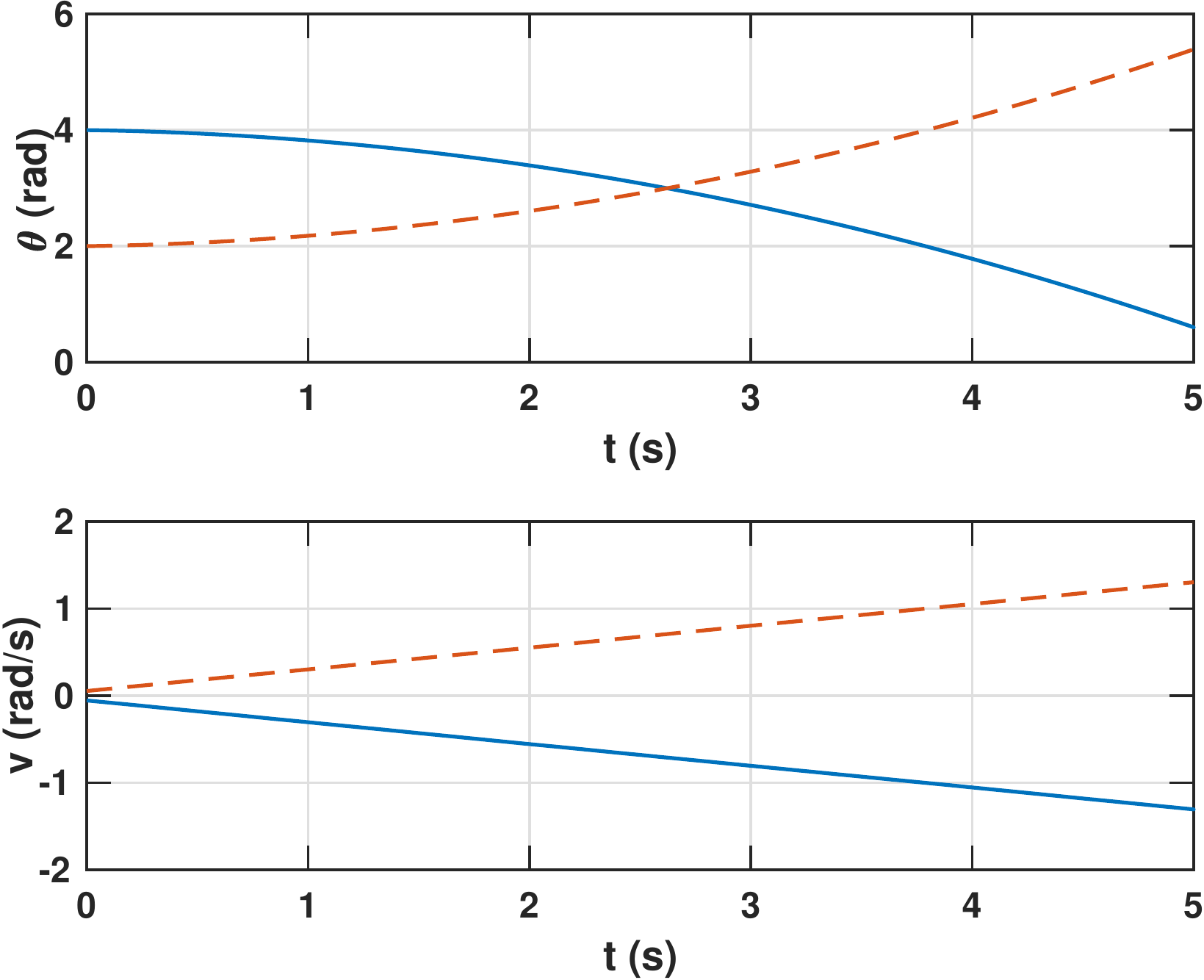}
           \caption{Ground truth the DoA tracking example with two moving targets. Top panel: the target DoAs, in radians, as a function of time. Bottom panel: the target velocities, in radians per second, as a function of time.}
            \label{fig:tracking_velocity_dynamic_gt}
\vspace{-2mm}\end{figure}
%
The results for the static and dynamical models are displayed in the top panels of Figures~\ref{fig:tracking_location_static} and \ref{fig:tracking_location_vel_dynamic}, respectively. Here, the trajectories in between the observation times, which are indicated by vertical dashed lines, are reconstructed using the interpolation procedure presented in Section~\ref{sub_sec:omt_dyn}, as defined by the obtained multi-marginal transport plan.
As noted in Section~\ref{sub_sec:omt_dyn}, this requires bi-marginal transport plans describing the mass transfer between the respective margins. These plans are computed as the bi-marginal projections detailed in Proposition~\ref{prp:sequential_cost}.
To simplify comparison, the ground truth trajectories from Figure~\ref{fig:tracking_velocity_dynamic_gt}, at the observation times, are superimposed in the plots. Also, the bottom panel of Figure~\ref{fig:tracking_location_static} displays the spectral estimates obtained by applying the Capon estimator \cite{Capon69} to the individual covariance matrix estimates. As may be noted, the Capon estimate contains several spurious peaks, caused by the noisy measurements. In contrast, as seen in Figure~\ref{fig:tracking_location_static}, the static model is able to produce reasonable spectral estimates for the observation times. However, the reconstructed trajectories fails to model the crossing of the paths of the targets. This is not unexpected: as the static OMT model penalizes movement, i.e., velocity, the cost of transport is smaller if the targets instead change course (note the trajectory between $t = 2 $ and $t = 3$). In contrast, the dynamical formulation, in which movement is expected, is able to produce considerably more accurate estimates, as may be seen in Figure~\ref{fig:tracking_location_vel_dynamic}, including the crossing of the paths of the targets. Note also that the dynamical formulation is able to reconstruct the spectrum also for the hidden velocity state.

%
%
\begin{figure}[t]
        \centering
            \includegraphics[width=.6\textwidth]{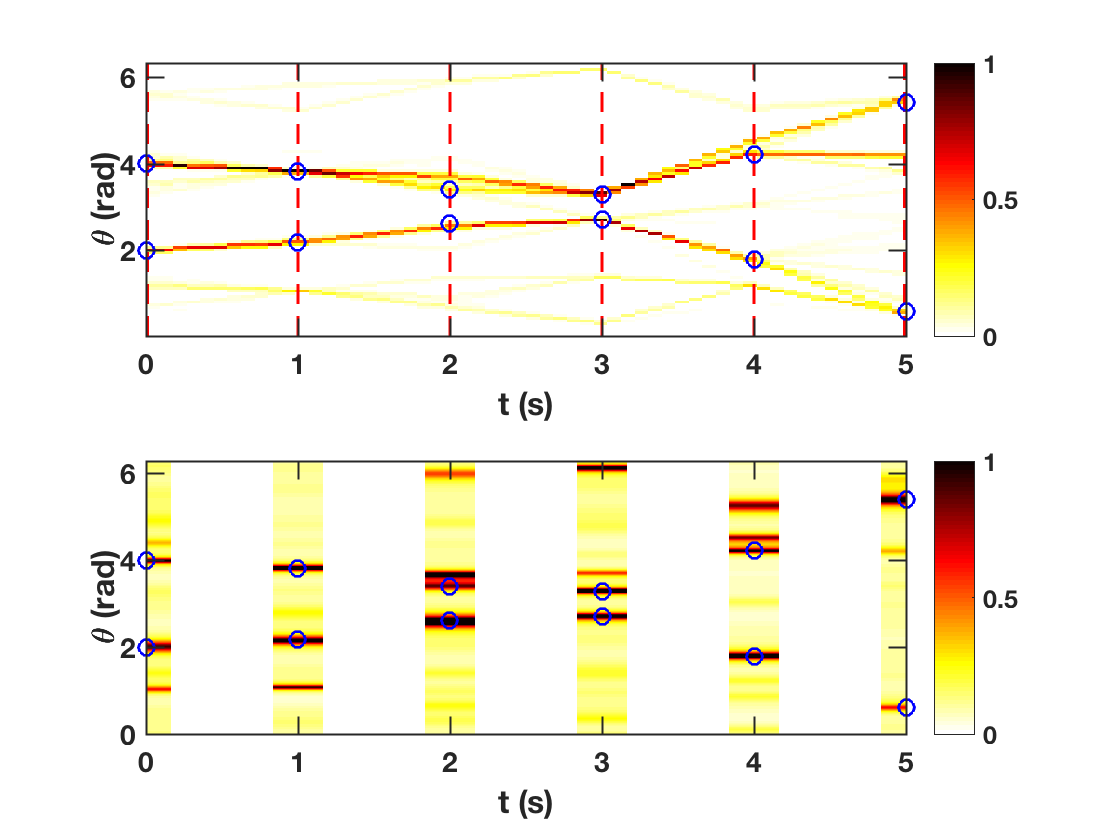}
           \caption{Top panel: reconstructed DoA spectrum using a static cost function in the OMT formulation, based on six observations of the array covariance matrix. The observation times are marked in dashed lines. The ground truth trajectories at the observation times are marked by rings. Bottom panel: estimates obtained using the Capon estimator applied to the individual covariance matrices.}
            \label{fig:tracking_location_static}
\vspace{-2mm}\end{figure}

\begin{figure}[t]
        \centering
            \includegraphics[width=.6\textwidth]{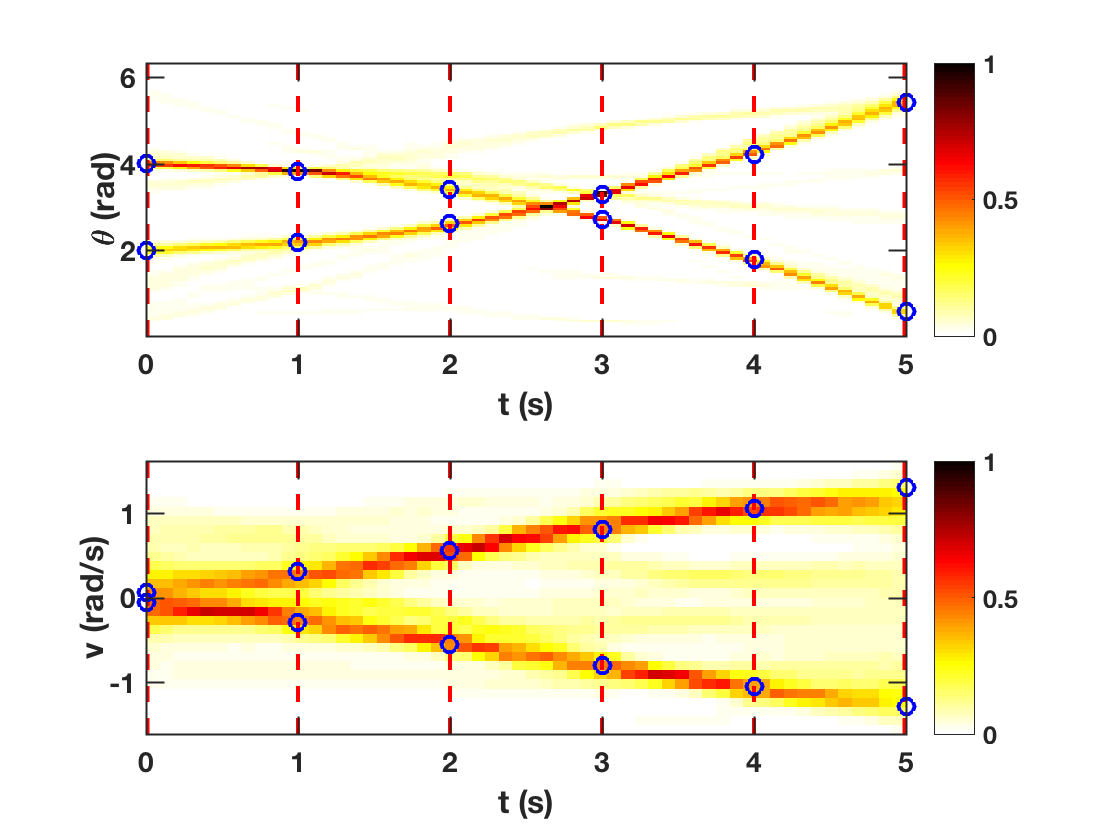}
           \caption{Reconstructed DoA and velocity spectra using a dynamic cost function in the OMT formulation, based on six observations of the array covariance matrix. The observation times are marked in dashed lines. The ground truth trajectories at the observation times are marked by rings. Top panel: reconstructed DoA spectrum. Bottom panel: reconstructed velocity spectrum.}
            \label{fig:tracking_location_vel_dynamic}
\vspace{-2mm}\end{figure}

\subsection{Sensor fusion in 3-D - audio example} \label{subsec:3D_sensor_fusion}
In order to illustrate the applicability of the multi-marginal OMT formulation in \eqref{eq:multi_margin_omt_PI}, in the form of \eqref{eq:multi_margin_omt_barycenter}, for performing sensor fusion, we consider a three-dimensional (3-D) localization problem, in which two sensor arrays observe two signal sources. The sources are modeled as localized speech sources, with the source signals being taken from babble noise excerpts from \cite{Audiotec}. The arrays consist of ten sensors each, arranged as a ULA for the first array, and as points on a circle, with two additional sensors perpendicular to the plane of the circle, for the second array.
The sensor spacing for the ULA is $0.1$ meters, and the radius for the circular array is $0.25$ meters.
The scenario is shown in Figures~\ref{fig:simulated_real_misaligned} and \ref{fig:simulated_real_misaligned_scenario_proj}, with Figure~\ref{fig:simulated_real_misaligned} displaying the full scenario, and Figure~\ref{fig:simulated_real_misaligned_scenario_proj} the scene as seen from above. 
The measured sensor signals are generated using acoustic impulse responses obtained from the randomized image method \cite{DeSenaAMW15_23}, using the room dimensions $4.3 \times 6.9 \times 2.6$ meters.  As the sources are broadband, we consider processing in the short-time Fourier transform (STFT) domain, i.e., equivalent to narrow-band filtering of the signal. Specifically, the signal is sampled at 16 kHz, and the STFT representation consists of 256 frequency bins, where each frame is constructed using 16 ms of the signal, using a Hann window with 50\% overlap.
We then compute estimates of the covariance matrices for the respective arrays corresponding to the frequency 2437.5 Hz, i.e., for the wavelength $0.1395$ meters, using the sample covariance estimate constructed from 100 signal snapshots. The proposed barycenter method in \eqref{eq:multi_margin_omt_barycenter} is then used to form a joint spatial spectral estimate using the two estimated covariance matrices, with the cost function $\costfunc(x_0,x_1) = \norm{x_0-x_1}_2^2$, for $x_0,x_1\in \RR^3$. For the discrete implementation in \eqref{eq:multimargin_generalcost}, we use a uniform gridding of the cube $[1.5,3.5]^3$, using $n = 75$ points in each dimension. The regularization parameters are $\gamma = 1$, common for all marginals, and $\epsilon = 10^{-2}$. Also, in order to illustrate the geometrical properties of the proposed formulation, we assume that the geometry of the ULA is only approximately known. Specifically, the assumed array geometry corresponds to a rotation in the $x-y$ plane of the true array, as shown in Figure~\ref{fig:simulated_real_misaligned_scenario_proj}. The obtained spectral estimate is superimposed in Figure~\ref{fig:simulated_real_misaligned}. As can be seen, the obtained estimate implies a spatial translation of the actual sources. However, it may be noted that the obtained estimate clearly identifies two targets, i.e., the error in the array orientation only results in a spatial perturbation, but no artifacts in the form of, e.g., spurious sources. This is also illustrated in Figure~\ref{fig:simulated_real_misaligned_proj}, showing the projection of the three-dimensional spectrum onto the two dimensional subspaces. As can be seen from the second and third panel, the position in the $z$-coordinate is unbiased, as the rotation only takes place in the $x-y$ plane.
%
%
\begin{figure}[t]
        \centering
            \includegraphics[width=.5\textwidth]{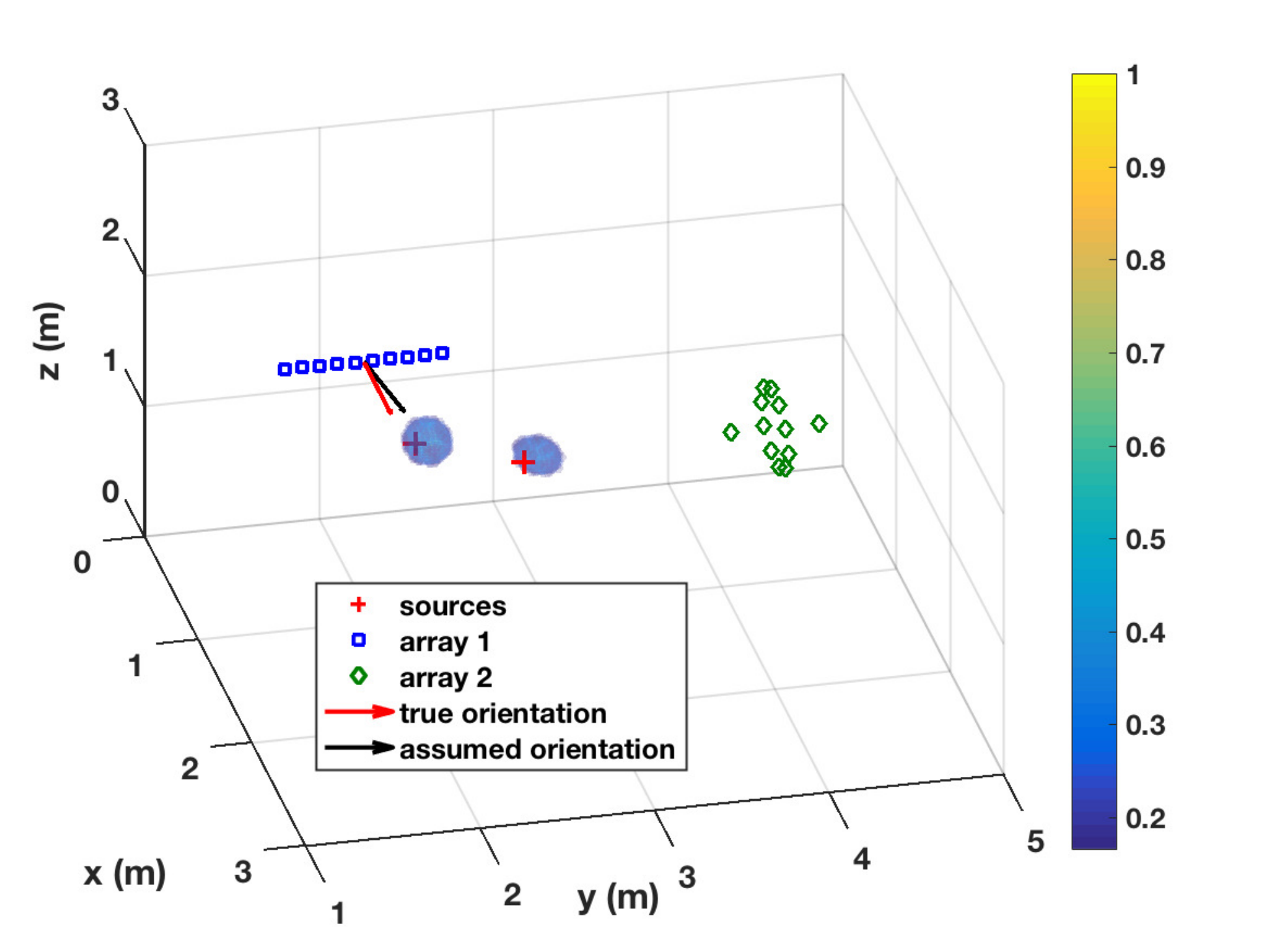}
           \caption{Scenario for localization in 3D using two sensor arrays and two signal sources. The true and assumed orientations of the first array are indicated by arrows. The three dimensional spectral estimate obtained using the formulation in \eqref{eq:multi_margin_omt_barycenter}  is superimposed.}
            \label{fig:simulated_real_misaligned}
\vspace{-2mm}\end{figure}
%
%
\begin{figure}[t]
        \centering
            \includegraphics[width=.5\textwidth]{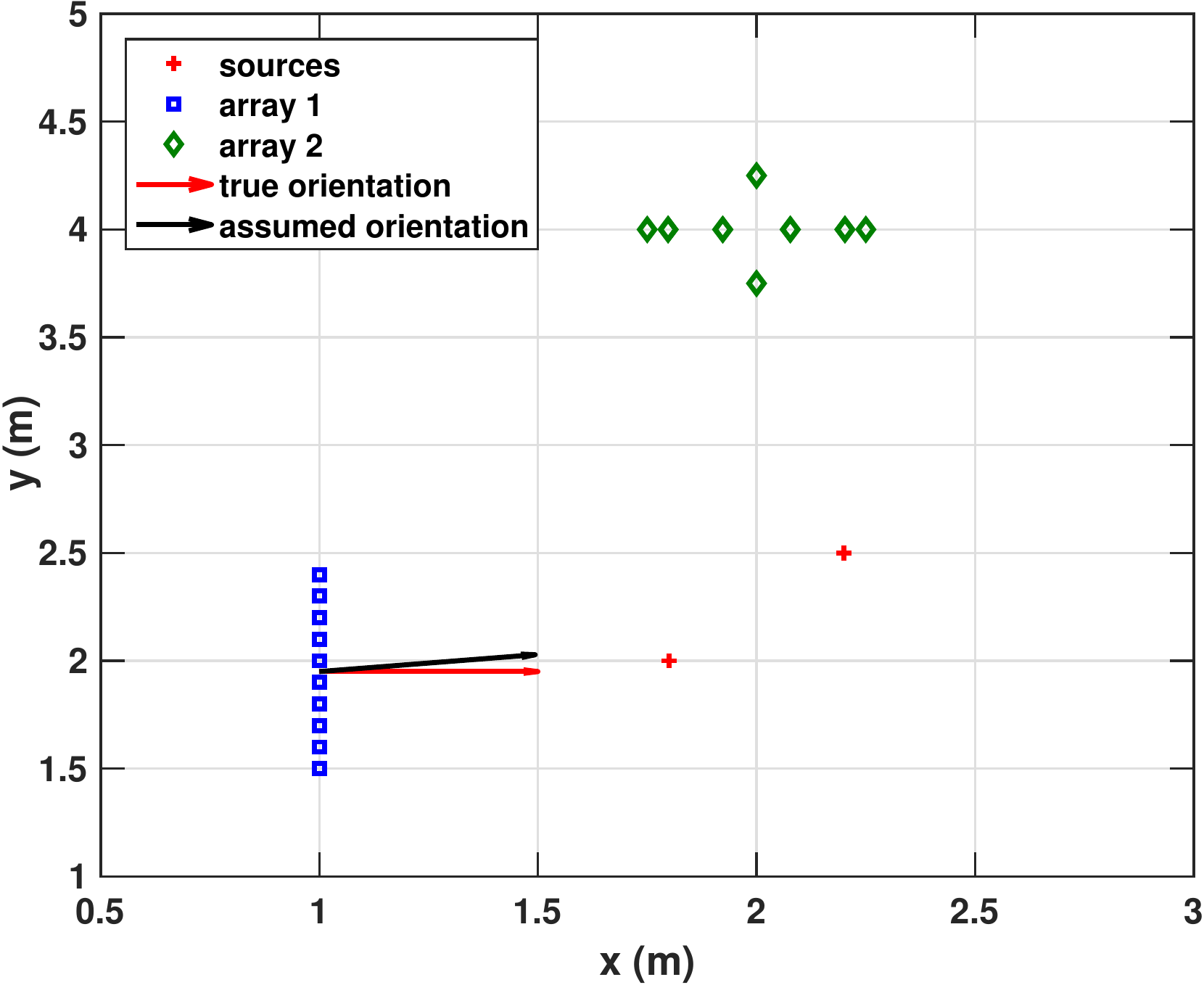}
           \caption{Birds eye view of the localization scenario in Figure~\ref{fig:simulated_real_misaligned}.}
            \label{fig:simulated_real_misaligned_scenario_proj}
\vspace{-2mm}\end{figure}
%
\begin{figure}[t]
        \centering
            \includegraphics[width=.5\textwidth]{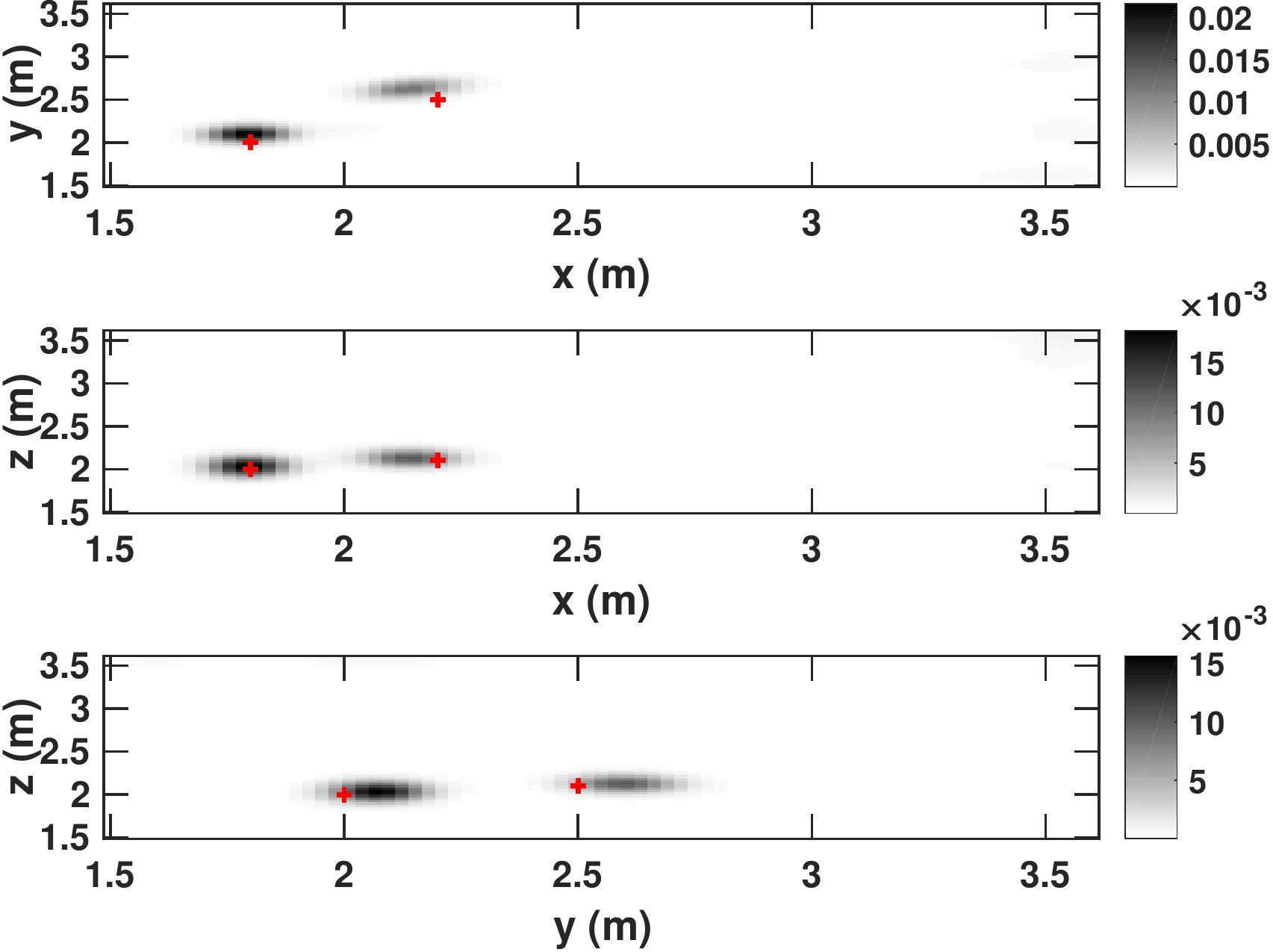}
           \caption{Two-dimensional projections of the three-dimensional spectral estimate obtained using the formulation in \eqref{eq:multi_margin_omt_barycenter}  for the estimation scenario in Figure~\ref{fig:simulated_real_misaligned}.}
            \label{fig:simulated_real_misaligned_proj}
\vspace{-2mm}\end{figure}

\subsection{Sensor fusion 2-D - simulation study}\label{subsec:2D_sensor_fusion_simulation}
%
%
\begin{figure}[t]
        \centering
        \vspace{1mm}
            \includegraphics[width=0.5\textwidth]{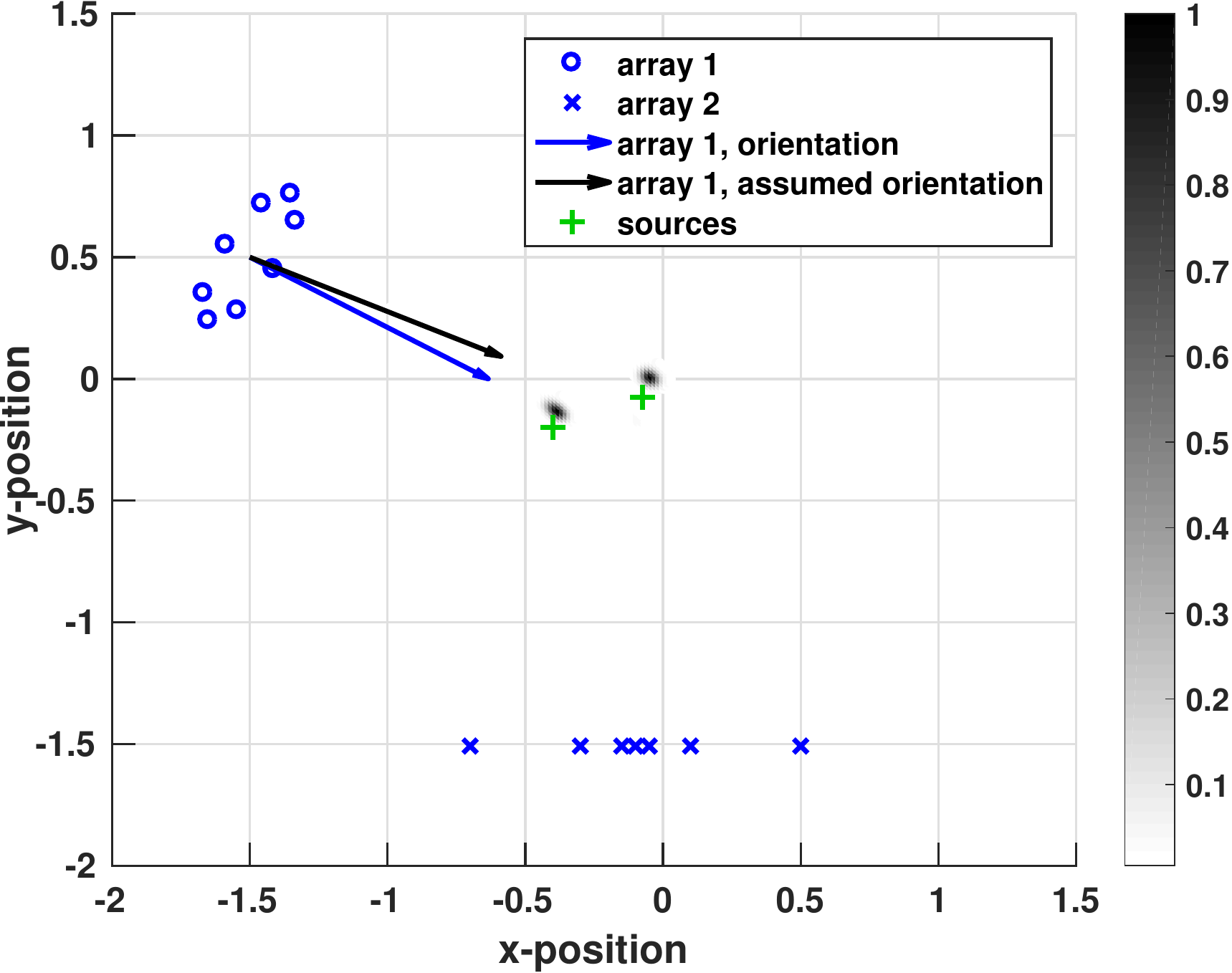}
            \vspace{.75mm}
           \caption{Spectral estimate as given by the multi-marginal barycenter formulation in \eqref{eq:multi_margin_omt_barycenter}. The alignment error is 6.7 degrees.}
            \label{fig:sensor_fusion_2D_OMT_multi_marginal}
\end{figure}
%
\begin{figure}[t]
        \centering
        \vspace{1mm}
            \includegraphics[width=0.5\textwidth]{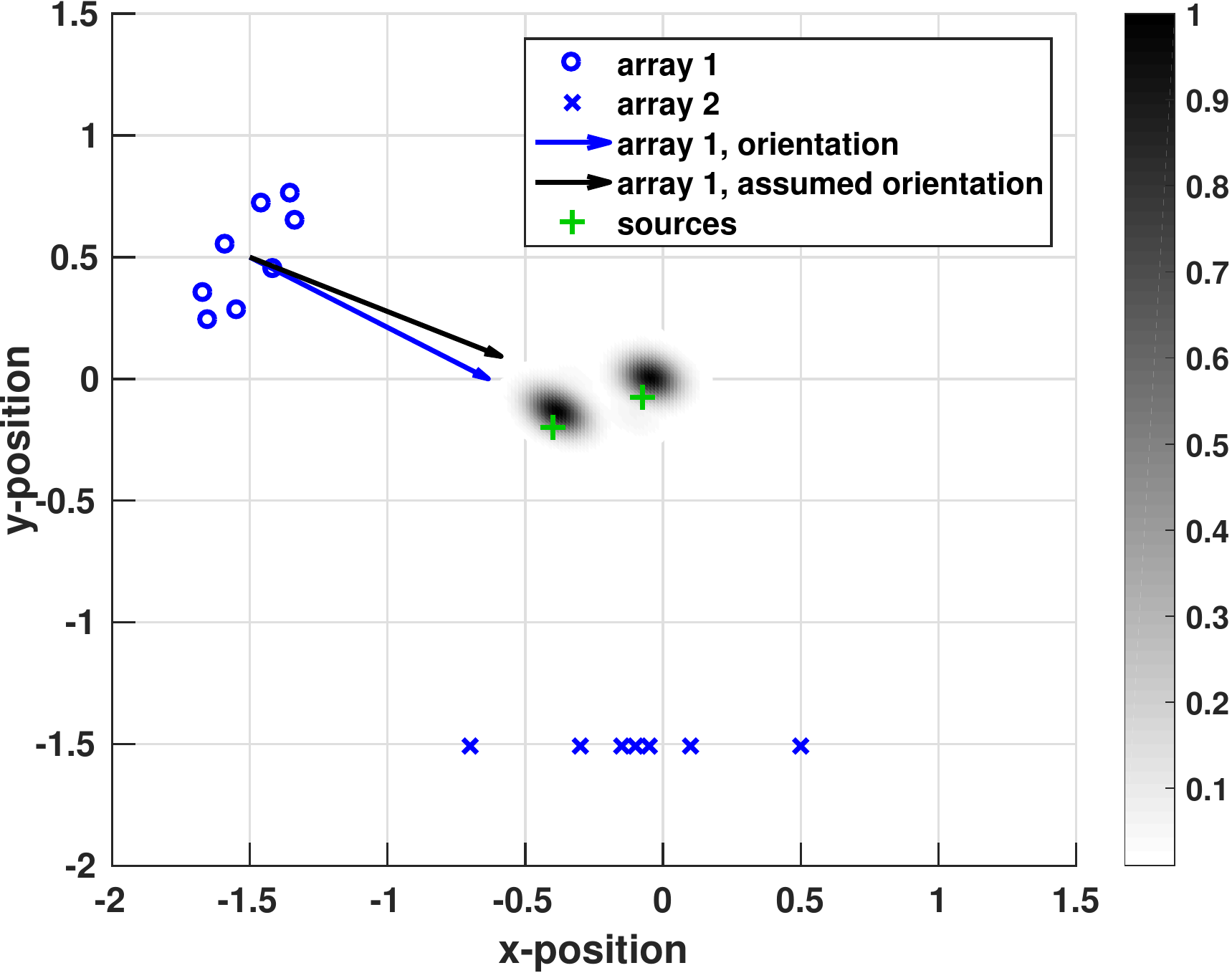}
            \vspace{.75mm}
           \caption{Spectral estimate as given by the pair-wise barycenter formulation in \eqref{eq:pair_wise_barycenter}. The alignment error is 6.7 degrees.}
            \label{fig:sensor_fusion_2D_OMT}
\end{figure}
%
%
\begin{figure}[t]
        \centering
        \vspace{1mm}
            \includegraphics[width=0.5\textwidth]{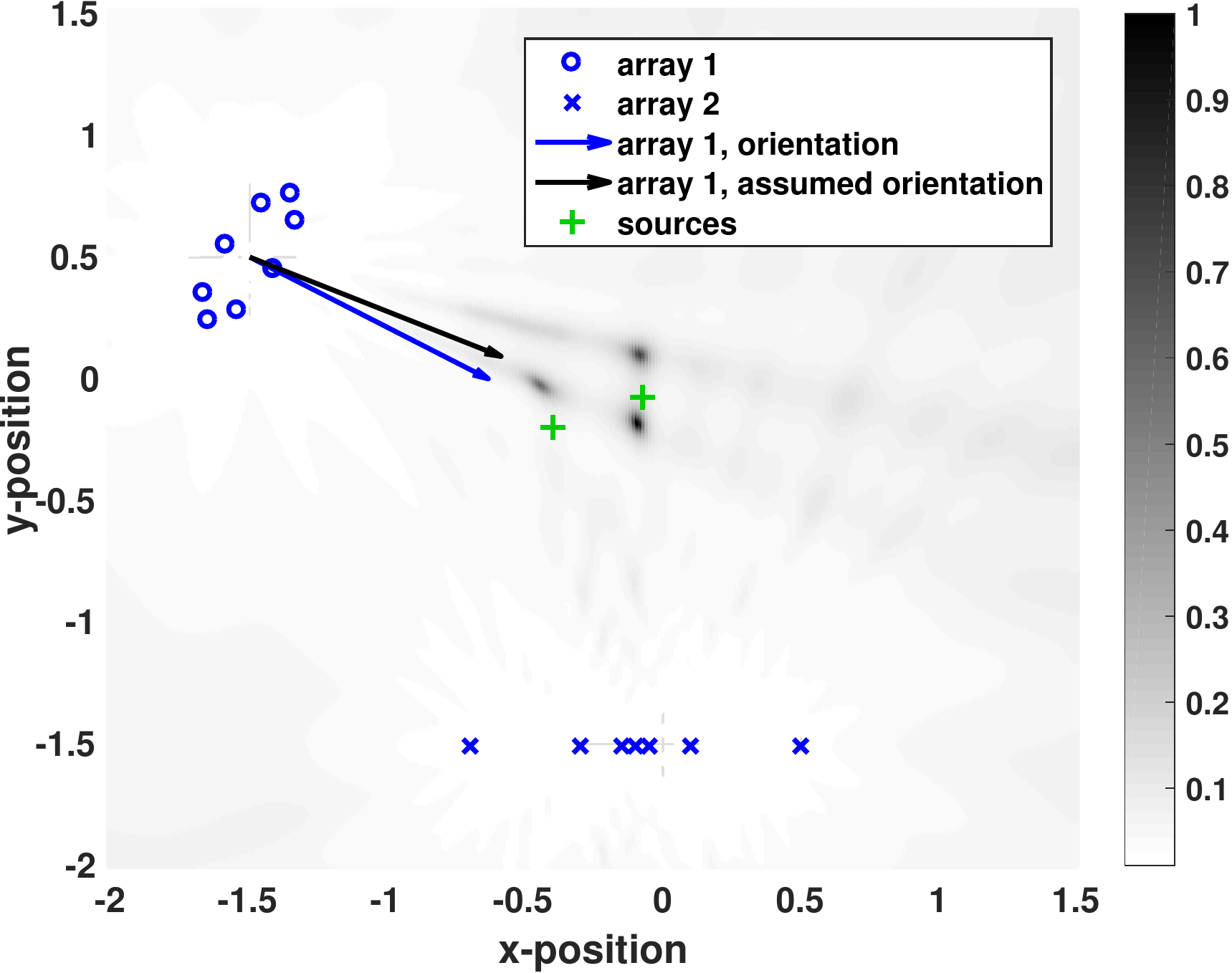}
            \vspace{.75mm}
           \caption{Pseudo-spectrum as given by non-coherent MUSIC. The alignment error is 6.7 degrees.}
            \label{fig:sensor_fusion_2D_music}
\end{figure}
%
\begin{figure}[t!]
        \centering
        \vspace{1mm}
            \includegraphics[width=0.5\textwidth]{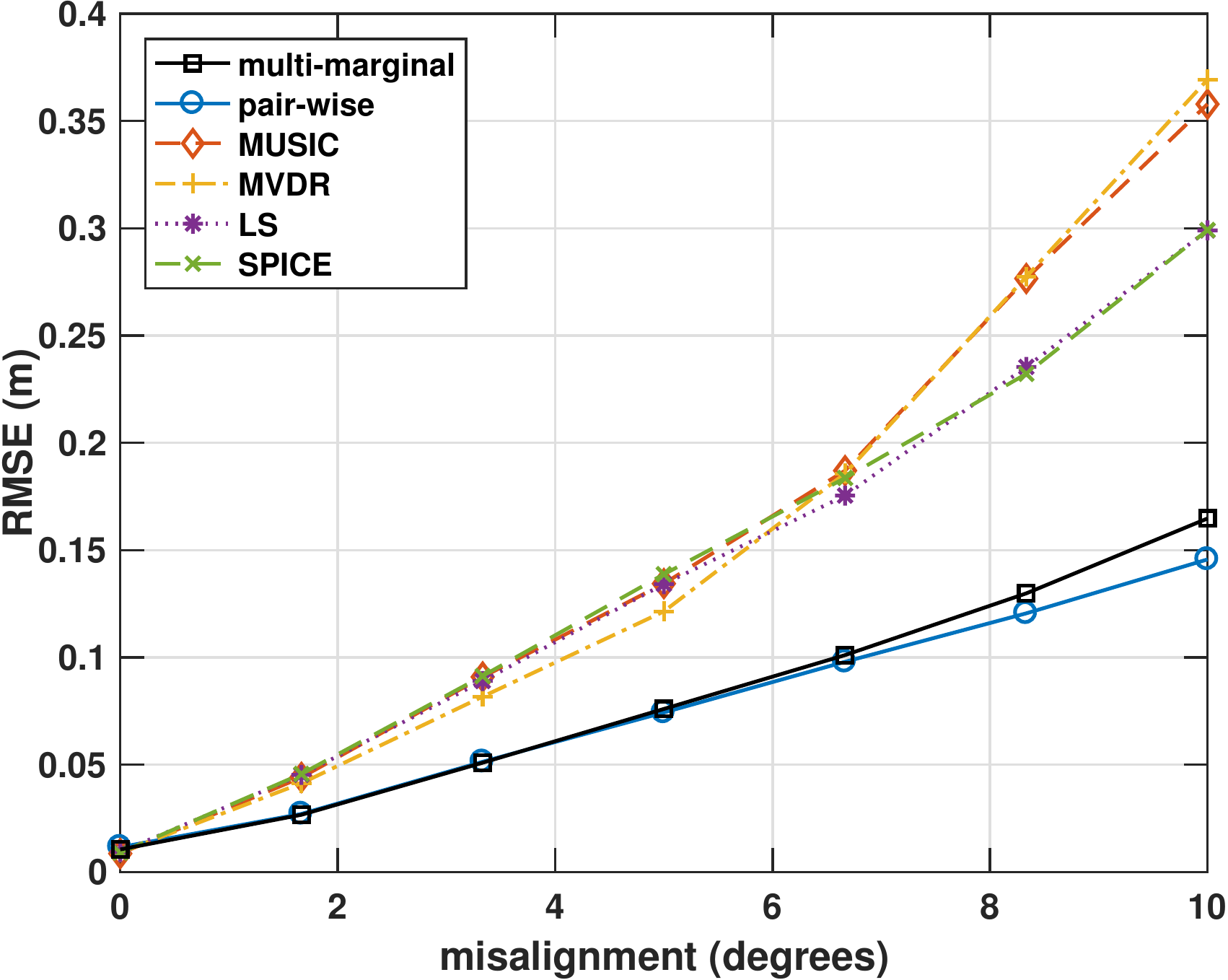}
            \vspace{.75mm}
           \caption{Error in location of spectral peaks, as function of the misalignment angle.}
            \label{fig:sensor_fusion_2D_simulation}
\end{figure}
%
%
%
%
As illustrated above, the proposed multi-marginal formulation is able to produce easily interpretable spectral estimates when used in 3-D sensor fusion scenarios, despite having erroneous information of the array geometry. Elaborating on this, we here conduct a Monte Carlo simulation study as to investigate the behavior of the spectral barycenter as a function of the array alignment error, and compare to other spectral estimation methods applicable to scenarios in which only the covariance matrices of individual sensor arrays, but not the inter-array covariances, are available. Specifically, consider a 2-D localization scenario, in which two uncorrelated sources impinge on two sensor arrays; one ellipsoidal shaped array consisting of 8 sensors, and one linear array consisting of 7 sensors. The wavefronts are here modeled as being circular. The scenario is shown in Figure~\ref{fig:sensor_fusion_2D_OMT_multi_marginal}. As may be seen, similar to the three-dimensional example, an unknown rotation is introduced to the ellipsoidal array, here varying this rotation between 0 and 10 degrees.
In the Monte Carlo simulation, 100 realizations are generated for each considered rotation angle. In each realization, the locations of the signal sources are randomized uniformly on the square $[-0.5,0.5]\times[-0.5,0.5]$. The sources are modeled as uncorrelated circularly symmetric Gaussian white noises with variance 100, and a spatially white noise with variance 1 is added to the sensor measurements. The wavelength of the impinging waves is twice that of the smallest sensor spacing in the linear array. We here consider spectral estimates obtained by the multi-marginal barycenter formulation in \eqref{eq:multi_margin_omt_barycenter}, as approximated by the discretization in \eqref{eq:multimargin_generalcost}. We also consider the formulation in \eqref{eq:pair_wise_barycenter} which is discretized and where the entropy regularization is employed separately for each pair-wise transport plan (see \cite{ElvanderHJK19_icassp} for details on the implementation of this problem). For both methods, we utilize the cost function $\costfunc(\discstate_k,\discstate_\ell) = \norm{\discstate_k-\discstate_\ell}_2^2$, for grid points $\discstate_k,\discstate_\ell \in \RR^2$. We grid the square $[-1,1]^2$ uniformly using $n = 100$ points in each dimension, and the algorithm parameters are $\gamma = 0.01$, common for all marginals, and $\epsilon = 10^{-3}$ and $\epsilon = 5\cdot 10^{-3}$ for the multi-marginal and pair-wise formulations, respectively. The array covariance matrices are estimated using as the sample covariance matrix from 500 signal snapshots.
As comparison, we also apply the non-coherent MUSIC and MVDR estimators, as described in \cite{RiekenF04_54}, as well as the least-squares (LS) estimator from \cite{LuoYWH17_270}, and the non-coherent SPICE estimator from \cite{SuleimanPPZ18_66} to the estimated covariance matrices in order to obtain estimates of the spatial spectrum.

Figures~\ref{fig:sensor_fusion_2D_OMT_multi_marginal}, \ref{fig:sensor_fusion_2D_OMT}, and \ref{fig:sensor_fusion_2D_music} provide illustrations of the behavior of the considered methods for the case of 6.7 degrees misalignment, for the multi-marginal formulation, the pair-wise regularized barycenter, and the MUSIC estimator, respectively. Comparing the two barycenter formulations, it may be noted that the estimate provided by the multi-marginal representation is considerably more concentrated, as compared to the formulation with pair-wise regularization. This is due to the optimization problem in \eqref{eq:multimargin_generalcost} being more well-conditioned, allowing for smaller entropy regularization while still retaining numerical stability. It should however be noted that the values of the parameter $\epsilon$ does not have exactly the same meaning for the two formulations, i.e., the problems are not equivalent even for identical parameter values. It may also be noted that the spectral estimates obtained using both barycenter formulations only imply a spatial perturbation of the signal sources. In contrast, the pseudo-spectral estimates obtained using the MUSIC estimator contains spurious peaks, in addition to larger deviations from the true source locations.

The results from the full simulation study are shown in Figure~\ref{fig:sensor_fusion_2D_simulation}, displaying the root mean squared error (RMSE) for the deviation of the spectral peaks to the true source locations. It may here be noted that for the case of no misalignment, the barycenter formulations produce estimates deviating slightly more from the ground truth than the comparison methods. However, as the misalignment increases, the RMSE for the barycenters estimators increase more slowly, indicating a greater robustness. This is not unexpected, as the formulations in \eqref{eq:multi_margin_omt_barycenter} and \eqref{eq:pair_wise_barycenter} allows for perturbations on the underlying domain, whereas the optimization criteria for the LS and SPICE estimators are related to $L_2$-distances for the spectra.

\subsection{Tracking of barycenters} \label{subsec:2D_barycenter_tracking}
As an illustration of the barycenter tracking formulation in Section~\ref{subsec:barycenter_tracking}, we consider a 2-D localization scenario in which two ULAs, each consisting of 15 sensors, measure signals generated by three moving sources. The scenario is illustrated in Figure~\ref{fig:scenario_4D_barycenter_tracking}. Here, the start location of each source is indicated by an asterisk. As can be seen, the trajectories of the targets all intersect at some point. Note that one of the ULAs has been rotated slightly, as indicated by the orientational arrows, motivating the use of a barycenter formulation in order to form an estimate of the spatial spectrum. At $\cT+1$ time points $t$, $t = 0,1,\ldots,\cT$, with $\cT = 7$, we collect 100 signal snapshots which are used to estimate the covariance matrices of the two arrays using the sample covariance matrix. The target signals are modeled as independent Gaussian sources, and spatially white Gaussian noise is added to the sensors, with an SNR of 20 dB. In order to model the tracking part, we use a two-dimensional extension of the state space description utilized in Section~\ref{subsec:1D_tracking}, where each of the two spatial components is endowed with a velocity state. With this description, the barycenter tracking problem considers transport on a four-dimensional space, i.e., over location and velocity for each spatial dimension. As the cost of transport between the barycenter at the observation times and the corresponding observation marginals, we use the Euclidean cost, i.e., $\costfunc(\discstate_0,\discstate_1) = \norm{\discstate_0-\discstate_1}_2^2$ for $\discstate_0, \discstate_1 \in \RR^2$. In forming the spectral estimates, we use the discretized version of the problem in \eqref{eq:multi_margin_omt_barycenter_tracking} using the parameters $\epsilon = 0.2$, $\gamma = 0.5$, common for all marginals, and $\alpha = 1$. The spatial domain is gridded uniformly with $n_x = 75$ points in each dimension, and the velocity domain is gridded uniformly with $n_v = 30$ points in each dimension. 

The resulting estimated spatial spectra are shown in Figure~\ref{fig:barycenter_tracking_4D_location}. As can be seen, the modes of the estimate correspond well to the ground truth, taking into account that the array rotation prevents perfect estimates. Note here that the rotation only causes a slight shift in location of the sources, but no spurious estimates. The corresponding estimated velocity spectra are shown in Figure~\ref{fig:barycenter_tracking_4D_velocity}. As can be seen, the distribution over velocity details three distinct modes that remain fairly constant over time, corresponding well to the ground truth constant velocity.
As comparison, Figure~\ref{fig:barycenter_tracking_4D_location_capon} displays the results obtained using the non-coherent MVDR estimator from \cite{RiekenF04_54}. Note that this estimator does not take any time dependence into account, and instead forms estimates at each separate time point $t$ using the available pair of array covariance matrices. It may be noted that these estimates display less concentrated estimates, as well as big differences in spectral power.

Finally, in order to illustrate how mass is transported on the 2-D spatial domain, Figure~\ref{fig:mass_trajectory_4D} tracks the location of the spectral peak corresponding to the third target along the trajectory, identified by computing the bi-marginal projections corresponding to consecutive barycenters, as detailed in Proposition~\ref{prp:sequential_cost} and Remark~\ref{remark:bi_projection_barycenter_tracking}. Note here that the distribution of mass remains fairly concentrated throughout the trajectory.

\begin{figure}[t!]
        \centering
        \vspace{1mm}
            \includegraphics[width=0.5\textwidth]{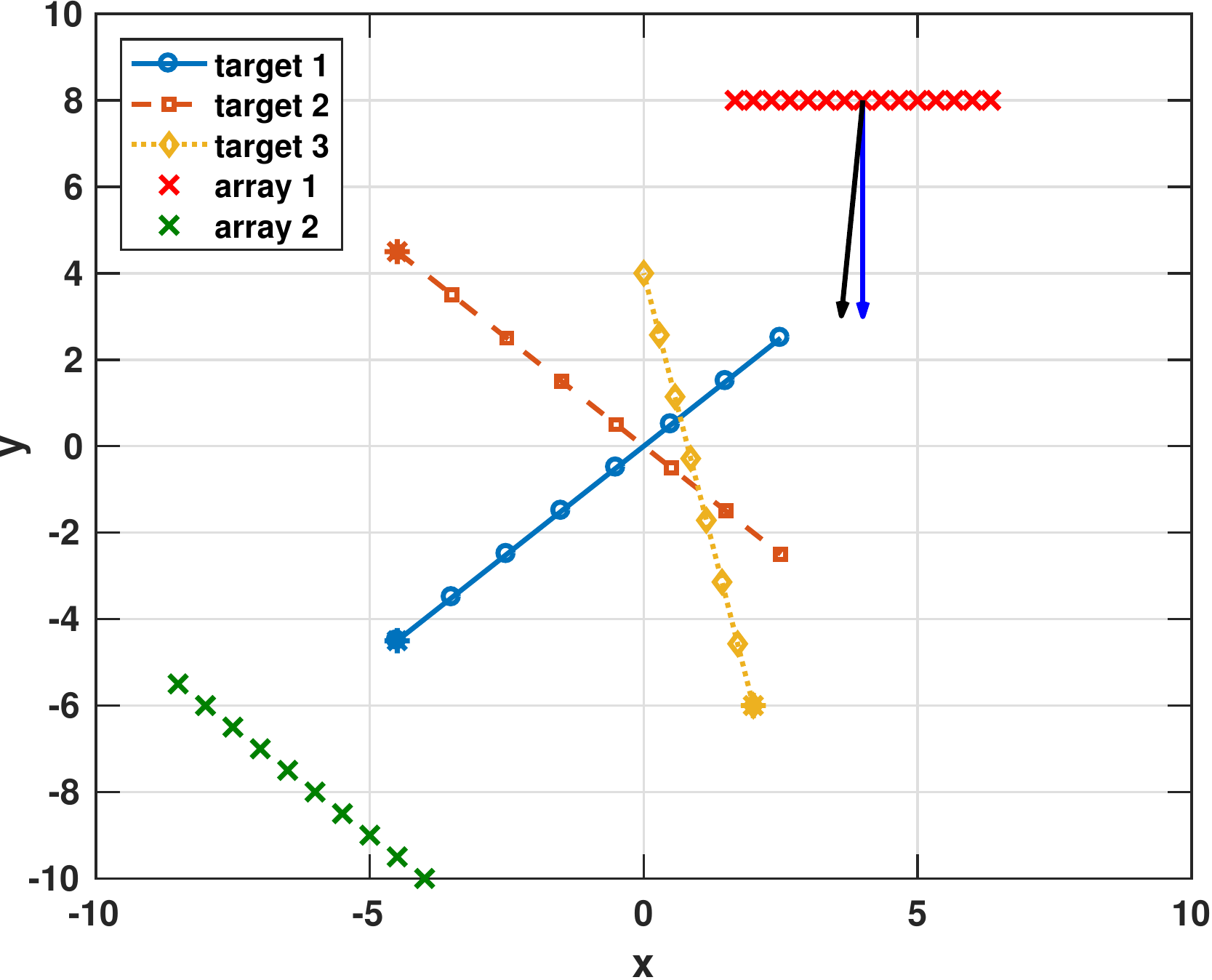}
            \vspace{.75mm}
           \caption{Scenario for the barycenter tracking problem with three moving sources. The start of the trajectory for each source is indicated by an asterisk $(*)$. Note here that one of the sensor arrays is slightly rotated, as indicated by the orientation arrows. }
            \label{fig:scenario_4D_barycenter_tracking}
\end{figure}

\begin{figure}[t!]
        \centering
        \vspace{1mm}
            \includegraphics[width=0.6\textwidth]{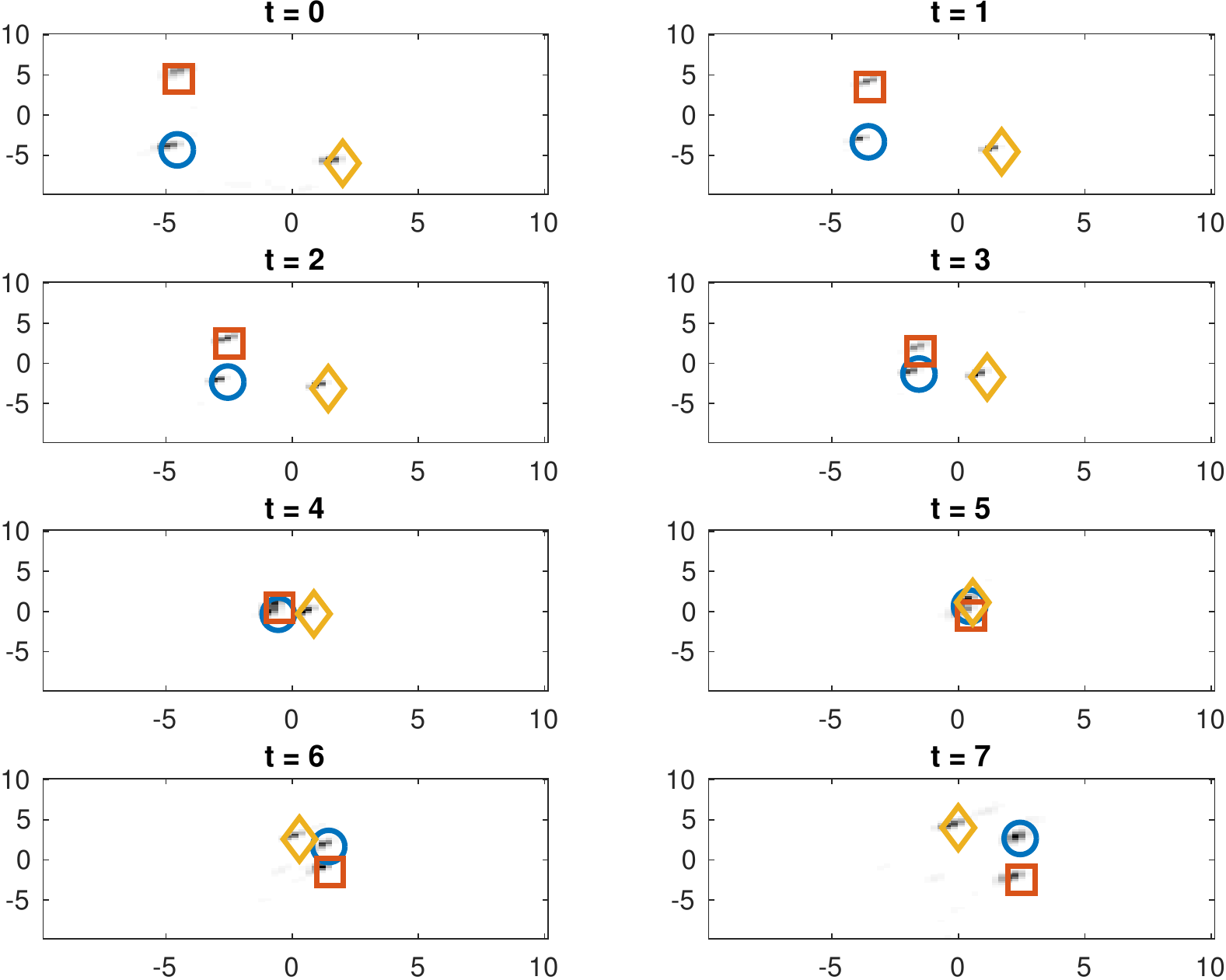}
            \vspace{.75mm}
           \caption{Estimated spatial spectra corresponding to the time points $t = 0,1,\ldots,\cT$, with $\cT = 7$, for the barycenter tracking problem. Note that the ground truth locations have been superimposed, using the same legend as in Figure~\ref{fig:scenario_4D_barycenter_tracking}.}
            \label{fig:barycenter_tracking_4D_location}
\end{figure}

\begin{figure}[t!]
        \centering
        \vspace{1mm}
            \includegraphics[width=0.6\textwidth]{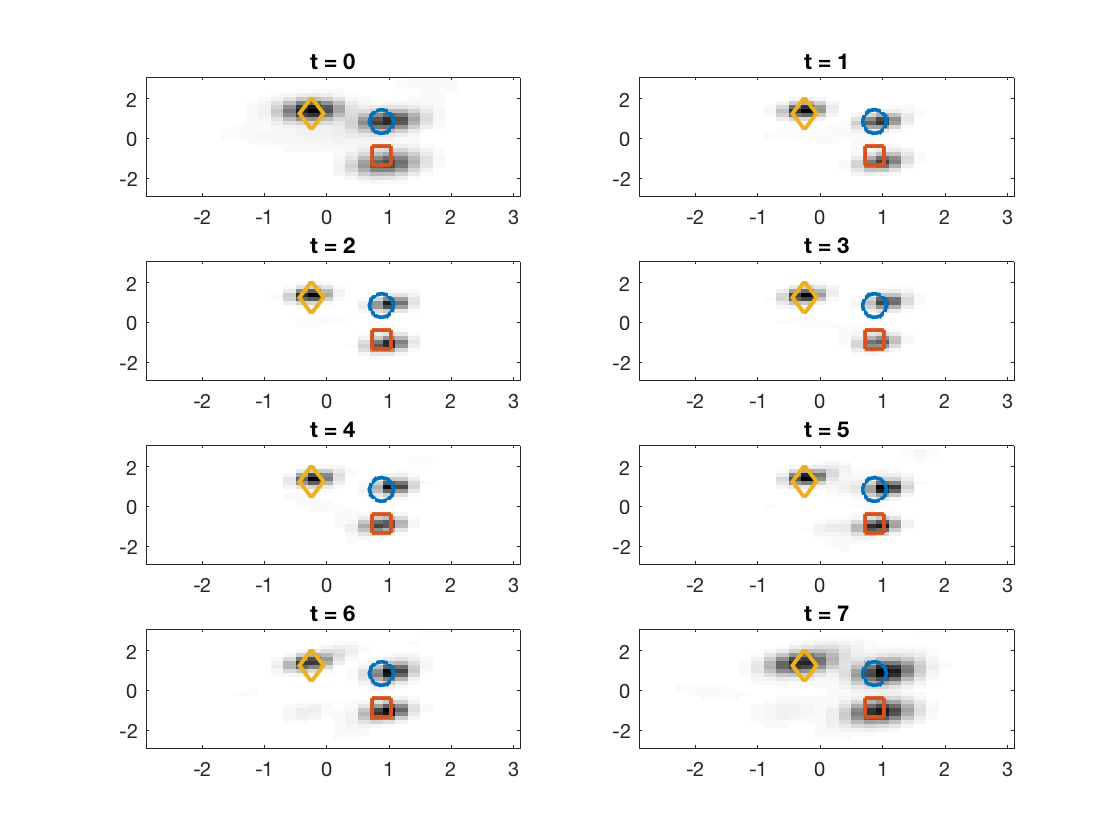}
            \vspace{.75mm}
           \caption{Estimated velocity spectra corresponding to the time points $t = 0,1,\ldots,\cT$, with $\cT = 7$, for the barycenter tracking problem. Note that the ground truth locations have been superimposed, using the same legend as in Figure~\ref{fig:scenario_4D_barycenter_tracking}.}
            \label{fig:barycenter_tracking_4D_velocity}
\end{figure}

\begin{figure}[t!]
        \centering
        \vspace{1mm}
            \includegraphics[width=0.6\textwidth]{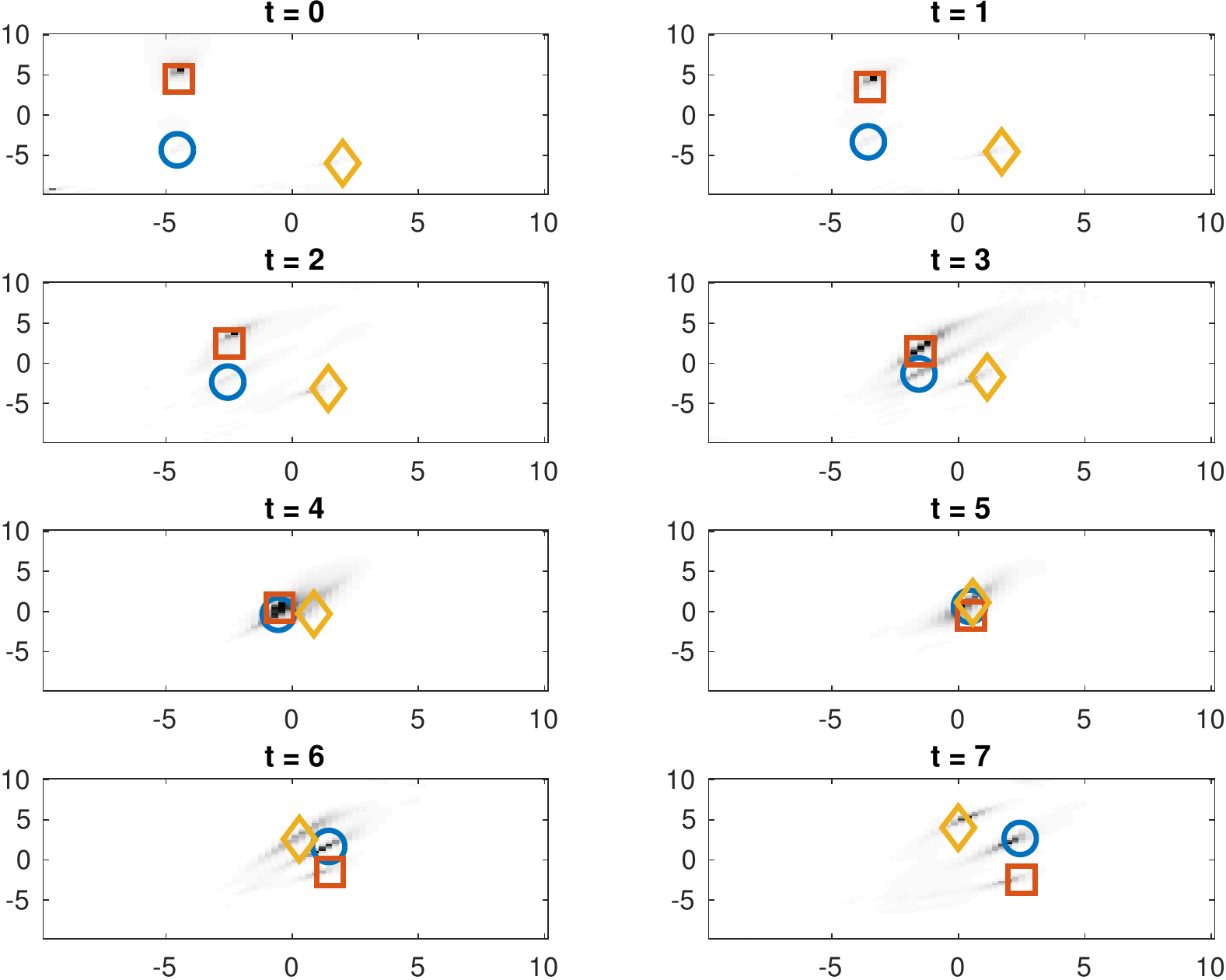}
            \vspace{.75mm}
           \caption{Estimated spatial spectra corresponding to the time points $t = 0,1,\ldots,\cT$, with $\cT = 7$, obtained using the non-coherent MVDR estimator. Note that the ground truth locations have been superimposed, using the same legend as in Figure~\ref{fig:scenario_4D_barycenter_tracking}.}
            \label{fig:barycenter_tracking_4D_location_capon}
\end{figure}

\begin{figure}[t!]
        \centering
        \vspace{1mm}
            \includegraphics[width=0.6\textwidth]{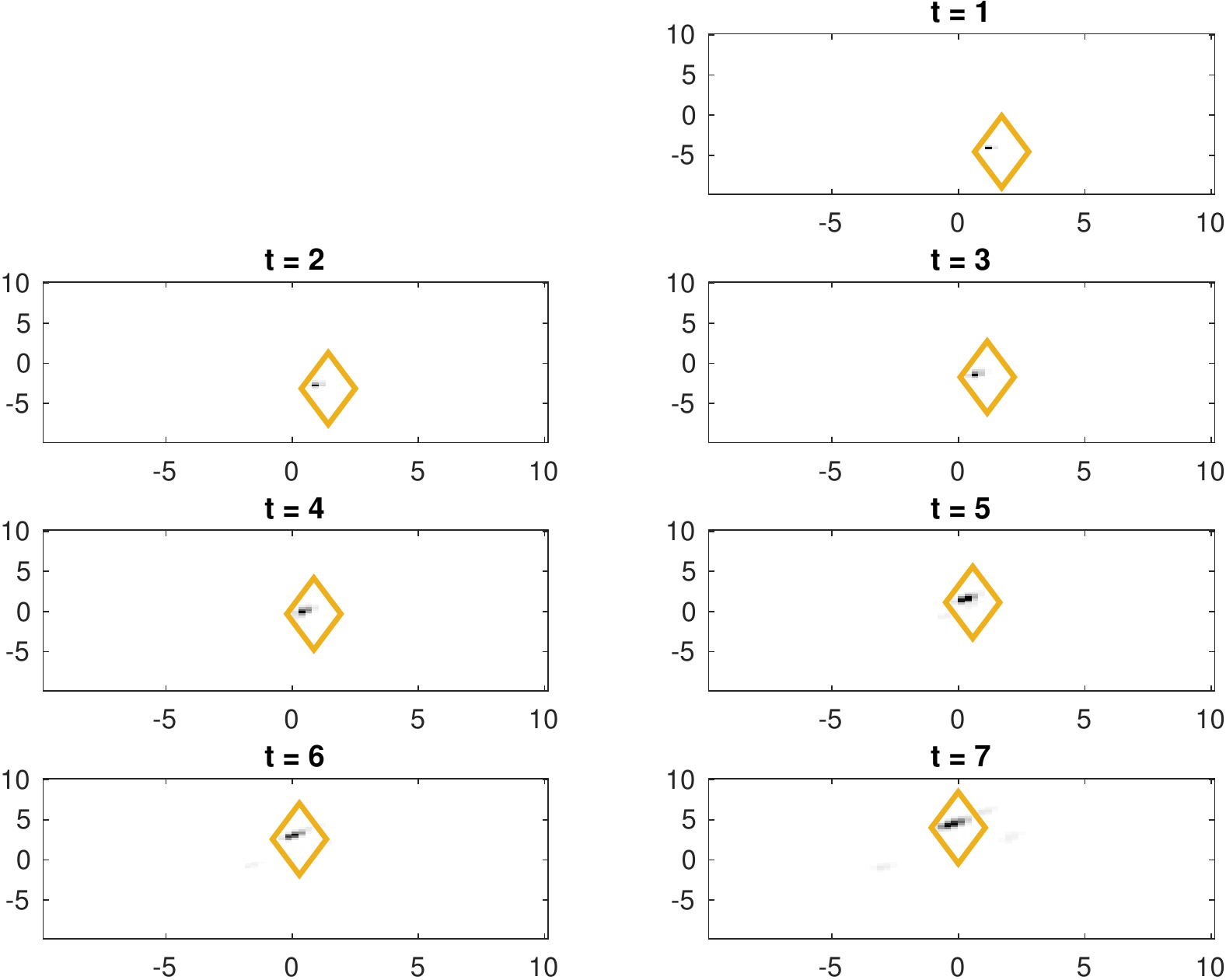}
            \vspace{.75mm}
           \caption{Tracking of the mass of the peak location corresponding to the third target at time $t = 0$.}
            \label{fig:mass_trajectory_4D}
\end{figure}

\subsubsection{Dimensionality and complexity analysis}
In this section, we provide a brief dimensionality analysis of the OMT problem used to model the barycenter tracking above, as well an analysis of the computational complexity of solving the dual problem in \eqref{eq:dual_multimarginal_omt} using Algorithm~\ref{alg:multitracking_sinkhornnewton}. Recall that the state space consists of a location and velocity component for each of the two spatial dimensions. Thus, the size of the discrete state space $X$ is $N= (n_x n_v)^d$, where $n_x=75$ and $n_v=30$ are the number of grid points on the spatial and velocity domain, respectively, and $d=2$ is the number of spatial dimensions. Accordingly, the size of the discrete measurement space $Y$ is $\tilde N = n_x^d$. As covariance matrices are estimated for $J = 2$ sensor arrays at $\cT +1$ time instances, with $\cT = 7$, this implies that the number of elements in the mass transport tensor $\bM$ in \eqref{eq:multimargin_generalcost} is
\begin{equation*}
N^{(\cT+1)} \tilde N^{(T+1)J} \approx 4.3 \cdot 10^{113}.
\end{equation*}
This number is larger than the number of particles in the observable universe; solving the primal OMT problem directly, or even constructing a tensor of this size, is thus infeasible. However, the dual problem in \eqref{eq:dual_multimarginal_omt} is formulated in terms of the dual vectors $\lambda_{(t,j)}\in\RR^{m}$, for $t=0,\dots,\cT,$ and  $j=1,\dots,J$. Here, $m = 225$ denotes the size of covariance vectors in \eqref{eq:cov_vector}, constructed by stacking the real and imaginary components of the corresponding covariance matrices, excluding the redundant parts resulting from the Hermitian structure. Thus, the number of real variables in the dual problem is given by
\begin{equation*}
(\cT+1)Jm = 3600,
\end{equation*}
constituting a dramatic complexity reduction as compared to the primal problem.

To analyze the computational complexity for solving the dual problem, we consider the complexity of performing one iteration sweep of Algorithm~\ref{alg:multitracking_sinkhornnewton}, i.e., one sweep through all index pairs $(t,j)$, for $t=0,\dots,\cT,$ and  $j=1,\dots,J$, as to update the corresponding dual vectors $\lambda_{(t,j)}$, where we run through $j$ in the inner cycle and $t$ in the outer cycle.
It may be noted that for each index pair $(t,j)$, this requires computing the vector $v_{(t,j)}$ in \eqref{eq:multimarginal_omt_partial_info_vt}, followed by solving \eqref{eq:tracking_sinkhorn_maximization0}, as detailed in Theorem~\ref{thm:multimarginal_scheme}.
The computation of $v_{(t,j)}$ is detailed in the projection in \eqref{eq:projection_tj} in Proposition~\ref{prp:combined_cost}, disregarding the vector $u_{(t,j)}$.
It may be noted that when cycling through $j = 1,2,\ldots,J$ for a given $t$, one does not have to re-compute $p_\tau$ for $\tau \neq t$, i.e., the multiplicative factors to the left and right of $p_t./(Ku_{(t,j)})$ in \eqref{eq:projection_tj} have to be updated only once for each $t$. 
Further, these factors may be stored so that updating them for a given $t$ requires only one multiplication with $K\in\RR^{N\times N}$. 

After updating the left and right factors in \eqref{eq:projection_tj}, computing the set of vectors $v_{(t,j)}$ for $j = 1,2,\ldots,J$, including the intermittent updates of $p_t$, requires $J(J+1)$ multiplications with $\tilde K \in \RR^{N \times n_x^2}$. By exploiting the structures of $K$ and $\tilde K$, as described in Remark~\ref{remark:tensor_mode_product}, i.e., the decoupling in the spatial dimensions, the total complexity for performing the update of the left and right factors and computing the set $\left\{ v_{(t,j)} \right\}_{j=1}^J$, is $\mathcal{O}(  d (n_x n_v)^3 +  J(J+1) d n_x^3 n_v)$. 
Finding the roots of \eqref{eq:tracking_sinkhorn_maximization0} by Newton's method requires solving a system of linear equations of size $m \times m$ in each inner Newton iteration. However, as indicated in Remark~\ref{remark:Newton}, after a few outer Sinkhorn iterations, Newton's method in general converges directly, i.e., it suffices to solve a single system of equations in each outer iteration.
Thus, after these initial outer iterations, the complexity for updating all dual variables $\lambda_{(t,j)}$, for $t = 0,1,\ldots,\cT$, and $j = 1,2,\ldots,J$, is 
\begin{equation*}
\mathcal{O}\left( \cT( d (n_x n_v)^3 +  J(J+1) d n_x^3 n_v + Jm^3) \right).
\end{equation*}
\section{Conclusions}\label{sec:conclusions}
In this work, we have proposed a framework for formulating multi-marginal OMT problems for scenarios in which the underlying mass distributions are only indirectly observable, referred to as partial information of the marginals. Motivated by examples from spatial spectral estimation in array processing, we have shown that the proposed formulations may be used for modeling information fusion, as well as for tracking the evolution of mass distributions over time. By leveraging the geometrical properties of OMT, the proposed formulations have been shown to yield robust spectral estimates, as well as allowing for exploiting prior knowledge of underlying dynamics. Also, we have presented computational tools, leading to computationally efficient solution algorithms for the transport problem. Even though the original primal OMT formulation may be prohibitively large, we have shown that by considering dual formulations, as well as exploiting inherent structures in the problem, one may arrive at tractable solvers even in high-dimensional settings.

\appendix

\section{Proofs}\label{appendix:proofs}
In this section, we provide proofs of Propositions~\ref{prp:sequential_cost} - \ref{prp:multimarginal_dual}, as well as for Theorem~\ref{thm:multimarginal_scheme}.
The proofs of Proposition~\ref{prp:sequential_cost} - \ref{prp:combined_cost}, for exploiting tensor structures in order to compute the projections on the marginals, are based on the following two lemmas.
\begin{lemma}\label{lm:projection}

Let $\bU=u_0 \otimes u_1\otimes \cdots \otimes u_\cT$ and suppose that
\[
\langle \bK,\bU\rangle=w_1^T \diag(u_t)w_2,
\]
where $w_1$ and $w_2$ are vectors that may depend on $u_\ell$, for $\ell\neq t$. Then,
\[
\projopdisc{t}{}(\bK \odot \bU) =w_1 \odot u_t \odot w_2.
\]

\end{lemma}

\begin{proof}
Let $e_{i_t}$ denote the unit vector whose element $i_t$ is equal to one, and with all other elements equal to zero. Then, by replacing $u_t$ with $u_t\odot e_{i_t}$ in the expression for $\langle \bK,\bU\rangle$, we get the element $i_t$ of the projection on the $t$:th marginal, i.e., 
\begin{align*}
&\langle \bK,u_0 \otimes \cdots \otimes u_{t-1} \otimes (u_{t}\odot e_{i_t}) \otimes u_{t+1} \otimes \cdots \otimes u_\cT\rangle\\
&= \sum_{\ell_0,\ell_1\ldots, \ell_\cT} \bK_{\ell_0,\ell_1\ldots, \ell_\cT} \left(\prod_{s=0}^{t-1}(u_s)_{\ell_s}\right)(u_{t}\odot e_{i_t})_{\ell_t}\left(\prod_{s=t+1}^\cT(u_s)_{\ell_s}\right)\\
&= \sum_{\substack{\ell_0,\ldots, \ell_{t-1} \\ \ell_{t+1},\ldots, \ell_\cT}} \bK_{\ell_0,\ell_1\ldots, \ell_\cT} \left(\prod_{s=0}^{t-1}(u_s)_{\ell_s}\right)(u_{t})_{i_t}\left(\prod_{s=t+1}^\cT(u_s)_{\ell_s}\right)\\
&=(\projopdisc{t}{}(\bK \odot \bU))_{i_t}.
\end{align*}
It may be noted that in the second line, a term in the sum is only non-zero if $i_t=\ell_t$, yielding the second equality.
By the assumption in the lemma, we get
\begin{align*}
(\projopdisc{t}{}(\bK \odot \bU))_{i_t}&=w_1^T \diag(u_t\odot e_{i_t})w_2 =(w_1)_{i_t} (u_t)_{i_t} (w_2)_{i_t},
\end{align*}
from which the result follows.

\end{proof}

\begin{lemma}\label{lm:coupling}

Let $\bU=u_0 \otimes u_1\otimes \cdots \otimes u_\cT$ and suppose that
\[
\langle \bK,\bU\rangle=w_1^T \diag(u_{t_1})W_2 \diag(u_{t_2})w_3,
\]
where $w_1, w_3\in \RR^n$ and $W_2\in \RR^{n\times n}$ may depend on $u_\ell$, for $\ell\notin \{t_1,t_2\}$. Then,
\[
\projopdisc{t_1,t_2}{}(\bK \odot \bU) =\diag(w_1 \odot u_t) W_2 \diag(u_{t_2}\odot w_3 ).
\]

\end{lemma}

\begin{proof}
Analogously to the proof in Lemma~\ref{lm:projection}, we may express the bi-marginal projections as
\[
(\projopdisc{t_1,t_2}{}(\bK \odot \bU))_{i_{t_1},i_{t_2}}=\langle \bK,u_0 \otimes \cdots \otimes u_{{t_1}-1} \otimes (u_{{t_1}}\odot e_{i_{t_1}}) \otimes u_{{t_1}+1} \otimes \cdots \otimes u_{{t_2}-1} \otimes (u_{{t_2}}\odot e_{i_{t_2}}) \otimes u_{{t_2}+1}\otimes \cdots \otimes u_\cT\rangle.
\]
Hence, by the assumption,
\begin{align*}
(\projopdisc{t_1,t_2}{}(\bK \odot \bU))_{i_{t_1},i_{t_2}}&=w_1^T \diag(u_{{t_1}}\odot e_{i_{t_1}})W_2\diag(u_{{t_2}}\odot e_{i_{t_2}})w_3\\&=(w_1)_{i_{t_1}} (u_{t_1})_{i_{t_1}} (W_2)_{i_{t_1},i_{t_2}}(u_{t_2})_{i_{t_2}}(w_3)_{i_{t_2}},
\end{align*}
and thus the result follows.

\end{proof}

\begin{proof}[Proof of Proposition~\ref{prp:sequential_cost}.]

	First, note that due to the assumption of a sequential cost, we have
	\begin{equation*}
	\bK_{i_0,\ldots, i_T} = \prod_{t=1}^\cT K_{i_{t-1},i_t},
	\end{equation*}
	for the tensor $\bK=\exp(-\bC/\epsilon)$ and matrix $K=\exp(-C/\epsilon)$. Therefore,
\begin{align*}
	\langle \bK, \bU\rangle&=\sum_{i_0,i_1,\ldots, i_\cT} \left(\prod_{t=1}^\cT K_{i_{t-1},i_t}\right) \prod_{t=0}^\cT (u_t)_{i_t}\\
	&=\sum_{i_0,i_1,\ldots, i_\cT} (u_0)_{i_0} \prod_{s=1}^\cT (K \diag(u_t))_{i_{t-1},i_{t}}\\
	&=u_0^T K\diag(u_1) K \ldots K \diag(u_{\cT-1}) Ku_{\cT}\\
	&=\ett^T\diag(u_0) K\diag(u_1) K \ldots K \diag(u_{\cT-1}) K\diag(u_{\cT})\ett.
\end{align*}
Thus, $\langle \bK, \bU\rangle$ may be written as in Lemma~\ref{lm:projection} with
\begin{align*}
w_1&=\left(\ett^T\diag(u_0) K \diag(u_1) K \ldots K\diag(u_{t-1}) K\right)^T,\\ 
w_2&=K\diag(u_{t+1})K \ldots K \diag(u_{\cT-1}) K \diag(u_\cT)\ett,
\end{align*}
and hence, 
\begin{align*}
\projopdisc{t}{}(\bK \odot \bU)=&\left( u_0^T K \diag(u_1) K \dots K\diag(u_{t-1}) K \right)^T  \odot u_t \odot \left( K\diag(u_{t+1})K \ldots K \diag(u_{\cT-1}) K u_\cT \right).
\end{align*}
Moreover, in order to derive an expression for $\projopdisc{t_1,t_2}{}(\bK \odot \bU)$, note that $\langle \bK, \bU\rangle$ can be written as in Lemma~\ref{lm:coupling} with 
\begin{align*}
w_1&=\left(\ett^T\diag(u_0) K \diag(u_1) K \ldots K\diag(u_{t_1-1}) K\right)^T\\ 
W_2&=K\diag(u_{t_1+1})K \ldots K \diag(u_{t_2-1}) K \\
w_3&=K\diag(u_{t_2+1})K \ldots K \diag(u_{\cT-1}) K \diag(u_\cT)\ett,
\end{align*}
and hence the expression \eqref{eq:coupling_tracking} follows.

\end{proof}

\begin{proof}[Proof of Proposition \ref{prp:barycenter_cost}.] 
Due to the structure of the cost tensor $\bC$, we have
\begin{equation*}
\bK_{i_0,\ldots, i_J} = \prod_{j=1}^J K_{i_0,i_j},
\end{equation*}
	for the tensor $\bK=\exp(-\bC/\epsilon)$ and the matrix $K=\exp(-C/\epsilon)$. Thus, one may write
\begin{align*}
	\langle \bK, \bU\rangle=&  \sum_{ i_0,i_1,\dots, i_{J} } \left(\prod_{\ell=1}^J K_{i_0,i_\ell}\right) \prod_{\ell=0}^J (u_\ell)_{i_\ell}  \\
= & \sum_{ i_0} (u_0)_{i_0} \sum_{ i_1,\dots, i_{J} } \prod_{\ell=1}^J K_{i_0,i_\ell} (u_\ell)_{i_\ell}  \\
= & \sum_{ i_0} (u_0)_{i_0} \prod_{\ell=1}^J ( K u_\ell )_{i_0}\\
= & u_0^T \left(\bigodot_{\ell=1}^J  K u_\ell \right).
\end{align*}
It may be noted that this may be expressed as
\begin{align*}
	\langle \bK, \bU\rangle&=\ett^T \diag(u_0) \left(\bigodot_{\ell=1}^J  K u_\ell \right),
\end{align*}
as well as
\begin{align*}
	\langle \bK, \bU\rangle&=(Ku_j)^T \left(u_0 \odot \bigodot_{\substack{\ell=1\\\ell\neq j}}^J  K u_\ell \right)
=\ett^T \diag(u_j)K^T \left(u_0 \odot \bigodot_{\substack{\ell=1\\\ell\neq j}}^J  K u_\ell \right).
\end{align*}
Applying Lemma~\ref{lm:projection} yields the expressions \eqref{eq:proj_barycenter_0}
and \eqref{eq:proj_barycenter_j}
for $\projopdisc{j}{}( \bK \odot \bU)$ for $j=0,\ldots, J$. Alternatively, one may rewrite $\langle \bK, \bU\rangle$ as
\begin{equation*}
	\langle \bK, \bU\rangle=\ett^T \diag(u_0)\, \diag\left( \bigodot_{\substack{\ell=1\\\ell\neq j}}^J  K u_\ell \right)K\diag(u_{j})\ett,
\end{equation*}
and
\begin{equation*}
\langle \bK, \bU\rangle =\ett^T \diag(u_{j_1})K^T \diag\left(u_0 \odot \bigodot_{\substack{\ell=1\\\ell\neq j_1,j_2}}^J  K u_\ell \right)K\diag(u_{j_2})\ett.
\end{equation*}
With Lemma~\ref{lm:coupling}, this leads to the expressions \eqref{eq:proj_barycenter_0j}
and \eqref{eq:proj_barycenter_jj}
for the couplings  $\projopdisc{j_1,j_2}{}( \bK \odot \bU)$ for $j_1,j_2=0,\ldots, J$.

\end{proof}

\begin{proof}[Proof of Proposition \ref{prp:combined_cost}.] 
 
Recall the definition of the set $\Lambda=\{(t,j)\,|\, t\in\{0,1,\ldots, \cT\}, j\in\{0,1\,\ldots,J\}\}$. The structure of the cost tensor $\bC$ then implies that each element of the tensor $\bK=\exp(-\bC/\epsilon)$ may be expressed as
\begin{equation*}
\bK_{\left(i_{(t,j)}|(t,j)\in \Lambda\right)} = \left(\prod_{t=1}^\cT K_{i_{(t-1,0)}i_{(t,0)}}\right) \prod_{t=0}^\cT \prod_{j=1}^J \tilde K_{i_{(t,0)}i_{(t,j)}},
\end{equation*}
with the matrices defined as $K = \exp( -C/\epsilon) $ and $\tilde K = \exp(-\tilde C/\epsilon)$.
Furthermore, the elements of the tensor $\bU$ are given by
\begin{equation*}
\bU_{\left(i_{(t,j)}|(t,j)\in \Lambda\right)} = \prod_{t=0}^\cT \prod_{j=0}^J (u_{(t, j)})_{i_{(t, j)}}.
\end{equation*}
Therefore,
\begin{align*}
\langle \bK, \bU\rangle=&  \sum_{ \substack{ i_{(s,\ell)}\\ (s,\ell)\in\Lambda } } \left(\prod_{t=1}^\cT K_{i_{(t-1,0)}i_{(t,0)}}\right)\left( \prod_{t=0}^\cT \prod_{j=1}^J \tilde K_{i_{(t,0)}i_{(t,j)}} \right)\prod_{t=0}^\cT \prod_{j=0}^J (u_{(t, j)})_{i_{(t, j)}}\\
=& \sum_{ \substack{ i_{(s,0)}\\ s=0,1,\ldots, \cT} }  \left(\prod_{t=1}^\cT K_{i_{(t-1,0)}i_{(t,0)}}\right) \prod_{t=0}^\cT \left(\sum_{ \substack{ i_{(t,\ell)}\\ \ell=1,2,\ldots, J } }  \left(\prod_{j=1}^J \tilde K_{i_{(t,0)}i_{(t,j)}}\right)  \prod_{j=0}^J (u_{(t, j)})_{i_{(t, j)}} \right).
\end{align*}
Note that the last sum is a projection of the type \eqref{eq:proj_barycenter_0} in Proposition~\ref{prp:barycenter_cost}, and we denote it as
\begin{align*}
 \sum_{\substack{i_{(t, \ell)} \\ \ell=1,2,\ldots, J}} \left(\prod_{j=1}^J \tilde K_{i_{(t, 0)}i_{(t, j)}}\right) \prod_{j=0}^J (u_{(t, j)})_{i_{(t, j)}}  
  =& \left( u_{(t, 0)} \odot  \bigodot_{j=1}^J \left( \tilde K u_{(t, j)} \right) \right)_{i_{(t, 0)}} = (p_t)_{i_{(t,0)}}.
\end{align*}
Hence, using the proof of Proposition~\ref{prp:sequential_cost}, one may write
\begin{align*}
\langle \bK, \bU\rangle&=  p_0^T K \diag(p_1) K   \ldots K \diag(p_{\cT-1}) K p_\cT\\
&=  p_0^T K    \ldots K\diag(p_{t-1})K \diag(u_{(t, 0)}) \diag\left(   \bigodot_{j=1}^J \left( \tilde K u_{(t, j)} \right) \right) K \diag(p_{t+1})K\ldots  K p_\cT
\end{align*}
Thus, the projections on the central marginals, corresponding to index $(t,0)$ for  $t = 0,1,\ldots,\cT$, may be computed as in Proposition~\ref{prp:sequential_cost}, yielding 
\begin{equation*}
\projopdisc{(t,0)}{}(\bK \odot \bU)=\left( p_0^T K \diag(p_1) K \dots \diag(p_{t-1}) K \right)^T  \odot p_t \odot \left( K\diag(p_{t+1}) \ldots K \diag(p_{\cT-1}) K p_\cT \right).
\end{equation*}
For the other marginals, i.e., for $(t,j)$ such that $j > 0$, one may express $\langle \bK, \bU\rangle$ similarly as in the proof of \eqref{eq:proj_barycenter_j} of Proposition~\ref{prp:barycenter_cost}, i.e.,
\begin{align*}
\langle \bK, \bU\rangle =& \ett^T \diag( u_{(t,j)}) \tilde K^T \bigg(( p_0^T K \diag(p_1) K \dots \diag(p_{t-1}) K )^T \\
	&\qquad\odot (p_t./(Ku_{(t,j)}) ) \odot ( K\diag(p_{t+1}) \ldots K \diag(p_{\cT-1}) K p_\cT )\bigg),
\end{align*}
which, with Lemma~\ref{lm:projection}, yields the expression \eqref{eq:projection_tj}.

\end{proof}

\begin{proof}[Proof of Proposition \ref{prp:multimarginal_dual}.] 
As to simplify the exposition, let $\lambda$ and $\Delta$ denote the sets $\lambda_0,\ldots,\lambda_\cT$ and $\Delta_0,\ldots,\Delta_\cT$, respectively. The Lagrangian of \eqref{eq:multimargin_generalcost}, with dual variable $\lambda$, is detailed as
\begin{equation}
\begin{split}
L(\bM,\Delta,\lambda) = & \sum_{i_0,\ldots, i_{\cT}} \bC_{i_0,\dots,i_\cT} \bM_{i_0,\ldots, i_\cT}+\epsilon \cD(\bM) \\ & \qquad + \sum_{t=0}^\cT\gamma_t \|\Delta_t\|_2^2 + \sum_{t=0}^\cT\lambda_t^T(r_t+\Delta_t- \disccovop_t \projopdisc{t}{}(\bM)).
\end{split} \label{eq:multimarginal_lagrangian}
\end{equation}
For fixed $\lambda$, the Lagrangian is minimized when the gradients with respect to $\bM$ and $\Delta$ vanish.
For $\bM$, this requires
\begin{align*}
& \bC_{i_0,\dots,i_\cT} +\epsilon \log \bM_{i_0,\ldots, i_\cT} -\sum_{t=0}^\cT \left(\lambda_t^T \disccovop_t \right)_{i_t} = 0 \\
 \Leftrightarrow \ & \bM_{i_0,\ldots, i_\cT} = \exp(-\bC_{i_0,\dots,i_\cT}/\epsilon) \prod_{t=0}^\cT \exp( \disccovop_t^T \lambda_t / \epsilon ).
\end{align*}
Thus, one may express the mass transport tensor as $ \bM = \bK \odot \bU$ for two sensors $\bK, \bU \in \RR^{n^{\cT+1}}$, defined as $\bK=\exp(-\bC/\epsilon)$ and
\begin{equation*}
\bU = u_0\otimes u_1\otimes \cdots \otimes u_\cT, \quad \text{ with } u_t = \exp(\disccovop_t^T \lambda_t / \epsilon).
\end{equation*}
This proves the first statement of the proposition. Further, the Lagrangian \eqref{eq:multimarginal_lagrangian} is minimized with respect to $\Delta$ when
\begin{equation*}
2\gamma_t \Delta_t+\lambda_t= 0 \quad \Rightarrow \ \Delta_t=-\frac{1}{2\gamma_t} \lambda_t \quad \text{for } t=0,1,\ldots,\cT. \label{eq:tracking_gradient_delta}
\end{equation*}
Thus, the Langrangian \eqref{eq:multimarginal_lagrangian}, when minimized with respect to $\bM$ and $\Delta$, becomes
\begin{equation*}
\min_{\bM,\Delta} L(\bM,\Delta,\lambda) = \ -\epsilon \! \sum_{i_0,\ldots, i_{\cT}}  \bM_{i_0,\ldots, i_\cT}- \sum_{t=0}^\cT \left( \frac{1}{4\gamma_t}\|\lambda_t\|_2^2+ \lambda_t^Tr_t \right),
\end{equation*}
yielding the dual problem
\begin{equation*}
\maxwrt_{\lambda_0,\ldots, \lambda_\cT} \ -\epsilon \! \sum_{i_0,\ldots, i_{\cT}} \bK_{i_0,\ldots, i_\cT} \bU_{i_0,\ldots, i_\cT} -\sum_{t=0}^\cT\frac{1}{4\gamma_t} \|\lambda_t\|_2^2+ \sum_{t=0}^\cT\lambda_t^Tr_t.
\end{equation*}

\end{proof}

\begin{proof}[Proof of Theorem \ref{thm:multimarginal_scheme}.] 

A block-coordinate ascent method for the dual problem \eqref{eq:dual_multimarginal_omt} is to iteratively update one of the dual variables $\lambda_t$, for $t = 0,1,\ldots,\cT$, as to maximize \eqref{eq:dual_multimarginal_omt} while keeping the other dual variables fixed. The maximizing $\lambda_t$ may be found as the root of the corresponding gradient. Therefore, note that substituting the explicit expression for the elements in $\bU$, as detailed in Proposition~\ref{prp:multimarginal_dual}, into the first term of the dual objective \eqref{eq:dual_multimarginal_omt} yields
\begin{align*}
\sum_{i_0,\ldots, i_{\cT}} \bK_{i_0,\ldots, i_\cT} \bU_{i_0,\ldots, i_\cT}  =  \sum_{i_0,\ldots, i_{\cT}} \bK_{i_0,\ldots, i_\cT} \prod_{t=0}^\cT  \left(\exp(\disccovop_t^T \lambda_t / \epsilon) \right)_{i_t}.
\end{align*}
The gradient of this term with respect to $\lambda_t$ can be written as
\begin{equation*}
- \disccovop_t ( v_t \odot u_t ),
\end{equation*}
with $u_t=\exp(\disccovop_t^T \lambda_t/\epsilon)$ and the vector $v_t$ defined as
\begin{equation*}
\begin{aligned}
(v_t)_{i_t} = & \sum_{ \substack{i_0,\dots, i_{t-1},\\ i_{t+1},\dots,i_\cT} } K_{i_0,\ldots, i_{t-1}, i_t, i_{t+1},\ldots, i_\cT} \prod_{ \substack{s=1\\s\neq t}}^\cT (u_s)_{i_s} \\
= & \left( \sum_{ \substack{i_0,\dots, i_{t-1},\\ i_{t+1},\dots,i_\cT}} \bK_{i_0,\dots,i_{t-1},i_t,i_{t+1},\dots,i_\cT} \prod_{s=0}^\cT (u_s)_{i_s} \right)_{i_t} / (u_t)_{i_t} \\
= & \left( \projopdisc{t}{}(\bK \odot \bU) \right)_{i_t} / (u_t)_{i_t}.
\end{aligned}
\end{equation*}
Hence, the gradient of the dual objective function \eqref{eq:dual_multimarginal_omt} with respect to $\lambda_t$ is
\begin{equation*}
- \disccovop_t ( v_t \odot u_t ) - \lambda_t/(2\gamma_t) + r_t.
\end{equation*}

\end{proof}

\bibliographystyle{plain}
\bibliography{ElvanderHJK19_multi_arXiv}

\begin{thebibliography}{10}

\bibitem{AbrahamABC17_75}
I.~Abraham, R.~Abraham, M.~Bergounioux, and G.~Carlier.
\newblock Tomographic {R}econstruction from a {F}ew {V}iews: {A}
  {M}ulti-{M}arginal {O}ptimal {T}ransport {A}pproach.
\newblock {\em Applied Mathematics and Optimization}, 75(1):55--73, 2017.

\bibitem{AdlerROK17_arxiv}
J.~Adler, A.~Ringh, O.~{\"O}ktem, and J.~Karlsson.
\newblock Learning to solve inverse problems using {W}asserstein loss.
\newblock {\em arXiv preprint arXiv:1710.10898}, 2017.

\bibitem{AminiEG05_52}
A.~N. Amini, E.~Ebbini, and T.~T. Georgiou.
\newblock Noninvasive estimation of tissue temperature via high-resolution
  spectral analysis techniques.
\newblock {\em IEEE Trans. Bio. Eng.}, 52(2):221--228, Feb. 2005.

\bibitem{AngenentHT03_35}
S.~Angenent, S.~Haker, and A.~Tannenbaum.
\newblock Minimizing {F}lows for the {M}onge-{K}antorovich {P}roblem.
\newblock {\em SIAM J. Math. Anal.}, 35(1):61--97, 2003.

\bibitem{ArjovskyCB17_arXiv}
M.~Arjovsky, S.~Chintala, and L.~Bottou.
\newblock Wasserstein {G}{A}{N}.
\newblock {\em arXiv preprint arXiv:1701.07875}, 2017.

\bibitem{Audiotec}
Auditec.
\newblock {Auditory Tests (Revised), Compact Disc, Auditec, St. Louis}.
\newblock St. Louis, 1997.

\bibitem{BeiglbockHLP13_17}
M.~Beiglb\"ock, P.~Henry-Labord\`ere, and F.~Penker.
\newblock Model-independent bounds for option prices - a mass transport
  approach.
\newblock {\em Finance and Stochastics}, 17(3):477--501, 2013.

\bibitem{benamou2000computational}
J.~D. Benamou and Y.~Brenier.
\newblock A computational fluid mechanics solution to the {M}onge-{K}antorovich
  mass transfer problem.
\newblock {\em Numerische Mathematik}, 84(3):375--393, 2000.

\bibitem{benamou2015iterative}
J.~D. Benamou, G.~Carlier, M.~Cuturi, L.~Lenna, and G.~Peyr\'e.
\newblock Iterative {B}regman {P}rojections for {R}egularized {T}ransportation
  {P}roblems.
\newblock {\em SIAM Journal of Scientific Computing}, 37(2):1111--1138, 2015.

\bibitem{BonneelRPP15_51}
N~Bonneel, J.~Rabin, G.~Peyr\'e, and H.~Pfister.
\newblock Sliced and {R}adon {W}asserstein barycenters of measures.
\newblock {\em Journal of Math. Imaging Vision}, 51(1):22--45, 2015.

\bibitem{Brenier08_237}
Y.~Brenier.
\newblock Generalized solutions and hydrostatic approximation of the {E}uler
  equations.
\newblock {\em Physica D: Nonlinear Phenomena}, 237:1982--1988, 2008.

\bibitem{ByrnesGL00_48}
C.~L. Byrnes, T.~T. Georgiou, and A.~Lindquist.
\newblock A new approach to spectral estimation: a tunable high-resolution
  spectral estimator.
\newblock {\em IEEE Trans. Signal Process.}, 48(11):3189--3205, Nov. 2000.

\bibitem{Capon69}
J.~Capon.
\newblock {H}igh {R}esolution {F}requency {W}ave {N}umber {S}pectrum
  {A}nalysis.
\newblock {\em Proc.~IEEE}, 57:1408--1418, 1969.

\bibitem{ChenGP17_62}
Y.~Chen, T.~T. Georgiou, and M.~Pavon.
\newblock Optimal {T}ransport {O}ver a {L}inear {D}ynamical {S}ystem.
\newblock {\em IEEE Trans. Autom. Control}, 62(5):2137--2152, May 2017.

\bibitem{chen2016relation}
Y.~Chen, T.T. Georgiou, and M.~Pavon.
\newblock On the relation between optimal transport and {S}chr{\"o}dinger
  bridges: {A} stochastic control viewpoint.
\newblock {\em Journal of Optimization Theory and Applications},
  169(2):671--691, 2016.

\bibitem{ChenK18_2}
Y.~Chen and J.~Karlsson.
\newblock State {T}racking of {L}inear {E}nsembles via {O}ptimal {M}ass
  {T}ransport.
\newblock {\em IEEE Control Systems Letters}, 2(2):260--265, Apr 2018.

\bibitem{chizat2018scaling}
L.~Chizat, G.~Peyr\'e, B.~Schmitzer, and F.~X. Vialard.
\newblock Scaling algorithms for unbalanced optimal transport problems.
\newblock {\em Mathematics of Computation}, 87(314):2563--2609, 2018.

\bibitem{Cuturi13}
M.~Cuturi.
\newblock Sinkhorn distances: {L}ightspeed computation of optimal transport.
\newblock {\em Proc. Adv. Neural Inf. Process. Syst.}, pages 2292--2300, 2013.

\bibitem{DeSenaAMW15_23}
E.~{D}e Sena, N.~Antonello, M.~Moonen, and T.~van Waterschoot.
\newblock On the {M}odeling of {R}ectangular {G}eometries in {R}oom {A}coustic
  {S}imulations.
\newblock {\em IEEE/ACM Trans. Audio Speech Lang. Process.}, 23(4):774--786,
  Apr. 2015.

\bibitem{AguilaJ18_66}
P.~del Aguila~Pla and J.~Jald\'en.
\newblock Cell {D}etection by {F}unctional {I}nverse {D}iffusion and
  {N}on-negative {G}roup {S}parsity—{P}art {II}: {P}roximal {O}ptimization
  and {P}erformance {E}valuation.
\newblock {\em IEEE Trans. Signal Process.}, 66(20):5422--5437, 2018.

\bibitem{DominitzT10_16}
A.~Dominitz and A.~Tannenbaum.
\newblock Texture {M}apping via {O}ptimal {M}ass {T}ransport.
\newblock {\em IEEE Trans. Vis. Comput. Graphics}, 16(3):419--433, 2010.

\bibitem{ElvanderHJK18_eusipco}
F.~Elvander, I.~Haasler, A.~Jakobsson, and J.~Karlsson.
\newblock Tracking and {S}ensor {F}usion in {D}irection of {A}rrival
  {E}stimation {U}sing {O}ptimal {M}ass {T}ransport.
\newblock In {\em 26th European Signal Processing Conference}, pages
  1617--1621, Rome, Italy, Sep. 3-7 2018.

\bibitem{ElvanderHJK19_icassp}
F.~Elvander, I.~Haasler, A.~Jakobsson, and J.~Karlsson.
\newblock Non-{C}oherent {S}ensor {F}usion via {E}ntropy {R}egularized
  {O}ptimal {M}ass {T}ransport.
\newblock In {\em Proc. 44th IEEE Int. Conf. on Acoustics, Speech, and Signal
  Processing}, pages 4415--4419, Brighton, UK, May 13-17 2019.

\bibitem{ElvanderJK18_66}
F.~Elvander, A.~Jakobsson, and J.~Karlsson.
\newblock Interpolation and {E}xtrapolation of {T}oeplitz {M}atrices via
  {O}ptimal {M}ass {T}ransport.
\newblock {\em IEEE Trans. Signal. Process}, 66(20):5285 -- 5298, Oct. 2018.

\bibitem{FerrantePR08_53}
A.~Ferrante, M.~Pavon, and F.~Ramponi.
\newblock Hellinger vs. {K}ullback-{L}eibler multivariable spectrum
  approximation.
\newblock {\em IEEE Trans. Autom. Control}, 53(5):954--967, May 2008.

\bibitem{FranklinPA94}
G.~F. Franklin, J.~D. Powell, and A.~{Amani-Naeini}.
\newblock {\em {F}eedback {C}ontrol of {D}ynamic {S}ystems}.
\newblock Addison-Wesley, Reading, MA, 1994.

\bibitem{Georgiou05_50}
T.~T. Georgiou.
\newblock Solution of the general moment problem via a one-parameter imbedding.
\newblock {\em IEEE Trans. Autom. Control}, 50(6):811--826, 2005.

\bibitem{GeorgiouKT09_57}
T.~T. Georgiou, J.~Karlsson, and M.~S. Takyar.
\newblock Metrics for power spectra: an axiomatic approach.
\newblock {\em IEEE Trans. Signal Process.}, 57(3):859--867, Mar. 2009.

\bibitem{GrayBGM80_28}
R.~M. Gray, A.~Buzo, A..~H.~Gray Jr, and Y.~Matsuyama.
\newblock Distortion measures for speech processing.
\newblock {\em IEEE Trans. Acoust., Speech, Signal Process.}, 28(4):367--376,
  1980.

\bibitem{GrenanderS58}
U.~Grenander and G.~{Szeg\"{o}}.
\newblock {\em {T}oeplitz {F}orms and {T}heir {A}pplications}.
\newblock University of California Press, Los Angeles, 1958.

\bibitem{HakerZTA04_60}
S.~Haker, L.~Zhu, A.~Tannenbaum, and S.~Angenent.
\newblock Optimal {M}ass {T}ransport for {R}egistration and {W}arping.
\newblock {\em International Journal of Computer Vision}, 60(3):225--240, 2004.

\bibitem{HuL03_11}
Y.~Hu and P.~C. Loizou.
\newblock {A} {P}erceptually {M}otivated {A}pproach for {S}peech {E}nhancement.
\newblock {\em IEEE Transactions on Speech and Audio Processing},
  11(5):457--465, Sept. 2003.

\bibitem{JiangLG12_60}
X.~Jiang, Z.~Q. Luo, and T.~T. Georgiou.
\newblock {G}eometric {M}ethods for {S}pectral {A}nalysis.
\newblock {\em IEEE Trans. Signal Process.}, 60(3):1064--1074, Mar. 2012.

\bibitem{JohnsonD93}
D.~H. Johnson and D.~E. Dudgeon.
\newblock {\em {A}rray {S}ignal {P}rocessing: {C}oncepts and {T}echniques}.
\newblock Prentice Hall, Englewood Cliffs, N.J., 1993.

\bibitem{Kaplan06_42}
L.~M. Kaplan.
\newblock Global node selection for localization in a distributed sensor
  network.
\newblock {\em IEEE Trans. Aerosp. Electron. Syst.}, 42(1):113--135, Jan. 2006.

\bibitem{KarlssonR17_10}
J.~Karlsson and A.~Ringh.
\newblock Generalized {S}inkhorn iterations for regularizing inverse problems
  using optimal mass transport.
\newblock {\em SIAM Journal on Imaging Sciences}, 10(4):1935--1962, 2017.

\bibitem{KolouriPTSR17_34}
S.~Kolouri, S.~R. Park, M.~Thorpe, D.~Slepcev, and G.~K. Rohde.
\newblock Optimal {M}ass {T}ransport: {S}ignal processing and machine-learning
  applications.
\newblock {\em IEEE Signal Process. Mag.}, 34(4):43--59, July 2017.

\bibitem{KullbackL51_22}
S.~Kullback and R.~A. Leibler.
\newblock {O}n {I}nformation and {S}ufficiency.
\newblock {\em The Annals of Mathematical Statistics}, 22(1):79--86, March
  1951.

\bibitem{lellmann2014imaging}
J.~Lellmann, D.~A. Lorenz, C.~Sch\"onlieb, and T.~Valkonen.
\newblock Imaging with {K}antorovich--{R}ubinstein {D}iscrepancy.
\newblock {\em SIAM J. Imaging Sci.}, 7(4):2833--2859, 2014.

\bibitem{LiWHS02_19}
D.~Li, K.~D. Wong, Y.~H. Hu, and A.~M. Sayeed.
\newblock Detection, classification, and tracking of targets.
\newblock {\em IEEE Signal Process. Mag.}, 19(2):17--29, Mar. 2002.

\bibitem{LorenzF89_114_115}
J.~Lorenz and L.~Franklin.
\newblock On the scaling of multidimensional matrices.
\newblock {\em Linear Algebra and its Applications}, 114-115:717--715, 1989.

\bibitem{LuoYWH17_270}
J-A. Luo, K.~Yu, Z.~Wang, and Y-H. Hu.
\newblock Passive source localization from array covariance matrices via joint
  sparse representations.
\newblock {\em Neurocomputing}, 270:82--90, 2017.

\bibitem{leonard2013schrodinger}
C.~Léonard.
\newblock A survey of the {S}chrödinger problem and some of its connections
  with optimal transport.
\newblock {\em Discrete \& Continuous Dynamical Systems - A}, 34(4):1533--1574,
  2014.

\bibitem{mccann1997convexity}
R.~J. McCann.
\newblock A convexity principle for interacting gases.
\newblock {\em Advances in mathematics}, 128(1):153--179, 1997.

\bibitem{MierisovaA01}
S.~Mierisov\'{a} and M.~Ala-Korpela.
\newblock {MR} spectroscopy quantitation: a review of frequency domain methods.
\newblock {\em NMR in Biomedicine}, 14(4):247--259, 2001.

\bibitem{ning2015matrix}
L.~Ning, T.~T. Georgiou, and A.~Tannenbaum.
\newblock On {M}atrix-{V}alued {M}onge-{K}antorovich {O}ptimal {M}ass
  {T}ransport.
\newblock {\em IEEE Trans. Autom. Control}, 60(2):373--382, 2015.

\bibitem{NingJG13_20}
L.~Ning, X.~Jiang, and T.~T. Georgiou.
\newblock {O}n the {G}eometry of {C}ovariance {M}atrices.
\newblock {\em IEEE Signal Process. L}, 20(8):787--790, Aug. 2013.

\bibitem{pass2015multi}
B.~Pass.
\newblock Multi-marginal optimal transport: theory and applications.
\newblock {\em ESAIM: Mathematical Modelling and Numerical Analysis},
  49(6):1771--1790, 2015.

\bibitem{PeleW09}
O.~Pele and M.~Werman.
\newblock Fast and robust {E}arth {M}over's {D}istances.
\newblock In {\em IEEE 12th Int. Conf. on Comp. Vis.}, pages 460--467, Kyoto,
  Japan, Sep. 29 - Oct. 2 2009.

\bibitem{PesaventoGW02_50}
M.~Pesavento, A.~B. Gershman, and K.~M. Wong.
\newblock Direction finding in partly calibrated sensor arrays composed of
  multiple subarrays.
\newblock {\em IEEE Trans. Signal Process.}, 50(9):2103--2115, Sep. 2002.

\bibitem{peyre2019computational}
G.~Peyr{\'e} and M.~Cuturi.
\newblock Computational optimal transport.
\newblock {\em Foundations and Trends{\textregistered} in Machine Learning},
  11(5-6):355--607, 2019.

\bibitem{Renyi61}
A.~R\'enyi.
\newblock On measures of entropy and information.
\newblock In {\em Proc. 4th Berkeley Sym. Mathematics of Statistics and
  Probability}, pages 547--561, 1961.

\bibitem{RiekenF04_54}
D.~W. Rieken and D.~R. Fuhrmann.
\newblock {G}eneralizing {MUSIC} and {MVDR} for {M}ultiple {N}oncoherent
  {A}rrays.
\newblock {\em IEEE Trans. Signal Process.}, 52(9):2396--2406, Sep. 2004.

\bibitem{RoyPK86_34}
R.~Roy, A.~Paulraj, and T.~Kailath.
\newblock {ESPRIT} -- {A} {S}ubspace {R}otation {A}pproach to {E}stimation of
  {P}arameters of {C}isoids in {N}oise.
\newblock {\em IEEE Trans. Acoust., Speech, Signal Process.}, 34(4):1340--1342,
  October 1986.

\bibitem{Schmidt79}
R.~Schmidt.
\newblock {M}ultiple emitter location and signal parameter estimation.
\newblock In {\em Proceedings of RADC Spectrum Estimation Workshop}, pages
  243--258, 1979.

\bibitem{SchmitzHBNCCPS18_11}
M.~A. Schmitz, M.~Heitz, N.~Bonneel, F.~Ngol\`e, D.~Coeurjolly, M.~Cuturi,
  G.~Peyr\'e, and J-L. Starck.
\newblock Wasserstein {D}ictionary {L}earning: {O}ptimal {T}ransport-{B}ased
  {U}nsupervised {N}onlinear {D}ictionary {L}earning.
\newblock {\em SIAM J. Imaging Sci.}, 11(1):643--678, 2018.

\bibitem{SeeG04_52}
C.~M. See and A.~B. Gershman.
\newblock Direction-of-arrival estimation in partly calibrated subarray-based
  sensor arrays.
\newblock {\em IEEE Trans. Signal Process.}, 2004.

\bibitem{SimouF19_icassp}
E.~Simou and P.~Frossard.
\newblock {G}raph {S}ignal {R}epresentation with {W}asserstein {B}arycenters.
\newblock In {\em Proc. 44th IEEE Int. Conf. on Acoustics, Speech, and Signal
  Processing}, pages 5386--5390, Brighton, UK, May 13-17 2019.

\bibitem{Sinkhorn67_74}
R.~Sinkhorn.
\newblock Diagonal {E}quivalence to {M}atrices with {P}rescribed {R}ow and
  {C}olumn {S}ums.
\newblock {\em The American Mathematical Monthly}, 74(4):402--405, 1967.

\bibitem{StoicaBL11_59b}
P.~Stoica, P.~Babu, and J.~Li.
\newblock {SPICE} : a novel covariance-based sparse estimation method for array
  processing.
\newblock {\em IEEE Trans. Signal Process.}, 59(2):629 --638, Feb. 2011.

\bibitem{StoicaM05}
P.~Stoica and R.~Moses.
\newblock {\em {S}pectral {A}nalysis of {S}ignals}.
\newblock Prentice Hall, Upper Saddle River, N.J., 2005.

\bibitem{SuleimanPPZ18_66}
W.~Suleiman, P.~Parvazi, M.~Pesavento, and A.~M. Zoubir.
\newblock {N}on-{C}oherent {D}irection-of-{A}rrival {E}stimation {U}sing
  {Pa}rtly {C}alibrated {A}rrays.
\newblock {\em IEEE Trans. Signal Process.}, 66(21):5776--5788, Nov. 2018.

\bibitem{SwardAJ18_143}
J.~Sw{\"a}rd, S.~I. Adalbj\"ornsson, and A.~Jakobsson.
\newblock {G}eneralized {S}parse {C}ovariance-based {E}stimation.
\newblock {\em Elsevier Signal Processing}, 143:311--319, February 2018.

\bibitem{trees1992detection}
H.~L.~Van Trees.
\newblock {\em Detection, {E}stimation, and {M}odulation {T}heory:
  {R}adar-{S}onar {S}ignal {P}rocessing and {G}aussian {S}ignals in {N}oise}.
\newblock Krieger Publishing Co., Inc., 1992.

\bibitem{VervlietDSBL16_tensorlab}
N.~Vervliet, O.~Debals, L.~Sorber, M.~Van Barel, and L.~De Lathauwer.
\newblock Tensorlab 3.0, Mar. 2016.
\newblock Available online.

\bibitem{Villani08}
C.~Villani.
\newblock {\em Optimal transport: old and new}.
\newblock Springer Science \& Business Media, 2008.

\bibitem{yamamoto2018regularization}
K.~Yamamoto, Y.~Chen, L.~Ning, T.~T. Georgiou, and A.~Tannenbaum.
\newblock Regularization and {I}nterpolation of {P}ositive {M}atrices.
\newblock {\em IEEE Trans. Autom. Control}, 63(4):1208--1212, 2018.

\end{thebibliography}
\end{document}